\documentclass[12pt,authoryear,review]{elsarticle}

\makeatletter

\def\ps@pprintTitle{%
	\let\@oddhead\@empty
	\let\@evenhead\@empty
	\def\@oddfoot{\centerline{\thepage}}%
	\let\@evenfoot\@oddfoot}
\makeatother

\usepackage{amsmath,amssymb,amsfonts}
\usepackage{graphicx}
\usepackage{caption}
\usepackage{subcaption}
\usepackage{dsfont}
\usepackage{color}
\usepackage[rightcaption]{sidecap}
\usepackage{epstopdf}
\usepackage{todonotes}
\usepackage{hyperref}
\usepackage{bbm}
\usepackage[normalem]{ulem}
\usepackage{mathtools}
\usepackage{multicol}
\usepackage{soul}

\newtheorem{theorem}{Theorem}

\newtheorem{proposition}[theorem]{Proposition}

\newtheorem{definition}[theorem]{Definition}

\newenvironment{proof}[1][Proof]{\begin{trivlist}
		\item[\hskip \labelsep {\bfseries #1}]}{\end{trivlist}}

\journal{TBA}

\begin{document}
	\begin{frontmatter}
		
		\title{Dynamic Inventory Management with Mean-Field Competition
			\tnoteref{t1}}
		\tnotetext[t1]{The authors would like to thank Matt Loring and Ronnie Sircar for their helpful comments on an earlier draft.}
		
		\author[author1]{Ryan Donnelly}\ead{ryan.f.donnelly@kcl.ac.uk}
		\author[author1]{Zi Li}\ead{zi.2.li@kcl.ac.uk}
		\address[author1] {Department of Mathematics, King's College London, \\ Strand, London, WC2R 2LS, United Kingdom}
		
		\date{}
		
		\begin{abstract}
			Agents attempt to maximize expected profits earned by selling multiple units of a perishable product where their revenue streams are affected by the prices they quote as well as the distribution of other prices quoted in the market by other agents. We propose a model which captures this competitive effect and directly analyze the model in the mean-field limit as the number of agents is very large. We classify mean-field Nash equilibrium in terms of the solution to a Hamilton-Jacobi-Bellman equation and a consistency condition and use this to motivate an iterative numerical algorithm. Convergence of this numerical algorithm yields the pricing strategy of a mean-field Nash equilibrium. Properties of the equilibrium pricing strategies and overall market dynamics are then investigated, in particular how they depend on the strength of the competitive interaction and the ability to oversell the product.
		\end{abstract}
		
	\end{frontmatter}
	
	\section{Introduction}
	
	This paper considers an optimal price setting model in which agents attempt to liquidate a given product inventory over a finite time horizon. Agents quote prices to potential buyers which affects both the revenue made on individual sales and the intensity at which sales are made. In addition, the intensity of sales is affected by the distribution of prices quoted across all agents, meaning the revenue rate of an individual agent is impacted by the effects of competition. When the finite time horizon is reached, each agent that has unsold units of the perishable good recovers a salvage cost for their remaining inventory, essentially selling it with an imposed penalty. We employ a mean-field game approach to the model which lowers the complexity of numerical computation of optimal pricing strategies compared to the finite agent setting.
	
	Models of dynamic inventory pricing have been used in the past to dictate optimal pricing strategies for products which have a finite shelf life, notably fashion garments, seasonal leisure spaces, and airline tickets (see for example \cite{gallego1994optimal}, \cite{anjos2004maximizing}, \cite{anjos2005optimal}, and \cite{gallego2014dynamic}). Some early work in dynamic inventory pricing restricts the underlying dynamics to be deterministic (see \cite{jorgensen1986optimal}, \cite{dockner1988optimal}, and \cite{eliashberg1991competitive}) which may allow for more tractability and further analysis of optimal policies, with or without considering competition. Models with stochastic demand have also been studied (see \cite{gallego1994optimal} and \cite{zhao2000optimal}), but most results pertain to the case of a monopolistic agent without competition. The paper \cite{gallego2014dynamic} considers an oligopolistic market, but equilibrium with stochastic revenue streams is only classified in terms of a system of coupled differential equations. Instead of trying to analyse the solution to these equations, which is very computationally intense even for only two agents, the authors show that the stochastic game is well approximated by a deterministic differential game under suitable scaling of model parameters. The mean-field setting we consider can be thought of as allowing the number of agents to grow very large rather than the parameters which control the underlying dynamics. This retains a level of computational complexity at a level similar to the single agent case which allows us to investigate the effects of competition.
	
	Our model is a generalization of the single agent model considered in \cite{gallego1994optimal}, which we briefly summarize and refer to as the reference model when discussing the mean-field case. Specifically, agents hold a positive integer number of units of an asset and quote a selling price for each unit continuously through time. Sales occur at random times with an intensity that depends on the price quoted by the agent such that higher prices result in less frequent trades, creating a trade-off between large but infrequent revenue of quoting high prices, versus small but frequent revenue of low prices. Additionally, competition between agents is modelled by specifying that the sell intensity also depends on the distribution of prices quoted by all agents. Thus, an agent's selling rate increases if other agents begin to quote higher prices. Our model setting is quite similar to the one in \cite{yang2013nonatomic}, especially in regards to the dynamics of agents' inventory levels. The main difference is that our model allows prices to be continuous rather than being selected from a finite set of either high price or low price. Additionally, as a mean-field extension of one of the models presented in \cite{dockner1988optimal}, \cite{chenavaz2021dynamic} has a similar sell intensity to our work that depends on the distribution of prices quoted by all agents. Our work mainly differs from \cite{chenavaz2021dynamic} in that their model considers agents with states determined by a continuous quantity which changes deterministically, whereas in our work a representative agent has a discrete inventory level which changes stochastically.
	
	By working in a mean-field setting, our conditions for equilibrium do not require the solution to a coupled system of differential equations. Instead, equilibrium is classified by a single differential equation and a consistency condition. This lends itself to an efficient iterative algorithm which upon convergence yields a mean-field Nash equilibrium. We are unable to prove the existence of a mean-field Nash equilibrium, but we have conducted extensive numerical experiments with the iterative algorithm which always converge to the same result within a small numerical tolerance. The low dimension of the system of equations which classify equilibrium allows us to easily demonstrate the resulting pricing strategies, and hence investigate how prices under the effects of competition compare to those of the reference model.
	
	The tractability offered by a mean-field game framework over finite agent models has also led to their use in studying competition in other types of markets. In \cite{chan2017fracking} and \cite{ludkovski2017mean}, the effects of competition through mean-field interaction are incorporated into models of energy production and commodity extraction. In \cite{donnelly2019effort}, agents compete for a reward in an R\&D setting within a mean-field framework, in which earlier success yields greater rewards for the expended effort. In \cite{li2024mean} agents expend effort to mine cryptocurrency, where an agent's rate of mining depends on their hash rate as well as that of the entire population of miners. Our modeling framework has some similarities to these papers which allows us to employ a nearly direct adaptation of relevant numerical methods to compute equilibrium in our model.
	
	A novel focus of our work is regarding how overall market behaviour is affected by features describing individual agents. In particular, we investigate how the magnitude of competitive interaction, the ability to oversell the asset (with penalty), and price caps affect the total wealth transferred from consumers in the market, the average price paid per unit asset, and the probability that a particular consumer will end up empty handed due to overselling of the asset. The dependence of market behaviour on these phenomena could be used to guide regulatory framework with the goal of achieving desired levels of various measurements of economic welfare.
	
	The rest of the paper is organized as follows: in Section \ref{sec:reference_model} we give an overview of the reference model, which is the single-agent equivalent to the mean-field setting we cover in more detail. In Section \ref{sec:competition_model} we specify our model that incorporates the effects of competition, including the definition of equilibrium which we consider. In Section \ref{sec:NumericalExperiments} we show several examples of numerically solving for equilibrium and investigate the effects of competition and other market phenomena. Section \ref{sec:conclusion} concludes, and longer proofs are contained in the Appendix.
	
	\subsection{General Notation}
	
	Here we introduce the general notation which is used throughout the remainder of the work.
	
		\begin{tabular}{ l l }
			$(\Omega, \mathcal{F}, \{{\mathcal{F}_t}\}_{0\le t\le T}, \mathbb{P})$ & A filtered probability space. \\
			$\mathcal{X}_A$ & Indicator function of the event $A$. \\
			$\mathbb{E}_{t,s,x,q}[Y]$ & Conditional expectation of $Y$ given $S_t = s$, $X_{t^-}=x$ and $Q_{t^-}=q$. \\
			$\overline{Q}, \underline{Q}$ & Finite upper and lower bounds of inventory. \\
			$T$ & Finite terminal time horizon. \\
			$S=(S_t)_{t\in[0,T]}$ & Reference price process. \\
			$\mathcal{A}$ & The set of admissible controls.\\
			$\delta^i,\delta^f$ & Spread process of agent $i$, or induced by feedback control $f$. \\
			$\overline{\delta},\overline{\delta}^f$ & Mean spread posted by all agents, or induced by the feedback control $f$. \\
			$N^{\lambda(\delta)},N^{\lambda(\delta^i, \overline{\delta})}_i$ & Counting process of the number of items sold.\\
			$Q^\delta,Q^{\delta^i, \overline{\delta}}_i$ & Remaining inventory held by an agent. \\
			$X^\delta,X^{\delta^i, \overline{\delta}}_i$ & The cash generated by an agent selling inventory. \\
			$\lambda(\delta),\lambda(\delta, \overline{\delta})$ & The intensity function which determines the rate of inventory sales. \\
			$f(t,q)$ & Feedback form of a Markov strategy. \\
			$P^{f,M}_q,P^f_q$ & Proportion process in the $M$-player setting and in the mean-field limit. \\
		\end{tabular}
	
	\section{Single Agent Reference Model}\label{sec:reference_model}
	
	In this section, we introduce a reference model where we only consider one agent. Many aspects of the dynamics we consider are equivalent to those found in \cite{gallego1994optimal} with some modifications made to the price process and agent's performance criterion. This style of model which relates intensity to price has also been used frequently in the literature on algorithmic trading. See for example \cite{gueant2012optimal}, \cite{gueant2015general}, and \cite{cartea2015optimal}.
	
	We work on a probability space $(\Omega,\mathcal{F},\mathbb{P})$ which we assume supports all random variables and stochastic processes defined below.  We consider an agent who has to liquidate a finite quantity $\overline{Q}$ of a given product within a finite time horizon of length $T$. The reference price process of the product is denoted by $S = (S_t)_{t\in[0,T]}$ with dynamics
	\begin{align}
		\label{eqn:ReferenceProcess}
		dS_t = \sigma\,dW_t\,,
	\end{align}
	where $\sigma>0$ is a constant and $W = (W_t)_{t\in[0,T]}$ is a Brownian motion.
	
	We denote the agent's spread process above the reference price by $\delta=(\delta_t)_{t\in[0,T]}$, so that she continuously quotes her selling price at $S_t + \delta_{t^-}$\footnote{Here by $t^-$, we mean the left limit to time $t$.} at every time point, and is committed to sell one item\footnote{Note that selling one item may be understood as selling a block of units of the product, each block being of the same size.} with the quoted price. Once the agent clears her inventory, she will stop trading. Given a non-negative bounded process $\lambda=(\lambda_t)_{t\in[0,T]}$, her number of items sold follows a counting process denoted by $N^{\lambda}=(N_t^{\lambda})_{t\in[0,T]}$ defined as follows: let $\{u_n\}_{n=1}^\infty$ be a sequence of independent standard uniform random variables. Define
	\begin{align*}
		\tau_0 &= 0\,,\\
		\tau_n &= \inf\{t\geq\tau_{n-1}: e^{-\int_{\tau_{n-1}}^t\lambda_u\,du} \leq u_n\}\,,\\
		N^\lambda_t &= \sup\{n\geq0: t\geq \tau_n\}\,.
	\end{align*}
	Then $N^\lambda$ is a doubly stochastic Poisson process with intensity process $\lambda$ (see \cite{lando1998cox}), and the sequence of times at which $N^\lambda$ jumps is $\{\tau_n\}_{n=1}^\infty$. Subsequently, we will let the intensity process depend on her spread through the relation $\lambda_t = \lambda(\delta_t)$ for a function to be specified later, and we will write $N^{\lambda(\delta)}$ to denote the number of items sold when the agent quotes a spread according to $\delta=(\delta_t)_{t\in[0,T]}$. We denote the indicator function of any event $A$ by $\mathcal{X}_A$. Thus, the agent's inventory $Q^{\delta} = (Q_{t}^{\delta})_{t\in[0,T]}$ satisfies
	\begin{align*}
		dQ_{t}^{\delta} =-\mathcal{X}_{Q_{t^-}^{\delta} > 0}\, dN_{t}^{\lambda(\delta)}\,,
	\end{align*}
	with initial value $Q_0^\delta = \overline{Q}$. As a consequence of her trades, the agent accumulates cash denoted by $X^{\delta} = (X_{t}^{\delta})_{t\in[0,T]}$ with dynamics given by
	\begin{align*}
		dX^{\delta}_{t} = \mathcal{X}_{Q_{t^-}^{\delta} > 0}\,\left(S_t + \delta_{t^-}\right)dN_{t}^{\lambda(\delta)},
	\end{align*}
	with given initial cash $X_{0}^{\delta}=x_0$.
	
	The agent's goal is to maximize her expected P\&L at time $T$ with an inventory penalty. Specifically, her value functional is
	\begin{align*}
		J(\delta) = \mathbb{E} \left[ X_{T}^{\delta}+Q_{T}^{\delta}\,\left(S_T-\alpha\,Q_{T}^{\delta}\right)-\phi\int^T_0 \left(Q_{u}^{\delta}\right)^2du\right]\,,
	\end{align*}
	where $\alpha$ and $\phi$ are positive constants. The objective functional consists of three parts. The first term in the expectation $X_{T}^{\delta}$ is the amount of cash at time $T$. The second term $Q_{T}^{\delta}\,\left(S_T-\alpha\, Q_{T}^{\delta}\right)$ corresponds to the salvage value of unsold items remaining at time $T$, and the parameter $\alpha$ represents a penalty for failing to sell all inventory by the end of the trading period. The running inventory penalty term $\phi\int^T_0 (Q_{u}^{\delta})^2du$ penalizes deviation from zero during the entire trading horizon, so it can be interpreted as an urgency penalty, and the parameter $\phi$ acts as a risk control. Thus, the agent's dynamic value function is
	\begin{align}
		\label{eqn:ValueFunctionSingle}
		H(t,s,x,q) = \sup_{(\delta_u)_{t\le u\le T}\in\mathcal{A}}\mathbb{E}_{t,s,x,q} \left[ X_{T}^{\delta}+Q_{T}^{\delta}\,\left(S_T-\alpha\,Q_{T}^{\delta}\right)-\phi\int^T_{t} \left(Q_{u}^{\delta}\right)^2du\right]\,,
	\end{align}
	where the set of admissible controls $\mathcal{A}$ consists of Markov feedback controls of the form $\delta_t = f(t,Q_t^\delta)$ which are bounded below by a large negative constant $\underline{B}$.\footnote{It is well known that restricting to Markov strategies generally does not reduce performance. See for example \cite{oksendal2007applied}.} That is, we have
	\begin{align*}
		\mathcal{A} = \left\{ \delta\, \Big\vert\, \delta_t = f(t,Q^\delta_t)\,, \quad \underline{B}\le f \right\}\,,
	\end{align*}
	The expectation $\mathbb{E}_{t,s,x,q}$ is conditional on $S_t = s$, $X_{t^-}=x$, and $Q_{t^-}=q$.
	
	To solve the optimal control problem described above we consider the associated Hamilton-Jacobi-Bellman (HJB) equation along with terminal conditions given by (see for example \cite{gueant2012optimal})
	\begin{align}
		\label{eqn:HJBOriginalSingle}
		\begin{split}
			&\partial_t H + \frac{1}{2}\,\sigma^2\,\partial_{ss}H-\phi\, q^2+\sup_{\delta\ge \underline{B}}\lambda(\delta) \left[ H(t,s,x+s+\delta,q-1)-H(t,s,x,q)\right]\,\mathcal{X}_{q>0} = 0\,,\\
			&H(T,s,x,q) = x + q\,(s-\alpha\,q)\,.
		\end{split}
	\end{align}
	To solve the HJB equation, we use the excess value ansatz for $H$ given by
	\begin{align}
		\label{eqn:ExecessValueAnsatz}
		H(t,s,x,q) = x+q\,s+h_q(t)\,.
	\end{align}
	The first term $x$ is the current cash in hand, the second term $q\,s$ accounts for the reference value of the current inventory, and the last term $h_q(t)$ represents the expected profit or loss from liquidating $q$ items with both the terminal penalty and the running inventory penalty. Substituting the ansatz into HJB equation gives a system of ODEs with terminal condition given by
	\begin{align}
		\label{eqn:HJBSingle}
		\begin{split}
			&\partial_t h_q -\phi\,q^2+ \sup_{\delta\ge \underline{B}} \lambda(\delta)\left[\delta+h_{q-1}(t)-h_q(t) \right]\,\mathcal{X}_{q>0} = 0\,,\\
			&h_q(T) = -\alpha\,q^2\,.
		\end{split}
	\end{align}
	To solve for the optimal feedback controls in terms of $h_q(t)$, we assume that the intensity function $\lambda$ follows 
	\begin{align}
		\label{eqn:PoissonIntensityAssumptionSingle}
		\lambda(\delta) = A\exp{\{-\kappa\,\delta\}}\,
	\end{align}
	where $A,\kappa \ge 0$. Here $\kappa$ is the sensitivity of the demand rate of the product with respect to the agent's spread\footnote{Other forms of the intensity function can be used, but a particular property which is desired is that there is a unique maximizer of $\delta\,\lambda(\delta)$. We choose this exponential form for tractability reasons.}. This means that the smaller the agent's spread is, the faster the item will be sold.
	
	We delay a proof of existence and uniqueness of a solution to equation \eqref{eqn:HJBSingle} until we have developed the control problem that incorporates multiple agents. The solution to equation \eqref{eqn:HJBSingle} is a special case of the solution to equation \eqref{eqn:HJB} which appears in Proposition \ref{prop:SolutiontoHJBEquation}. A verification that solutions to the HJB equation \eqref{eqn:HJBOriginalSingle} yield the value function defined in equation \eqref{eqn:ValueFunctionSingle} is also postponed to the more general setting when we include competition (see Theorem \ref{theo:VerificationTheorem}).
	
	\begin{proposition}[Optimal Feedback Control]
		\label{prop:SolutiontoHJBEquationSingle}
		The optimal feedback controls of the HJB equation are given by
		\begin{align}
			\label{eqn:OptimalFeedbakControlSingle}
			\delta^*(t,q) = \max\left\{  \frac{1}{\kappa}+h_q(t)-h_{q-1}(t), \underline{B}\right\}\,, \quad q\neq0\,.
		\end{align}
	\end{proposition}
	\begin{proof}
		Substituting the intensity function from equation \eqref{eqn:PoissonIntensityAssumptionSingle} into equation \eqref{eqn:HJBSingle}, we have
		\begin{align*}
			\partial_t h_q -\phi\,q^2+ \sup_{\delta\ge \underline{B}} A\,\exp{\{-\kappa\,\delta\}}\left[\delta+h_{q-1}(t)-h_q(t) \right]\,\mathcal{X}_{q>0} = 0\,.
		\end{align*}
		Applying first order conditions to the $\sup$ term and letting the result equal $0$ gives us
		\begin{align*}
			\Tilde{\delta}^{*}(t,q) =   \frac{1}{\kappa}+h_q(t)-h_{q-1}(t)\,, \quad q\neq0\,.
		\end{align*}
		It is an easy task to check that the first derivative of the sup term is positive for $\delta<\Tilde{\delta}^{*}$ and negative for $\delta>\Tilde{\delta}^{*}$. Therefore, when the lower bound $\underline{B}\le \Tilde{\delta}^{*}$, the stationary point $\Tilde{\delta}^{*}$ is the maximizer; when $\underline{B}>\Tilde{\delta}^{*}$, the lower bound itself is the maximizer. Thus, the optimal feedback control is given by \eqref{eqn:OptimalFeedbakControlSingle}.\qed
	\end{proof}
	
	To intuitively understand the optimal feedback control form of the optimal spread, note that the first term $\frac{1}{\kappa}$ maximizes the rate of expected incoming profit, without any consideration for risk or terminal penalties. The quantity $h_{q-1}(t) - h_q(t)$ is the change of the expected future risk-adjusted P\&L due to selling one item. Thus, the negative of this value, $h_{q}(t) - h_{q-1}(t)$, represents the amount that the agent is willing to adjust her spread, positively or negatively, depending on the change of future risks and potential future profits by having one fewer item to sell.
	
	\begin{figure}[!htp]
		\centering
		\includegraphics[width=0.7\textwidth]{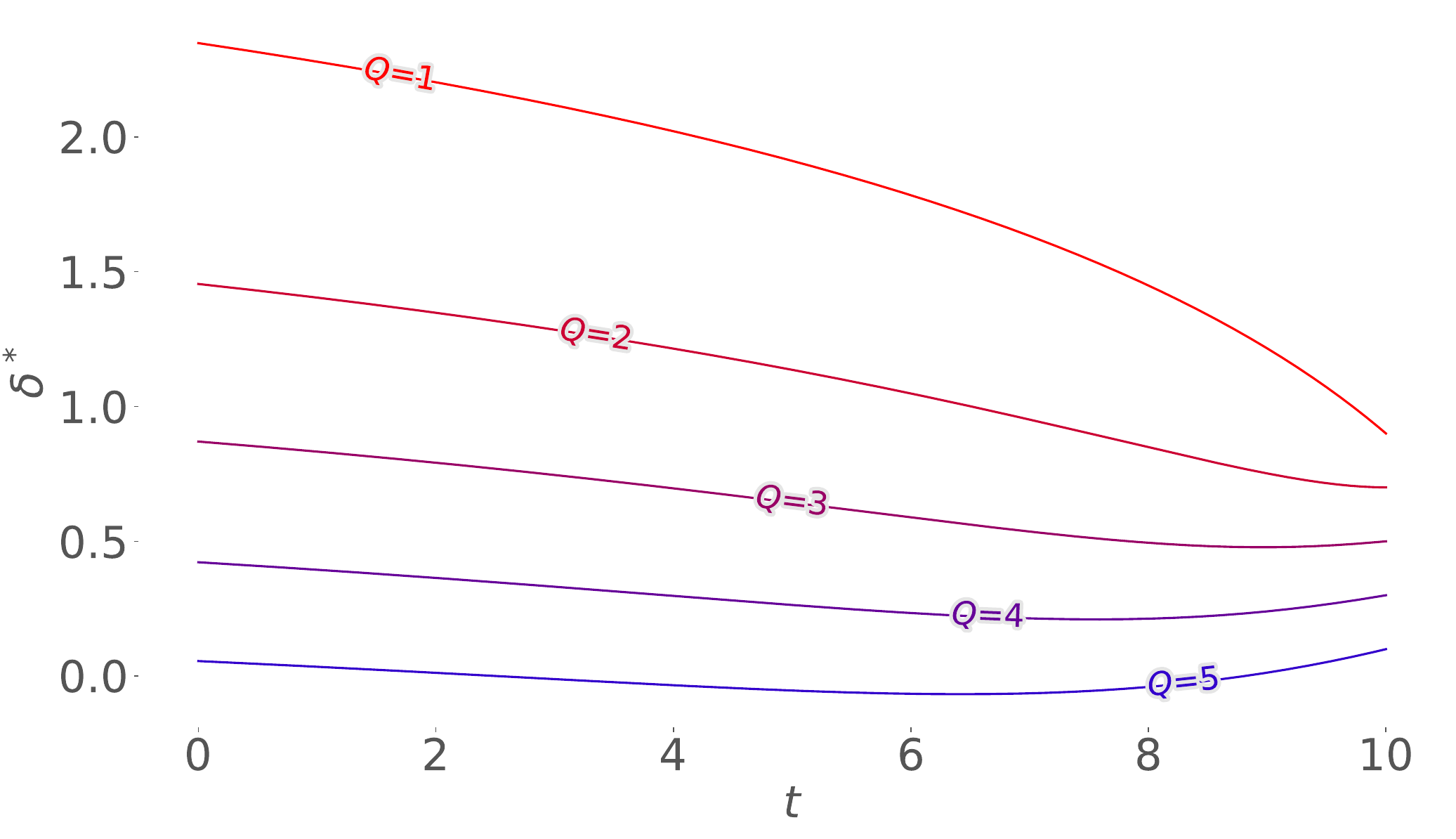}
		\caption{Optimal pricing strategy in reference model as a function of $t$ for various $Q$. Other parameters are $T=10$, $\overline{Q}=5$, $\alpha=0.1$, $\kappa=1$, $\phi=0.03$, $A=1$, and $\underline{B}=-10$.}
		\label{fig:MultiQuoteReference}
	\end{figure}
	
	Figure \ref{fig:MultiQuoteReference} shows the optimal spreads when equation $\eqref{eqn:HJBSingle}$ is numerically solved after substituting the optimal feedback controls in Proposition \ref{prop:SolutiontoHJBEquationSingle}. We observe that the optimal spreads are decreasing in inventory level monotonically. This is sensible as the agents with more inventory have more urgency to sell the products and avoid the terminal penalty
	and uncertainty of the underlying price. In addition, for lower inventory levels, the optimal spreads are decreasing in time. However, for higher inventory levels, the agent would choose a smaller spread in the early time interval to sell fast, and then increase her spread near the terminal time $T$. This is because at the final time, the agent is more willing to benefit from a trading opportunity through proposing a higher quote and receiving a lower probability to sell. Further technical properties in the single agent setting, such as the growth rate of the optimal spread, can be found in \cite{gueant2012optimal}.
	
	We also see that optimal spread can become negative at some times for particular inventory levels. This happens when there is a need to liquidate products at a very fast speed to avoid a terminal penalty or underlying price risk. Depending on the model parameters, this could result in a negative quoted price. One method of resolving this situation would be to make the lower bound of admissible strategies be time-dependent based on the reference price, such as $\delta_t \in [-S_t,\infty)$, but we do not pursue this further.
	
	\section{Model with Competition}\label{sec:competition_model}
	
	In this section we introduce a model with multiple agents, specify the dynamics of each agent's inventory and wealth processes based on the prices they quote, and specify the optimization problem that each agent attempts to solve. Additionally, and distinct from the previous section with only one agent, we give the dynamics of the distribution of inventory across all agents in the mean-field limit, and finally we define the corresponding notion of equilibrium. This model is similar to that of \cite{gallego2014dynamic}, and while we initially consider a finite number of agents we will only search for equilibrium in the mean-field limit of an infinite number of agents because this lowers the complexity of numerical computation, thus allowing further analysis of the solution.
	
	We consider a finite collection of agents indexed by $i\in\{1,\dots,M\}$ aiming to liquidate shares of a given product. The reference price process is the same as in equation \eqref{eqn:ReferenceProcess}.
	
	\subsection{Inventory, Wealth, and Performance Criterion}
	
	Each agent chooses a spread process $\delta^i = (\delta^i_t)_{t\in[0,T]}$, but only actively posts a price when they hold non-zero inventory. Once their inventory reaches zero, their inventory state remains constant until time $T$, effectively leaving the market. The mean spread chosen by all agents in the market is denoted by $\overline{\delta} = (\overline{\delta}_t)_{t\in[0,T]}$. Each agent has an associated inventory and wealth process denoted by $Q^{\delta^i,\overline{\delta}}_{i} = (Q^{\delta^i,\overline{\delta}}_{i,t})_{t\in[0,T]}$ and $X^{\delta^i,\overline{\delta}}_{i} = (X^{\delta^i,\overline{\delta}}_{i,t})_{t\in[0,T]}$, respectively, and we assume all agents begin at time $t=0$ with the same maximum inventory $Q^{\delta^i,\overline{\delta}}_{i,0} = \overline{Q}$. The superscripts of $\delta^i$ and $\overline{\delta}$ are to indicate that the inventory and wealth processes of a single agent depend not only on their own strategy, but also on the mean strategy across all agents\footnote{We assume that dynamics depend on the distribution of spreads only through the mean spread, but other dependencies on the distribution of spreads across agents could also be implemented.}. Keeping in mind that spreads are only posted by agents with non-zero inventory, the mean spread at time $t$ is equal to
	\begin{align}\label{eqn:delta_bar_M}
		\overline{\delta}_t &= \frac{\sum_{i=1}^M \delta^i_t\,\mathcal{X}_{Q_{i, t}^{\delta^i,\overline{\delta}}>0}}{\sum_{i=1}^M \mathcal{X}_{Q_{i,t}^{\delta^i,\overline{\delta}}>0}}\,.
	\end{align}
	The inventory and wealth of agent $i$ will change based on the arrivals of a counting process denoted by $N^{\lambda^i}_{i} = (N^{\lambda^i}_{i,t})_{t\in[0,T]}$ which is defined similarly to the single agent reference model. Given a collection of non-negative bounded processes $\lambda^i = (\lambda^i_t)_{0\leq t \leq T}$ for each $i$, let $\{u^i_n\}_{i,n=1}^\infty$ be a collection of independent standard uniforms. Define
	\begin{align*}
		\tau^i_0 &= 0\,,\\
		\tau^i_n &= \inf\{t\geq\tau_{n-1}^i: e^{-\int_{\tau^i_{n-1}}^t\lambda^i_u\,du} \leq u^i_n\}\,,\\
		N^{\lambda^i}_{i,t} &= \sup\{n\geq0: t\geq \tau^i_n\}\,.
	\end{align*}
	Then each $N^{\lambda^i}_i$ is a doubly stochastic Poisson process with intensity process $\lambda^i$. Subsequently, we let the intensity process of agent $i$ depend on her own spread and the mean spread of all agents through the relation $\lambda^i_t = \lambda(\delta_t^i,\overline{\delta}_t)$ for some function $\lambda$ to be specified later, and we write $N^{\lambda(\delta^i,\bar{\delta})}_{i}$ to denote the number of items sold by agent $i$.
	
	Thus, the inventory of agent $i$ changes according to the dynamics
	\begin{align}
		\label{eqn:InventoryProcess}
		dQ_{i,t}^{\delta^i,\overline{\delta}} = -\mathcal{X}_{Q_{i,t^-}^{\delta^i,\overline{\delta}} > 0}\,dN^{\lambda(\delta^i,\overline{\delta})}_{i,t}\,.
	\end{align}
	When a transaction is executed with agent $i$ at time $t$ it occurs at a price of $S_t+\delta^i_{t^-}$ so that their wealth changes according to
	\begin{align}
		\label{eqn:CashProcess}
		dX_{i,t}^{\delta^i,\overline{\delta}} = \mathcal{X}_{Q_{i,t^-}^{\delta^i,\overline{\delta}} > 0}\,(S_t + \delta^i_{t^-})\,dN^{\lambda(\delta^i,\overline{\delta})}_{i,t}\,.
	\end{align}
	with given initial cash $X_{i,0}^{\delta^i,\overline{\delta}} = x_0^i$. For a single agent with index $i$, given fixed strategies for each other agent, their performance criterion is given by
	\begin{align}
		\label{eqn:FiniteObjectiveFunctional}
		J^M_i(\delta^i;\delta^{-i}) = \mathbb{E}\biggl[X_{i,T}^{\delta^i,\overline{\delta}} + Q_{i,T}^{\delta^i,\overline{\delta}}\,(S_{T} - \alpha\,Q_{i,T}^{\delta^i,\overline{\delta}}) - \phi\int_0^T (Q_{i,t}^{\delta^i,\overline{\delta}})^2\,dt\biggr]\,,
	\end{align}
	where $\delta^{-i}\coloneqq (\delta^1, \cdots, \delta^{i-1}, \delta^{i+1}, \cdots, \delta^M)$ is the collection of spreads excluding agent $i$.
	
	\subsection{Mean-Field Population Dynamics}
	
	We now proceed to the mean-field limit $M\rightarrow\infty$. Because the agents are homogeneous with respect to their inventory dynamics, wealth dynamics, and performance criteria, it is expected in equilibrium that they will employ Markov strategies with the same feedback form.\footnote{As mentioned previously, the restriction to Markov strategies generally does not reduce performance. See \cite{oksendal2007applied}.} Thus, we may consider a representative agent for all subsequent computations which allows us to suppress the index $i$ from all quantities. Suppose all agents choose their spread process $\delta^f=(\delta^f_t)_{t\in[0,T]}$ according to the function $f$, continuous in its first argument, such that
	\begin{align}
		\delta^{f}_t &= f(t,Q_t^{\delta^{f},\overline{\delta}^f})\,,\label{eqn:delta_f}
	\end{align}
	where $\overline{\delta}^f=(\overline{\delta}^f_t)_{t\in[0,T]}$ represents the mean spread process posted by all agents.\footnote{This form of feedback control is inspired by the optimizer in the single agent model. We have been unable to prove in the mean-field setting that equilibria do not depend on the common noise $S$, but the classification we give below results in a bona fide mean-field Nash equilibrium.} An important quantity to track is the proportion of all agents that hold any particular value of inventory. For the inventory level $q\in\{0,1,\dots,\overline{Q}\}$, denote this proportion by $P_q^{f} = (P_{q,t}^{f})_{t\in[0,T]}$ which is equal to
	\begin{align*}
		P_{q,t}^{f} &= \lim_{M\rightarrow\infty} \frac{1}{M}\sum_{i=1}^M \mathcal{X}_{Q_{i, t}^{\delta^f,\overline{\delta}^f}=q}\,.
	\end{align*}
	With all agents using the same Markov feedback strategy, the collection of inventory processes is an exchangeable Markov mean-field particle system, and the limit above exists by the propagation of chaos and convergence of the empirical measure of the inventory processes (see \cite{chaintron2022propagation} for more details). Additionally, the empirical measure converges to the law of the representative agent's inventory process $Q^{\delta^{f},\overline{\delta}^f}$, giving the relation
	\begin{align*}
		\mathbb{P}(Q^{\delta^{f},\overline{\delta}^f}_t = q) &= P^f_{q,t}\,.
	\end{align*}

		The mean spread $\overline{\delta}^f$ can then be written as a weighted average
		\begin{align}
			\overline{\delta}^f_t &= \frac{\sum_{q=1}^{\overline{Q}} f(t,q)\, P^f_{q, t}}{1 -  P^f_{0, t}}\,.\label{eqn:delta_bar_f}
		\end{align}
		Further, any single agent with inventory $q$ leaves that state with intensity $\lambda(f(t,q), \overline{\delta}_t^f)$, which means there is a flow of agents from state $q$ to state $q-1$ at a rate of $\lambda(f(t,q), \overline{\delta}_t^f)\,  P^f_{q,t}$. Thus, the dynamics of each $P^f_q$ can be written (see \cite{yang2013nonatomic} and \cite{sun2006exact}) as
		\begin{subequations}\label{eqn:dP}
			\begin{align}
				&dP^f_{\overline{Q},t} = -P^f_{\overline{Q},t}\,\lambda\left(f(t,\overline{Q}), \overline{\delta}_t^f\right)\, dt\,,\\
				&dP^f_{q,t} = P^f_{q+1,t}\,\lambda\left(f(t,q+1), \overline{\delta}_t^f\right)\,dt -P^f_{q,t}\,\lambda\left(f(t,q), \overline{\delta}_t^f\right)\,dt\,,\qquad q\neq 0, \,\overline{Q}\,,\\
				&dP^f_{0,t} = P^f_{1,t}\,\lambda\left(f(t,1), \overline{\delta}_t^f\right)\,dt\,,\\
				\nonumber&\hspace{0mm}\text{and because each agent begins with maximum inventory $\overline{Q}$ the initial conditions are given by}\\
				&\hspace{40mm}P^f_{\overline{Q},0} = 1\,, \quad P^f_{q,0} = 0\,, \qquad q\neq \overline{Q}\,.
			\end{align}
		\end{subequations}
		
		 Inspection of \eqref{eqn:delta_bar_f} and \eqref{eqn:dP} reveals that for a fixed function $f$, the resulting $\overline{\delta}^f$ is deterministic and continuous on $[0,T]$ and therefore also bounded.
		
		\subsection{Mean-Field Performance Criterion and Equilibrium}
		
		Each agent still wishes to maximize their expected terminal wealth subject to all other agents fixing their strategy. Inspection of equations \eqref{eqn:delta_bar_f} and \eqref{eqn:dP} show that for a fixed $f$, the process $\overline{\delta}^f$ is deterministic, and in this case the dynamics of a representative agent's wealth and inventory do not have explicit dependence on $P^f$, but only on $\overline{\delta}^f$. Thus, when specifying the performance criterion of a representative agent, it also will not depend on the entire flow of the population distribution $P$, but only on the mean spread $\overline{\delta}$. For a given deterministic mean spread, we write the performance criterion as
		\begin{align*}
			J(\delta^f; \overline{\delta}) = \mathbb{E}\biggl[X_T^{\delta^f,\overline{\delta}} + Q_T^{\delta^f,\overline{\delta}}\,(S_T - \alpha\,Q_T^{\delta^f,\overline{\delta}}) - \phi\int_0^T (Q_t^{\delta^f,\overline{\delta}})^2\,dt\biggr]\,.
		\end{align*}
		
		\begin{definition}[Markov Mean-Field Nash Equilibrium]\label{def:mean_field_equilibrium}
			A Markov mean-field Nash equilibrium is a pair of functions $f$ and $P$ such that
			\begin{enumerate}[i)]
				\item $P = P^f$, where $P^f$ satisfies \eqref{eqn:dP} with $\overline{\delta}^f$ given by \eqref{eqn:delta_bar_f},
				\item $J(\delta^{f}; \overline{\delta}^f) \ge J(\delta^{g}; \overline{\delta}^f)$, for all functions $g\neq f$, with $\overline{\delta}^f$ given by \eqref{eqn:delta_bar_f}.
			\end{enumerate}
		\end{definition}
		The interpretation of the equilibrium given in Definition \ref{def:mean_field_equilibrium} is that we seek a function $f$ (along with its induced mean-field population distribution, $P$) which determines the strategy of all agents through the relation \eqref{eqn:delta_f} such that no particular agent can increase their performance by deviating from that strategy. The analytical tractability of this problem over the finite-agent case is due to the fact that in the mean-field setting, no single agent can affect the dynamics of the whole population, as determined by equation \eqref{eqn:dP} by changing their own strategy. Finding the mean-field Nash equilibrium thus comes down to solving a single optimization problem along with checking a consistency condition. Specifically, we take a mean spread process $\overline{\delta}$ as given and solve for the optimal strategy for a representative agent in feedback form. Then we check for the consistency relation that if every agent uses the computed strategy then the mean spread process, as determined by equations \eqref{eqn:delta_bar_f} and \eqref{eqn:dP} is given by $\overline{\delta}$, the same process initially given. If this holds then the feedback control and the corresponding population distribution is an equilibrium.
		
		\subsection{Mean-Field Optimization Problem}
		
		The intensity function $\lambda(\delta,\overline{\delta})$ for the representative agent states that her sales depend on both her own spread and the mean spread of all agents. We assume that $\lambda$ is twice differentiable in both inputs. According to \cite{dockner1988optimal}, the intensity $\lambda$ should satisfy the following assumptions:
		\begin{subequations}
			\label{eqn:IntensityConditions}
			\begin{align}
				\label{eqn:DerivativeAgainstOwnQuote}
				\frac{\partial\lambda}{\partial\delta} &< 0\,,\\
				\label{eqn:DerivativeAgainstOtherQuote}
				\frac{\partial\lambda}{\partial\overline{\delta}} &> 0\,,\\
				\label{eqn:DerivativeSum}
				\frac{\partial\lambda}{\partial\delta}+\frac{\partial\lambda}{\partial\overline{\delta}} &<0\,,\\
				\label{eqn:DerivativeProduct}
				\frac{\partial^2\lambda}{\partial\delta\partial\overline{\delta}} &\le 0\,,\\
				\label{eqn:TechniqueOne}
				\lambda\frac{\partial^2\lambda}{\partial\delta^2} &<2 \left(\frac{\partial\lambda}{\partial\delta}\right)^2\,.
			\end{align}
		\end{subequations}
		
		Conditions \eqref{eqn:DerivativeAgainstOwnQuote}-\eqref{eqn:DerivativeProduct} are standard assumptions in competition theory. Condition \eqref{eqn:DerivativeAgainstOwnQuote} is equivalent to assuming a downward-sloping demand curve, that is, with a fixed mean spread a lower spread quoted from the representative agent leads to higher demand. Condition \eqref{eqn:DerivativeAgainstOtherQuote} states that an increase in the mean spread causes the sales of the representative agent to rise. Condition \eqref{eqn:DerivativeSum} implies that if all agents raise their spreads by the same amount, their sales will decrease simultaneously. Condition \eqref{eqn:DerivativeProduct} implies that the higher the mean spread, the easier it is for the representative agent to increase her probability of selling by reducing her own spread. Conversely, if the mean spread is higher, the representative agent will lose market share more quickly when she raises her spread. Condition \eqref{eqn:TechniqueOne} is a technical condition, originating from the strict concavity of the Hamiltonian with respect to spread (i.e., the control variable).
		
		We assume that the instantaneous intensity $\lambda$ is of the following form  similar to \cite[(4.2)]{chenavaz2021dynamic}
		\begin{align}
			\begin{split} \label{eqn:PoissonIntensityAssumption} \lambda(\delta, \overline{\delta}) &= A\,\exp\left\{-\kappa\, \delta+\beta\,(\overline{\delta}-\delta)\right\}\\
				&=A\,\exp\left\{-(\kappa+\beta)\,\delta+\beta\,\overline{\delta}\right\},
			\end{split}
		\end{align}
		for constants $A,\kappa,\beta > 0$. This assumption satisfies all the conditions stated in equation \eqref{eqn:IntensityConditions}. Compared to the single-agent intensity function defined in equation \eqref{eqn:PoissonIntensityAssumptionSingle}, there is an extra term in the exponent that represents the competitiveness in the market whose strength is characterized by $\beta$. Under this assumption, there are two drivers for faster execution: lower quoted price of one's own, and higher mean quoted price of other market participants. This is consistent with the expected sales loss to competitors when one's price becomes higher relative to the market price for the same product.
		
		In equation \eqref{eqn:DerivativeSum}, we introduce that total sales of the market decreases when all agents raise their spreads by the same amount. Now we investigate the change of total sales of the market in another scenario where a representative agent, denoted by index $i$, fixes her spread $\delta^i$, and other agents increase their spreads by identical amounts. First we consider the finite $M$-player game and then take $M$ to infinity. According to equation \eqref{eqn:PoissonIntensityAssumption}, the intensity $\lambda^i$ for the $i$-th agent is
		\begin{align*}
			\lambda^i&=A\exp{\left\{-\left(\kappa+\beta\right)\delta^i+\frac{\beta}{M}\sum_{j=1}^M\delta^j\right\}}\\
			&=A\exp{\left\{-\left(\kappa+\frac{\beta(M-1)}{M}\right)\delta^i+\frac{\beta}{M}\sum_{j\neq i}\delta^j\right\}}\,.
		\end{align*}
		The change in total market intensity then depends on the sum of the derivatives of individual intensities with respect to spreads of all other agents, which is
		\begin{align*}
			\sum^M_{j=1}\sum_{k\neq i}\frac{\partial\lambda^j}{\partial\delta^k} &= \sum_{j\neq i}\frac{\partial\lambda^i}{\partial\delta^j}+\sum_{j\neq i}\sum_{k\neq i}\frac{\partial\lambda^j}{\partial\delta^k}\\
			&=\frac{\beta (M-1)}{M}\lambda^i+\sum_{j\neq i}\left(-\left(\kappa+\frac{\beta(M-1)}{M}\right)\lambda^j+\frac{\beta(M-2)}{M}\lambda^j\right)\\
			&=\beta\left(1-\frac{1}{M}\right)\lambda^i-\left(\kappa+\frac{\beta}{M}\right)\sum_{j\neq i} \lambda^j\\
			&\le \beta\left(1-\frac{1}{M}\right)\lambda^i-\left(\kappa \left(M-1\right)+\beta\left(1-\frac{1}{M}\right)\right)\min_{j\neq i} \lambda^j\,.
		\end{align*}
		For sufficiently large $M$, this quantity is negative. This means that if every agent except for the representative agent raises their spreads by the same amount, although the intensity of selling for the representative agent increases, the total sales across all agents decrease which is economically sensible.
		
		In the mean-field game setting, for a given deterministic function $\overline{\delta}$, the representative agent optimizes
		\begin{align}
			\label{eqn:ControlProblem}
			H(t,s,x,q;\overline{\delta})=\sup_{(\delta_u)_{t\le u\le T}\in\mathcal{A}} \mathbb{E}_{t,s,x,q} \left[ X_T^{\delta,\overline{\delta}} + Q_T^{\delta,\overline{\delta}}\,(S_T-\alpha\, Q_T^{\delta,\overline{\delta}})-\phi\,\int^T_t (Q_u^{\delta,\overline{\delta}})^2du \right],
		\end{align}
		where the set of admissible controls $\mathcal{A}$ is given by
		\begin{align*}
			\mathcal{A} = \left\{ \delta\, \Big\vert\, \delta_t = f(t,Q_t^{\delta,\overline{\delta}})\,, \quad \underline{B}\le f, \quad f \mbox{ is continuous in } t \right\}\,,
		\end{align*}
		for some large negative constant $\underline{B}$. Similar to Section \ref{sec:reference_model}, this optimization problem has an associated HJB equation and terminal condition given by (see again \cite{gueant2012optimal})
		\begin{align}    \label{eqn:HJBOriginal}
			\begin{split}
				&\partial_t H + \frac{1}{2}\,\sigma^2\,\partial_{ss}H - \phi\, q^2+\sup_{\delta\ge \underline{B}}\lambda(\delta, \overline{\delta}_t) \left[ H(t,s,x+s+\delta,q-1;\overline{\delta})-H(t,s,x,q;\overline{\delta})\right]\,\mathcal{X}_{q>0} = 0\,,\\
				&H(T,s,x,q;\overline{\delta}) = x+q\,(s-\alpha\,q)\,.
			\end{split}
		\end{align}
		We use a similar ansatz as in equation \eqref{eqn:ExecessValueAnsatz} to solve the HJB equation.
		
		\begin{proposition}[Solution to HJB Equation]
			\label{prop:SolutiontoHJBEquation}
			The HJB equation \eqref{eqn:HJBOriginal} admits the ansatz $H(t,s,x,q;\overline{\delta}) = x+q\,s+h_q(t;\overline{\delta})$, where $h$ satisfies
			\begin{align}
				\label{eqn:HJB}
				\partial_t h_q -\phi\,q^2+ \sup_{\delta\ge \underline{B}} A\,\exp{\left\{-(\kappa+\beta)\,\delta+\beta\, \overline{\delta}_t\right\}}\left[\delta+h_{q-1}(t;\overline{\delta})-h_q(t;\overline{\delta}) \right]\,\mathcal{X}_{q>0} = 0\,,
			\end{align}
			subject to $h_q(T;\overline{\delta}) = -\alpha\, q^2$. The optimum in equation \eqref{eqn:HJB} is achieved at
			\begin{align}
				\label{eqn:OptimalFeedbackControls}
				\delta^*(t,q;\overline{\delta}) = \max\left\{  \frac{1}{\kappa+\beta}+h_q(t;\overline{\delta})-h_{q-1}(t;\overline{\delta})\,,\,\, \underline{B}\right\}\,, \quad q\neq0\,.
			\end{align}
			Furthermore, equation \eqref{eqn:HJB} has an unique classical solution.
		\end{proposition}
		\begin{proof}
			See Appendix \ref{proof:SolutiontoHJBEquation}.
		\end{proof}
		
		We see in Proposition \ref{prop:SolutiontoHJBEquation} that the feedback form of the optimal strategy is similar in form to the single agent case. There are two differences, one being a change based on immediate effects of competition which sees $\kappa$ replaced $\kappa+\beta$, and the other is the fact that the change in expected future risk-adjusted P\&L, represented by $h_{q-1}(t;\overline{\delta}) - h_q(t;\overline{\delta})$, will also be different from the corresponding term in Proposition \ref{prop:SolutiontoHJBEquationSingle}. The terminal condition $h_q(T;\overline{\delta}) = -\alpha q^2$ determines the optimal spreads at time $T$ regardless of equilibrium considerations, and inspection shows that these spreads will be smaller than the single agent case, as would be expected with the effect of competition. The same cannot necessarily be said for all $t<T$ without solving for equilibrium.
		
		To show that the solution to the HJB equation is indeed the solution to the control problem, we prove the following verification theorem. In particular, this establishes that the strategy computed in the iterative algorithm of Section \ref{sec:Algorithm} is optimal given an assumed mean spread process.
		
		\begin{theorem}[Verification Theorem]
			\label{theo:VerificationTheorem}
			Given continuous $\overline{\delta}:[0,T]\rightarrow\mathbb{R}$, let $h_q(t;\overline{\delta})$ be the solution to equation \eqref{eqn:HJB}. Then $H(t,s,x,q;\overline{\delta}) = x+q\,s+h_q(t;\overline{\delta})$ is the value function
			to the agent’s control problem \eqref{eqn:ControlProblem} and the optimal controls are given by equation \eqref{eqn:OptimalFeedbackControls} in feedback form.
		\end{theorem}
		\begin{proof}
			See Appendix \ref{proof:VerificationTheorem}.
		\end{proof}
		
		In the statement of Theorem \ref{theo:VerificationTheorem}, we make the technical assumption that the given mean spread, $\overline{\delta}$, is continuous. In fact, we only require that it is bounded, for technical reasons needed in the proof. However, this is not a restrictive assumption because any optimal strategy with a feedback form given by equation \eqref{eqn:OptimalFeedbackControls} will be continuous, and therefore bounded, and the resulting $\overline{\delta}$ arising from all agents using this feedback control will also have these properties. Thus, the assumption is without loss of generality.
		
		\section{Numerical Experiments}
		\label{sec:NumericalExperiments}
		
		In this section, we numerically solve for the equilibrium, and illustrate the behaviour of the optimal strategy $\delta^*$, the mean spread $\overline{\delta}$, the inventory distribution across agents $P$, and other quantities of interest.
		
		\subsection{Algorithm}
		\label{sec:Algorithm}
		
		Inspired by \cite{li2024mean}, we use the following algorithm to numerically find an equilibrium.\footnote{The existence of an equilibrium according to Definition \ref{def:mean_field_equilibrium} has not been proven, but for a fixed set of parameters, all of our numerical experiments using the algorithm specified here converge to the same limit (within numerical tolerance) regardless of the choice of initialization function $\overline{\delta}^0$.} First, for a given mean spread we solve the control problem, the optimality of which is ensured by our verification theorem. Then, we use the optimal individual spread to find the population dynamics. With these two solutions, we calculate the corresponding new mean spread and repeat this process until convergence. The details are as follows:
		\begin{enumerate}
			\item We divide the entire time horizon $T$ into an equidistant time grid. Initialize with a mean spread $t \to \overline{\delta}^{(0)}_t$ for every time point.
			\item Given any mean spread, $\overline{\delta}^{(n)}$, we solve the optimal control problem numerically. This is done with a fully explicit backward finite difference method along the lines of \cite{leveque2007finite} to solve equation \eqref{eqn:HJB} starting from time $T$. This simultaneously gives the optimal control for the given mean spread, $\delta^*(\overline{\delta}^{(n)})$, in feedback form through equation \eqref{eqn:OptimalFeedbackControls}.
			\item With the optimal control $\delta^*(\overline{\delta}^{(n)})$ given in feedback form, we compute the population process $P^{\delta^*}$ using equation \eqref{eqn:dP}. This is also done numerically using a fully explicit forward finite difference method.
			\item We introduce a learning rate parameter $\gamma \in [0,1)$ to update the mean spread process. To reduce oscillations in searching for the equilibrium, we choose $\gamma$ to be a small number, and we update the mean spread according to
			\begin{align*}
				\overline{\delta}^{(n+1)}_t = (1-\gamma)\, \overline{\delta}^{(n)}_t + \gamma\,\frac{\sum_{q=1}^{\overline{Q}}\delta^*(t,q;\overline{\delta}^{(n)})\,P^{\delta^*}_{q,t}}{1-P_{0,t}^{\delta^*}}.
			\end{align*}
			These steps are repeated until convergence of the sequence $\{\overline{\delta}^{(n)}\}_{n\geq 0}$ to within a specified tolerance\footnote{In the following numerical examples, we identify the algorithm as converged if the standard error of $\overline{\delta}^{(n+1)}$ and $\overline{\delta}^{(n)}$ is not more than $10^{-12.5}$. We do not claim that equilibrium is guaranteed to exist or is unique, but for all sets of parameters considered, convergence was attained and was found not to depend on the initializing $\overline{\delta}$. More numerical details and explanations can be found in Appendix \ref{exp:NumericalUniqueness}.}, and we drop the counting index in the final mean spread process.
		\end{enumerate}
		Once convergence is attained, we will have found a function $f$, given by $f(t,q) = \delta^*(t,q;\overline{\delta})$, along with its corresponding process $P^f$ which satisfy the conditions of equilibrium given in Definition \ref{def:mean_field_equilibrium}.
		
		\subsection{Optimal Spreads in Equilibrium}
		
		\begin{figure}[!htp]
			\centering
			\includegraphics[width=0.7\textwidth]{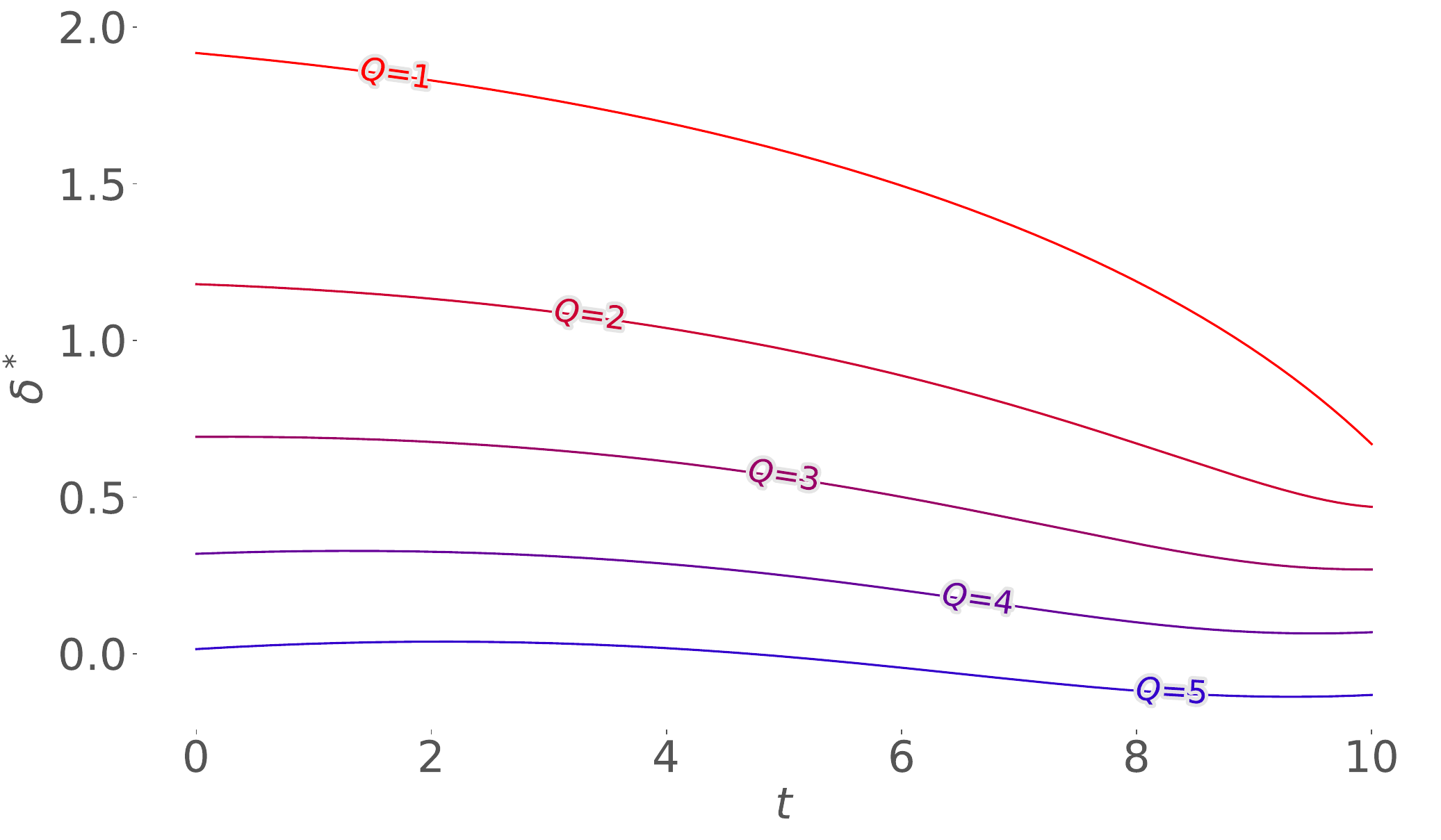}
			\caption{Optimal spreads in equilibrium as a function of $t$ for various $Q$. Other parameters are $T=10$, $\overline{Q}=5$, $\alpha=0.1$, $\kappa=1$, $\phi=0.03$, $A=1$, $\beta = 0.3$, $\gamma=0.1$, and $\underline{B}=-10$.}
			\label{fig:MultiQuote}
		\end{figure}
		
		The optimal spreads $\delta^*$ are presented in Figure \ref{fig:MultiQuote}, as well as the mean spread $\overline{\delta}$ in Figure \ref{fig:MeanQuote} and the dynamics of the proportion process $P$ in Figure \ref{fig:DynamicsOfP}. From Figure \ref{fig:MultiQuote}, we observe that the optimal spreads in the mean-field case share the same decreasing pattern against inventory level with ones in the one-agent case in Figure \ref{fig:MultiQuoteReference}. In addition, for lower inventory levels ($Q\le3$) the optimal spreads are decreasing with time, whereas for higher inventory levels ($Q\ge4$) the optimal spreads slowly increase early in the time interval and then decrease. This effect is caused by the changing nature of competition embedded in the mean spread $\overline{\delta}$ (see Figure \ref{fig:MeanQuote}). At early times, everyone is posting small spreads, but as some agents begin to sell the mean spread quickly increases. For an agent that remains in a high inventory state, this increase in average price in the market creates less competition, and they begin to benefit themselves by quoting slightly higher prices.
		
		\begin{figure}[!htp]
			\centering
			\includegraphics[width=0.7\textwidth]{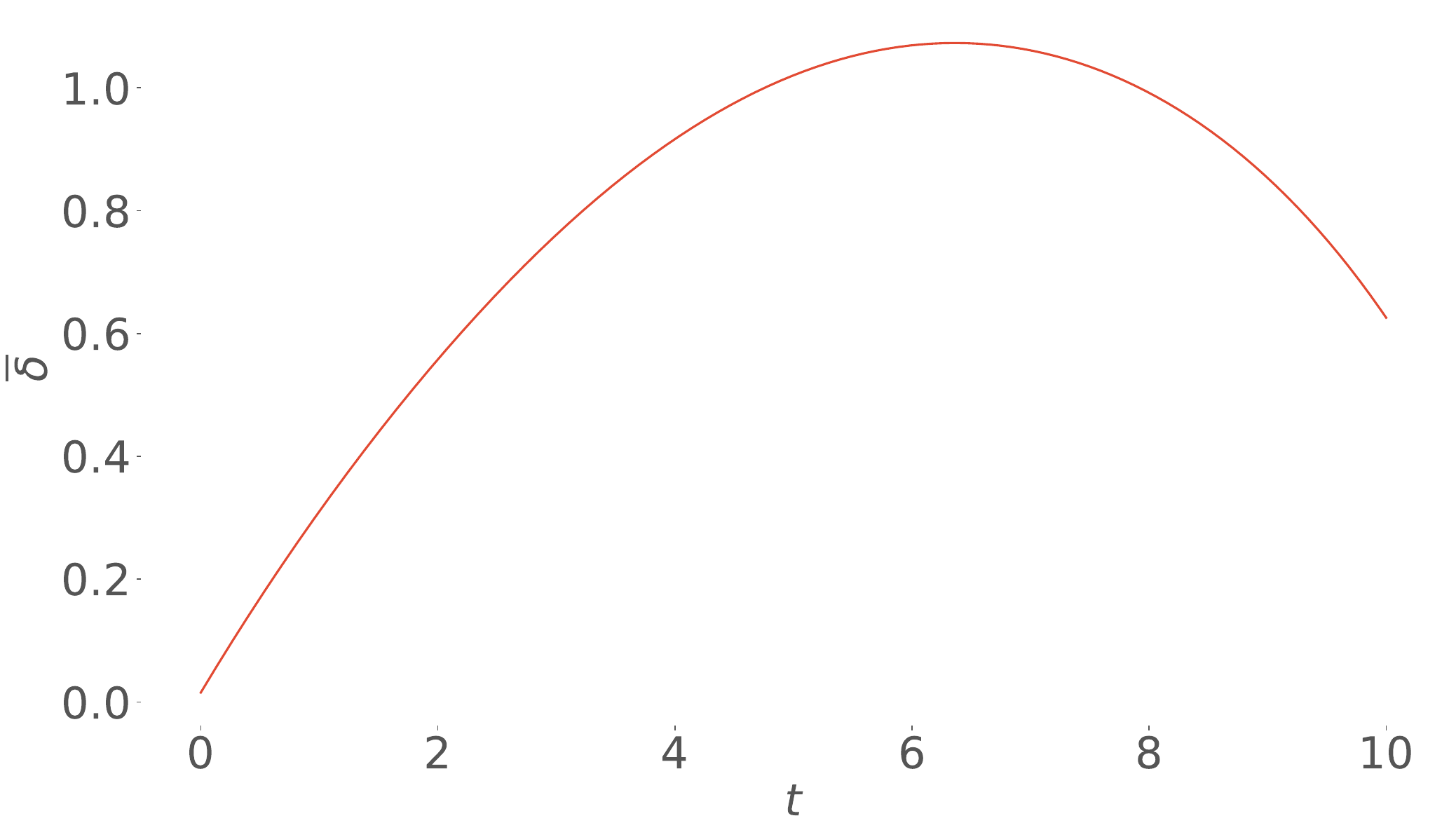}
			\caption{Mean spread in equilibrium as a function of $t$. Other parameters are $T=10$, $\overline{Q}=5$, $\alpha=0.1$, $\kappa=1$, $\phi=0.03$, $A=1$, $\beta = 0.3$, $\gamma=0.1$, and $\underline{B}=-10$.}
			\label{fig:MeanQuote}
		\end{figure}
		
		\begin{figure}[!htp]
			\centering
			\includegraphics[width=0.7\textwidth]{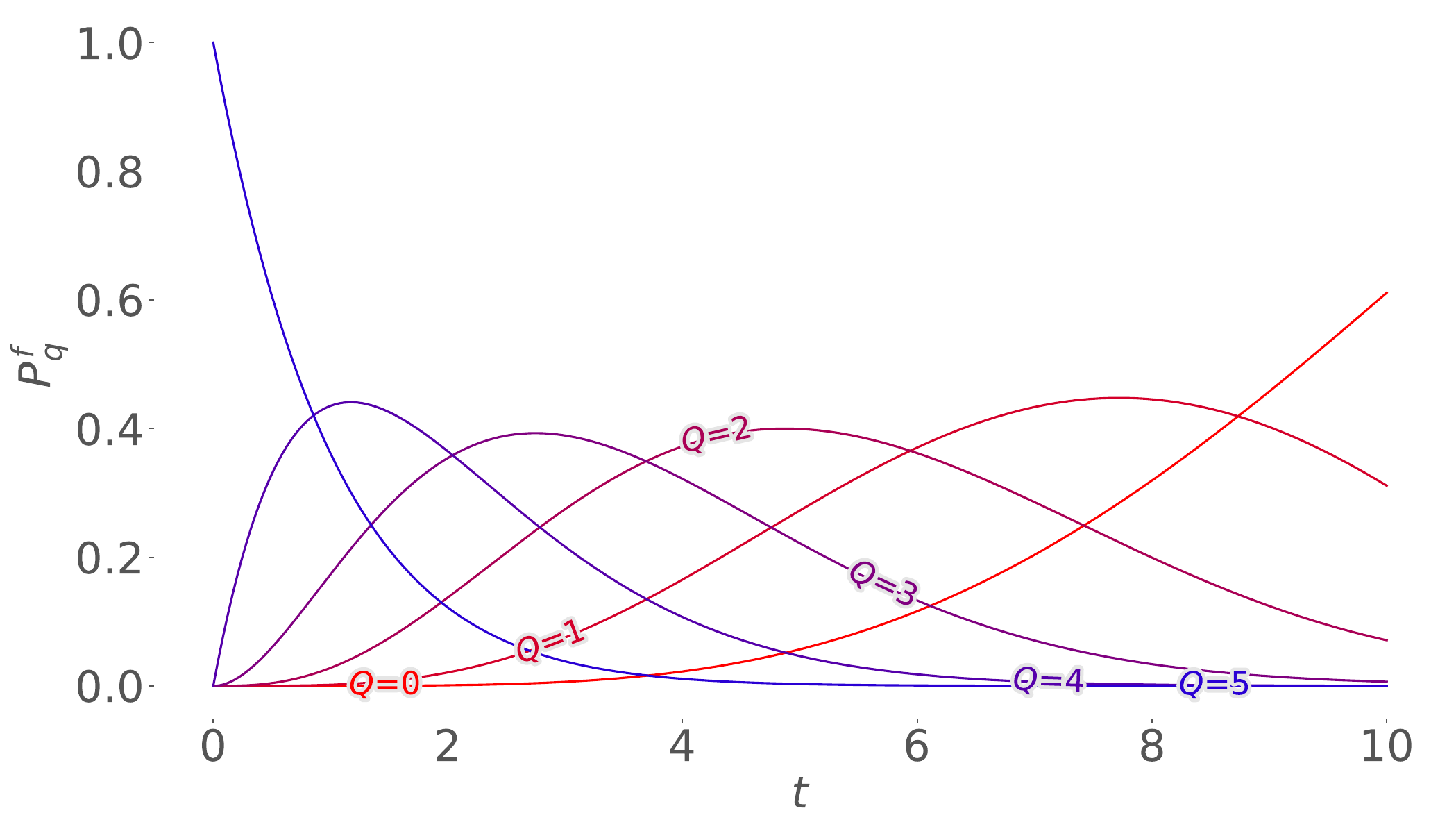}
			\caption{Dynamics of the population proportion process, $P_{q,t}$, in equilibrium as a function of $t$. Other parameters are $T=10$, $\overline{Q}=5$, $\alpha=0.1$, $\kappa=1$, $\phi=0.03$, $A=1$, $\beta = 0.3$, $\gamma=0.1$, and $\underline{B}=-10$.}
			\label{fig:DynamicsOfP}
		\end{figure}
		
		In Figure \ref{fig:MeanQuote}, because each agent begins with maximum inventory $\overline{Q}$, the mean spread begins at the optimal spread for that inventory level. It quickly rises as agents begin to sell inventory and move to lower inventory states, causing them to increase their quoted price. At later times the mean spread begins to decrease due to two effects. First, as time approaches the end of the trading period agents are incentivized to quote smaller prices to avoid the liquidation penalty. Second, more and more agents end up fully liquidating their positions, removing their large quoted prices from the market and thereby decreasing the mean spread. Figure \ref{fig:DynamicsOfP} shows the distribution, $P$, of the agents' inventories through time. For this set of parameters, the majority of agents sell all of their inventory by time $T$, indicated by the fact that $P_{0,T}>0.5$.
		
		\begin{figure}[h!]
			\centering
			\includegraphics[width=0.48\textwidth]{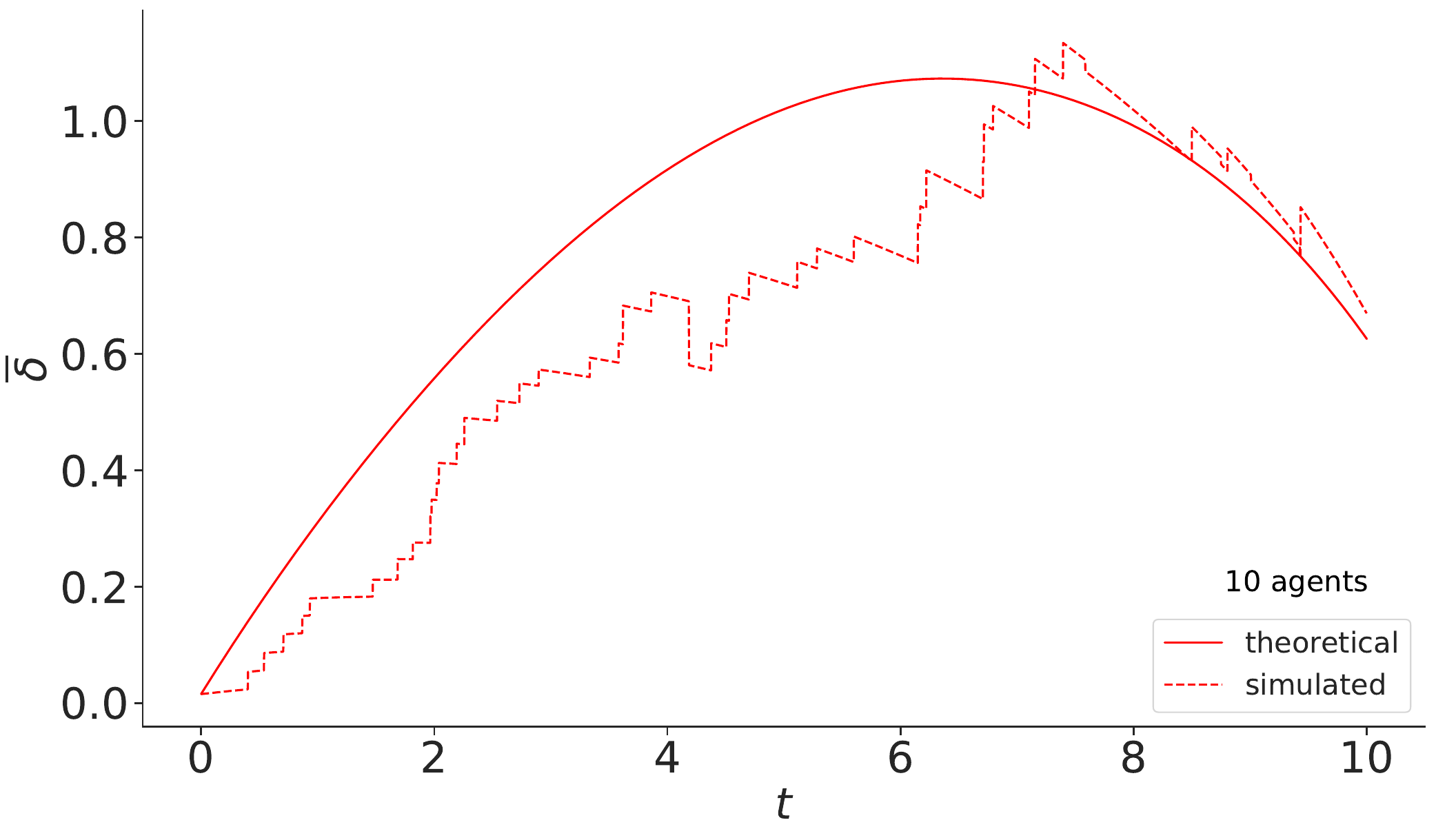}
			\includegraphics[width=0.48\textwidth]{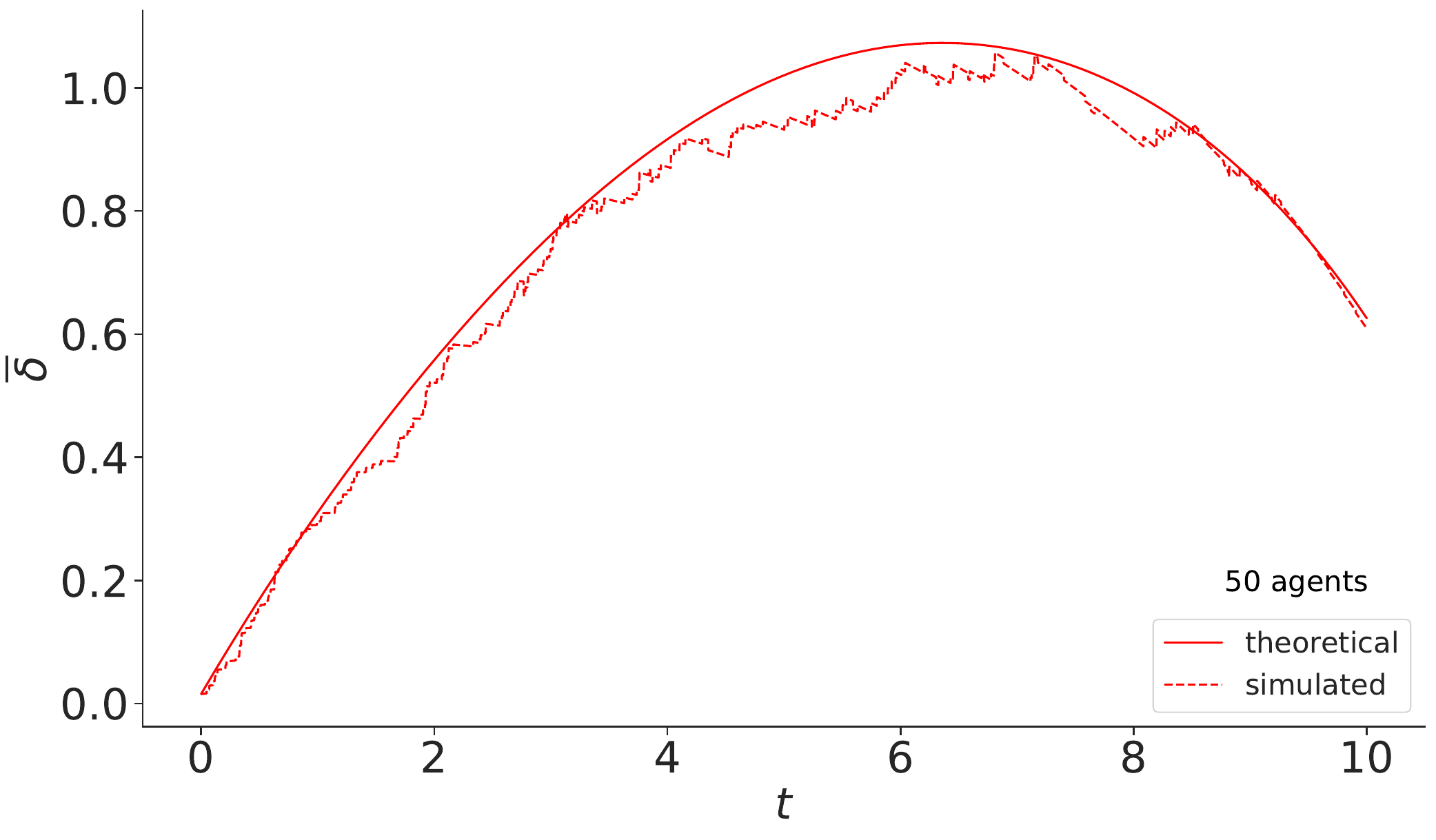}
			\includegraphics[width=0.48\textwidth]{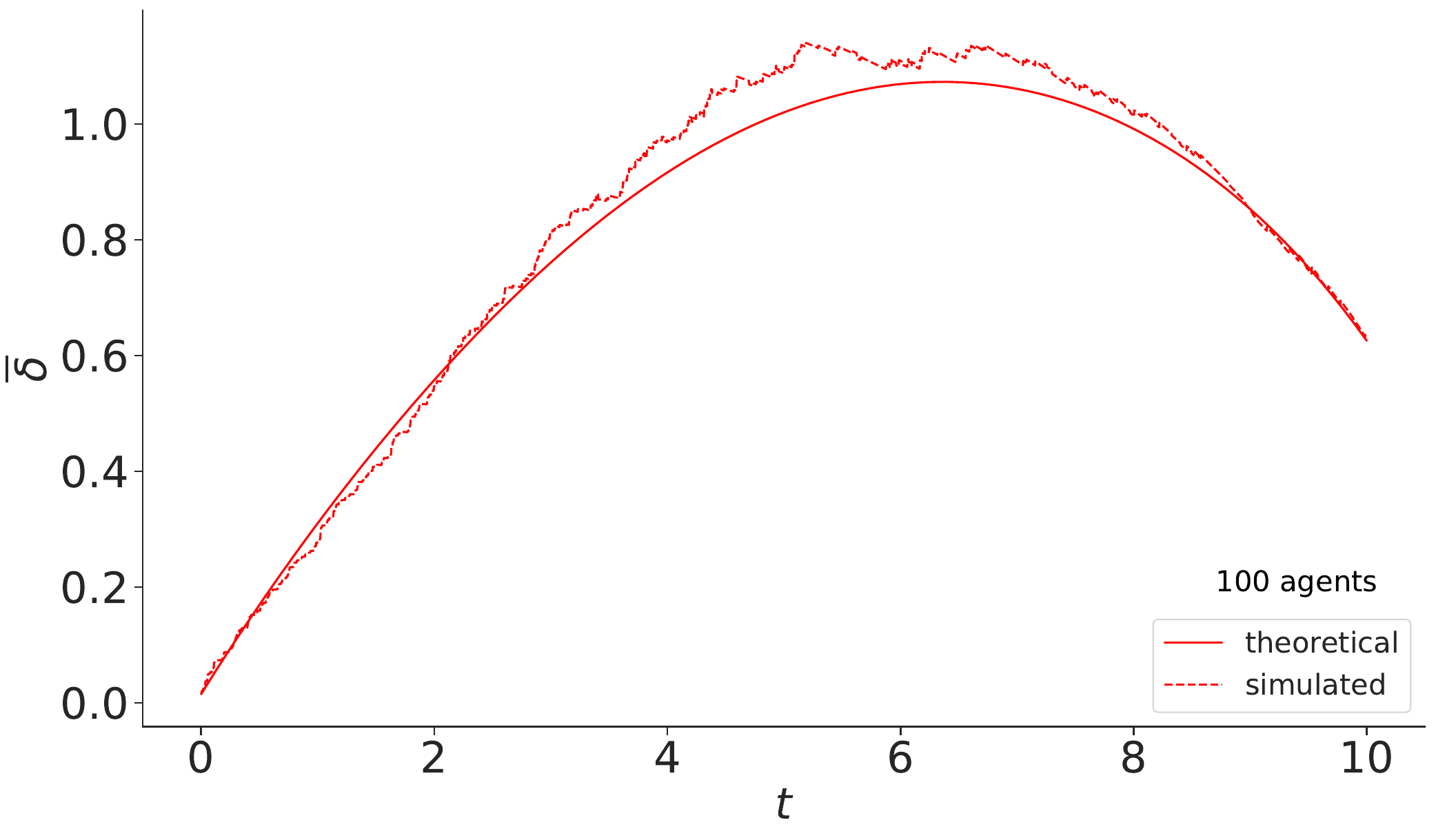}
			\includegraphics[width=0.48\textwidth]{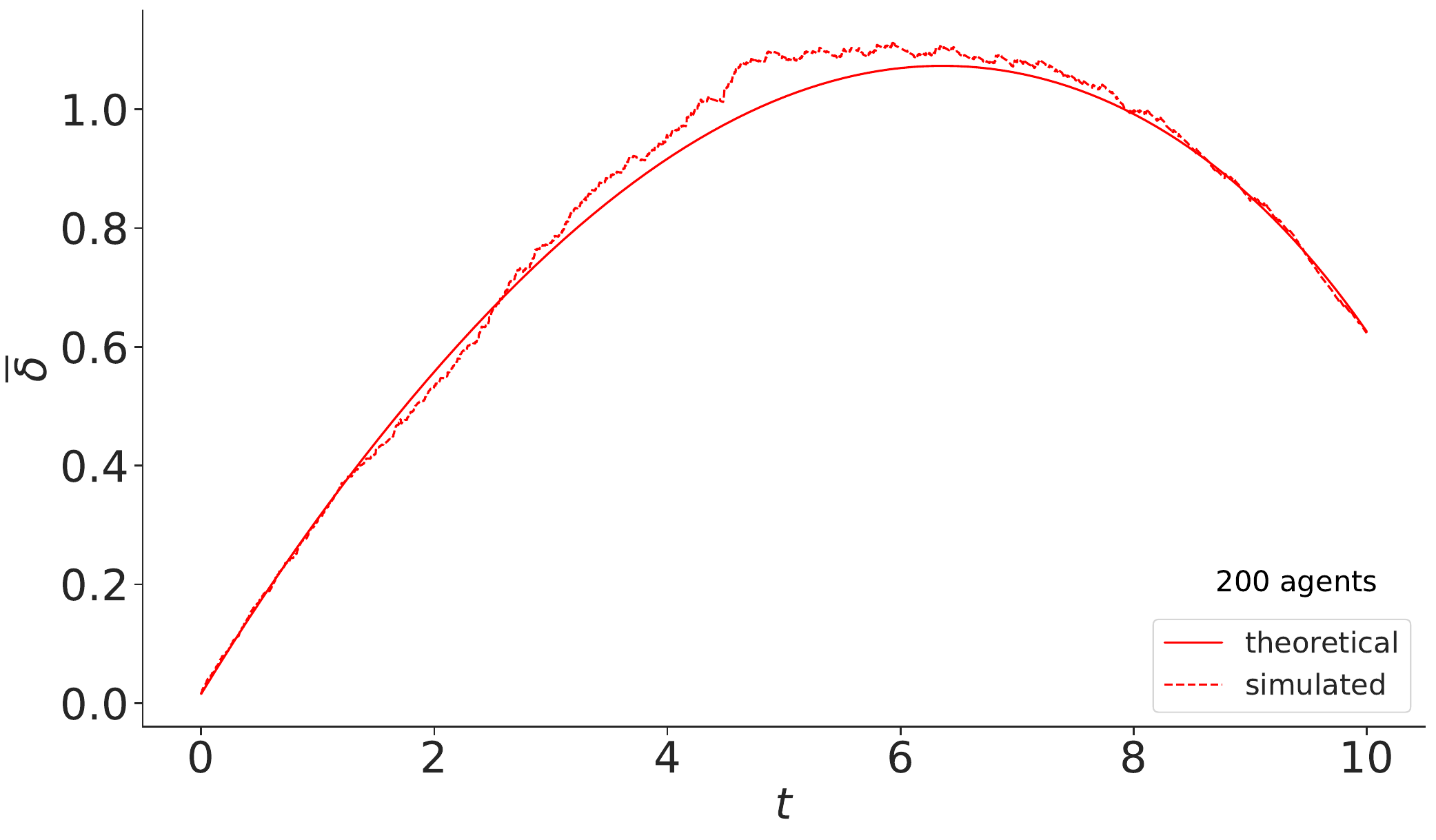}
			\caption{Mean spread process in mean-field equilibrium (solid curves) and simulated mean spreads fo finite number of agents (dotted curves). Other parameters are $T=10$, $\overline{Q}=5$, $\alpha=0.1$, $\kappa=1$, $\phi=0.03$, $A=1$, $\beta = 0.3$, $\gamma=0.1$, and $\underline{B}=-10$.}
			\label{fig:SimulatedMeanSpread}
		\end{figure}
		
		In order to illustrate the difference between the mean spread quantity computed in \eqref{eqn:delta_bar_M} with a finite number of agents and that computed in \eqref{eqn:delta_bar_f} in the mean-field setting, we simulate the inventory processes of multiple agents and plot both of these mean spreads in Figure \ref{fig:SimulatedMeanSpread}. The simulation is performed as follows:
			
			\begin{enumerate}
				\item On an equidistant time grid $0=t_0 < t_1 < \cdots < t_N = T$ perform the numerical iteration described in Section \ref{sec:Algorithm}, and denote the limiting deterministic function by $\overline{\delta}$ (this is the smooth function plotted in each panel of Figure \ref{fig:SimulatedMeanSpread}).
				\item At time $t_0=0$, compute the right hand side of \eqref{eqn:delta_bar_M} and denote this by $\overline{\delta}^M_{t_0}$.
				\item For each subsequent point on the time grid, $t_k$, simulate the inventory process of all $M$ agents using \eqref{eqn:InventoryProcess} where agent $i$ is subject to the intensity $\lambda(\delta^*(t,Q^{\delta^*,\overline{\delta}^M}_{i,t_{k-1}};\overline{\delta}),\overline{\delta}^M_{t_{k-1}})$ with $\lambda$ as in \eqref{eqn:PoissonIntensityAssumption} and $\delta^*$ as in \eqref{eqn:OptimalFeedbackControls}, and recompute $\overline{\delta}^M_{t_k}$ using \eqref{eqn:delta_bar_M}, iterating this step until the terminal time $T$.
			\end{enumerate}
			Each panel in Figure \ref{fig:SimulatedMeanSpread} shows the deterministic function $\overline{\delta}$ along with a single path of $\overline{\delta}^M$ for various values of $M$. The convergence of the finite agent mean spread to that of the mean-field equilibrium is a demonstration of the propagation of chaos.

		\subsection{Comparison with Single-Agent Reference Model}
		
		\begin{figure}[!htp]
			\centering
			\includegraphics[width=0.7\textwidth]{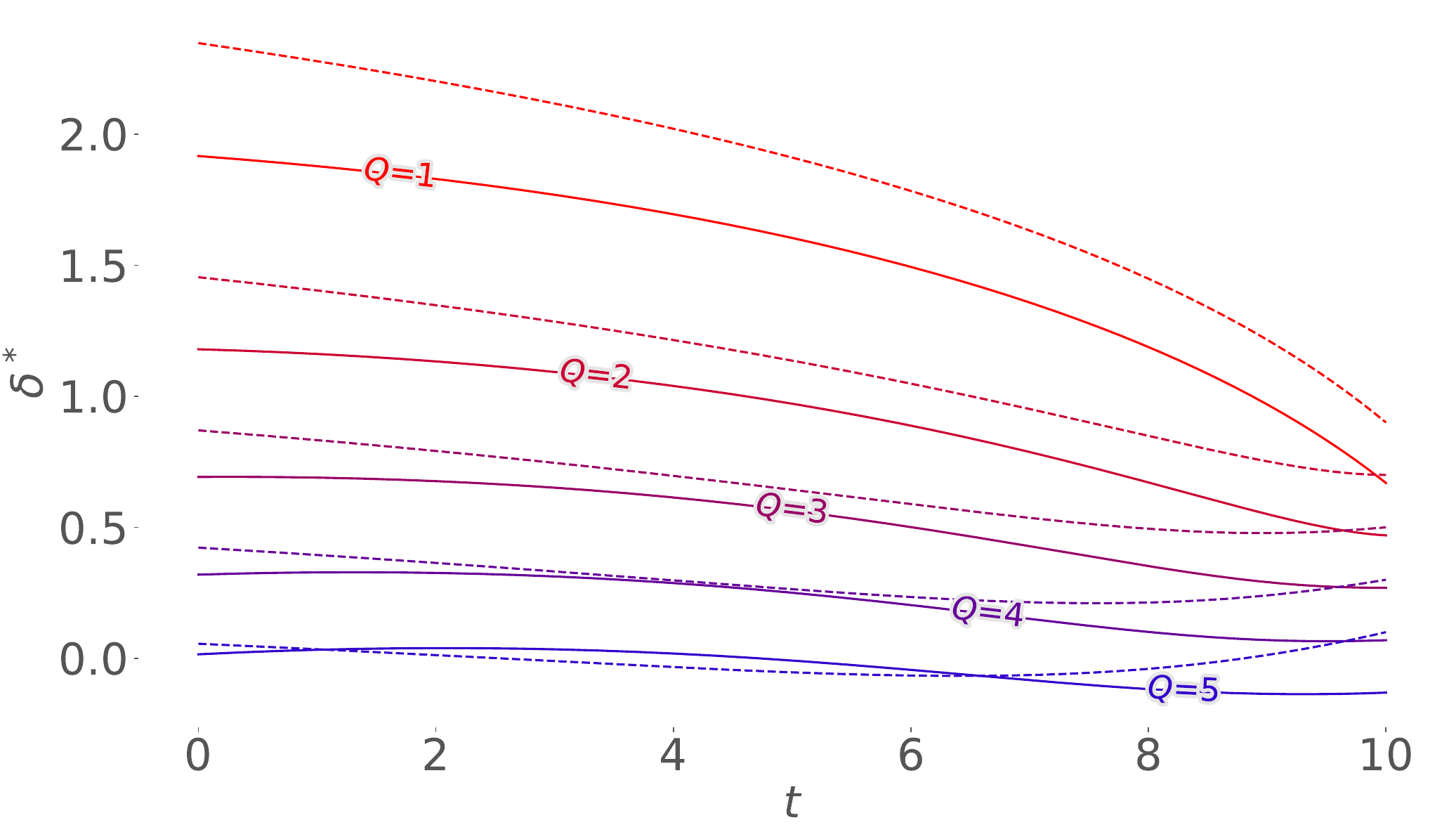}
			\caption{Optimal spreads in mean-field equilibrium (solid curves) compared to single-agent reference model (dotted curves). Other parameters are $T=10$, $\overline{Q}=5$, $\alpha=0.1$, $\kappa=1$, $\phi=0.03$, $A=1$, $\beta = 0.3$, $\gamma=0.1$, and $\underline{B}=-10$.}
			\label{fig:One-MultiComparison}
		\end{figure}
		
		Figure \ref{fig:One-MultiComparison} compares the optimal spreads in the monopoly and mean-field cases. For lower inventory levels $Q \le 4$, the optimal spreads in the mean-field case are consistently lower than the ones in the single-agent case. However, for the maximum inventory level $Q=5$, the optimal spreads in the mean-field case are sometimes higher than in the single-agent case. This trend is because the optimal spreads corresponding to higher inventory levels tend to be lower than the mean spread in the market. Thus, the items in the mean-field case have a higher probability to be executed compared to the single-agent case as per equation \eqref{eqn:PoissonIntensityAssumption}. As a result, the agent can be more aggressive by setting a higher spread in the mean-field case.
		
		\subsection{Effect of Competitiveness on Equilibrium}
		
		\begin{figure}[!htp]
			\centering
			\includegraphics[width=0.48\textwidth]{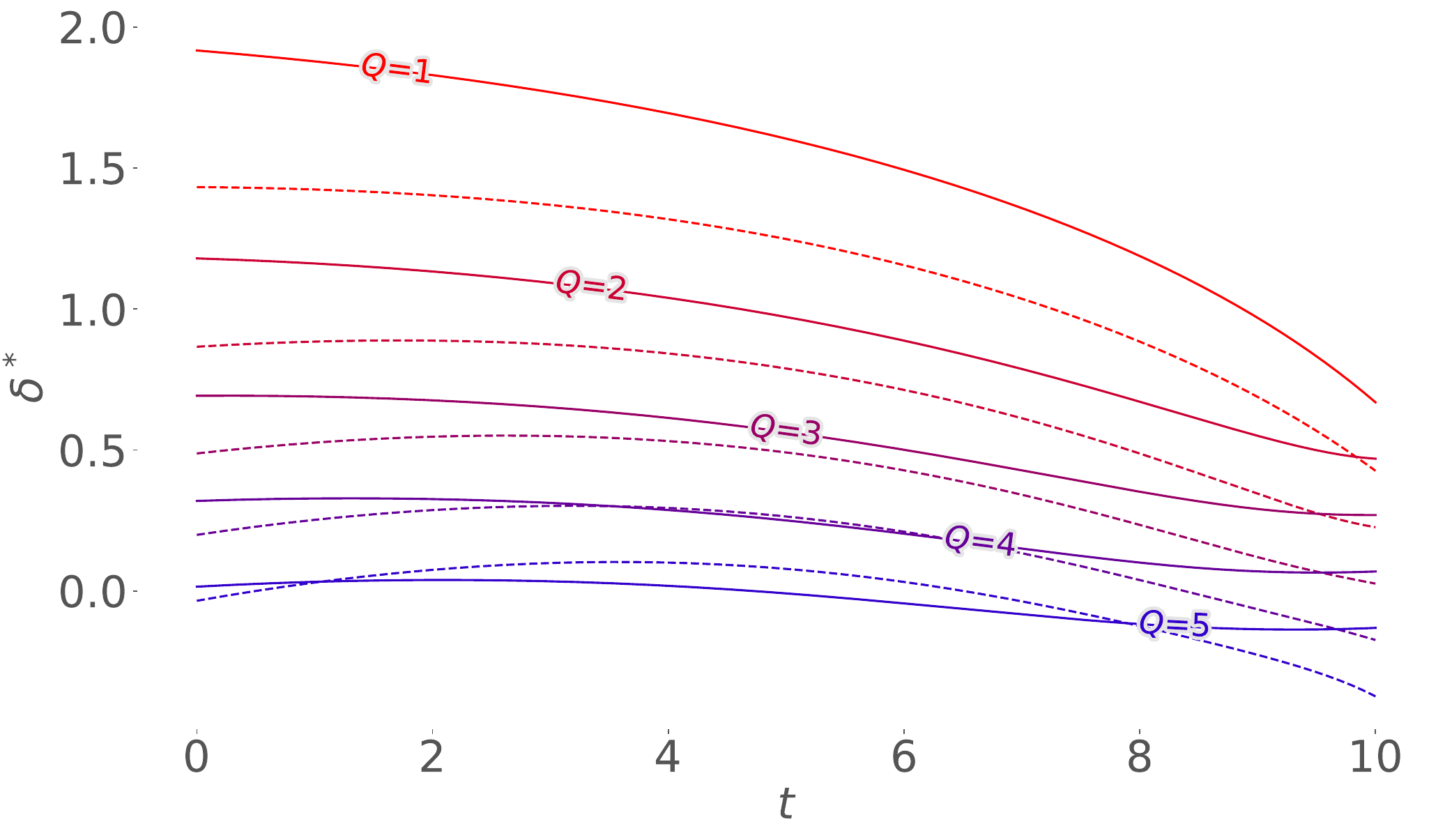}
			\includegraphics[width=0.48\textwidth]{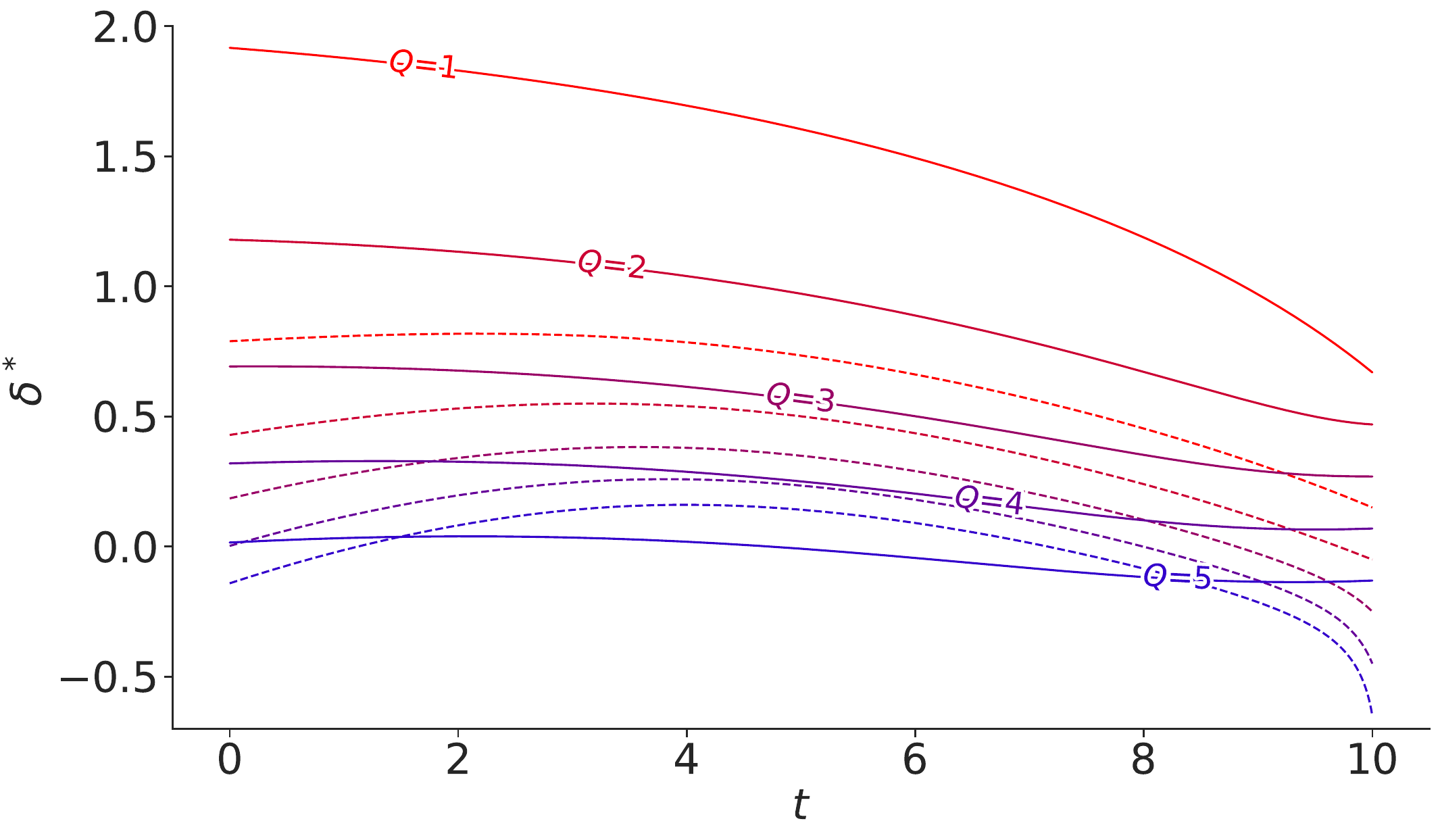}
			\caption{Optimal spreads in equilibrium for different values of competition parameter $\beta$. The solid curves represent a low level of competitive interaction ($\beta = 0.3$) and the dotted curves represent a high level of competitive interaction (Left: $\beta = 0.9$. Right: $\beta = 3$). Other parameters are $T=10$, $\overline{Q}=5$, $\alpha=0.1$, $\kappa=1$, $\phi=0.03$, $A=1$, $\gamma=0.1$, and $\underline{B}=-10$.}
			\label{fig:DifferentBeta}
		\end{figure}
		
		Figure \ref{fig:DifferentBeta} shows how the competitiveness parameter, $\beta$, affects the optimal spreads. In general, higher competitiveness of the market brings the optimal spreads of different inventories closer together. However, note that higher competitiveness does not always move the optimal spreads of different inventory levels in the same direction. For lower inventory levels, larger $\beta$ leads to lower optimal spreads, but for higher inventory levels the direction of price change depends on time. In Figure \ref{fig:DifferentBetaMeanQuote&P_T}, the mean spread decreases in $\beta$ at all points in time. This means that more intense competition brings a lower market quote price, which is economically sensible. In addition, we see that a larger proportion of agents end up selling their entire inventory when the competitiveness parameter is larger.
		
		\begin{figure}
			\centering
			\includegraphics[width=0.48\textwidth]{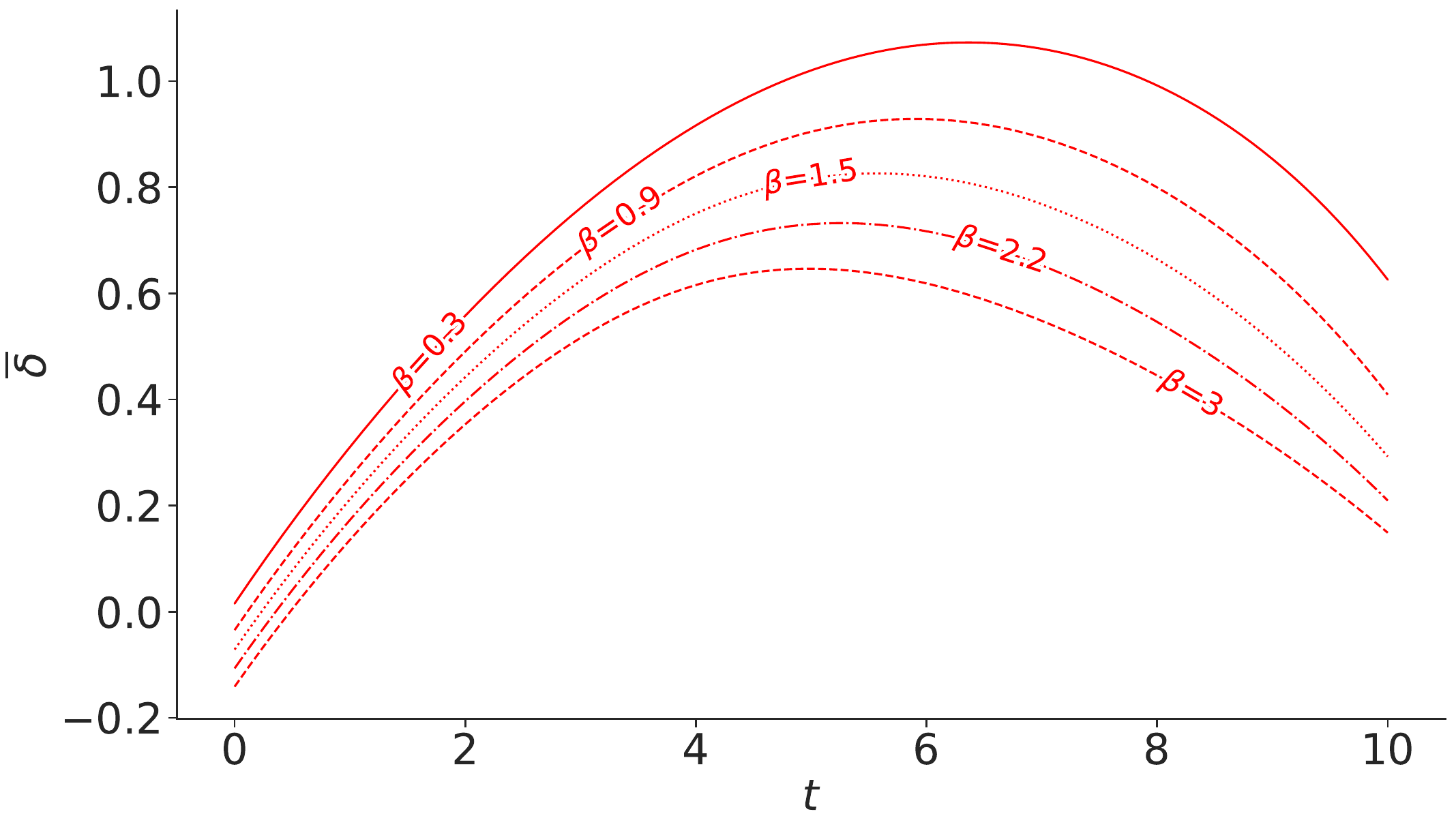}
			\includegraphics[width=0.48\textwidth]{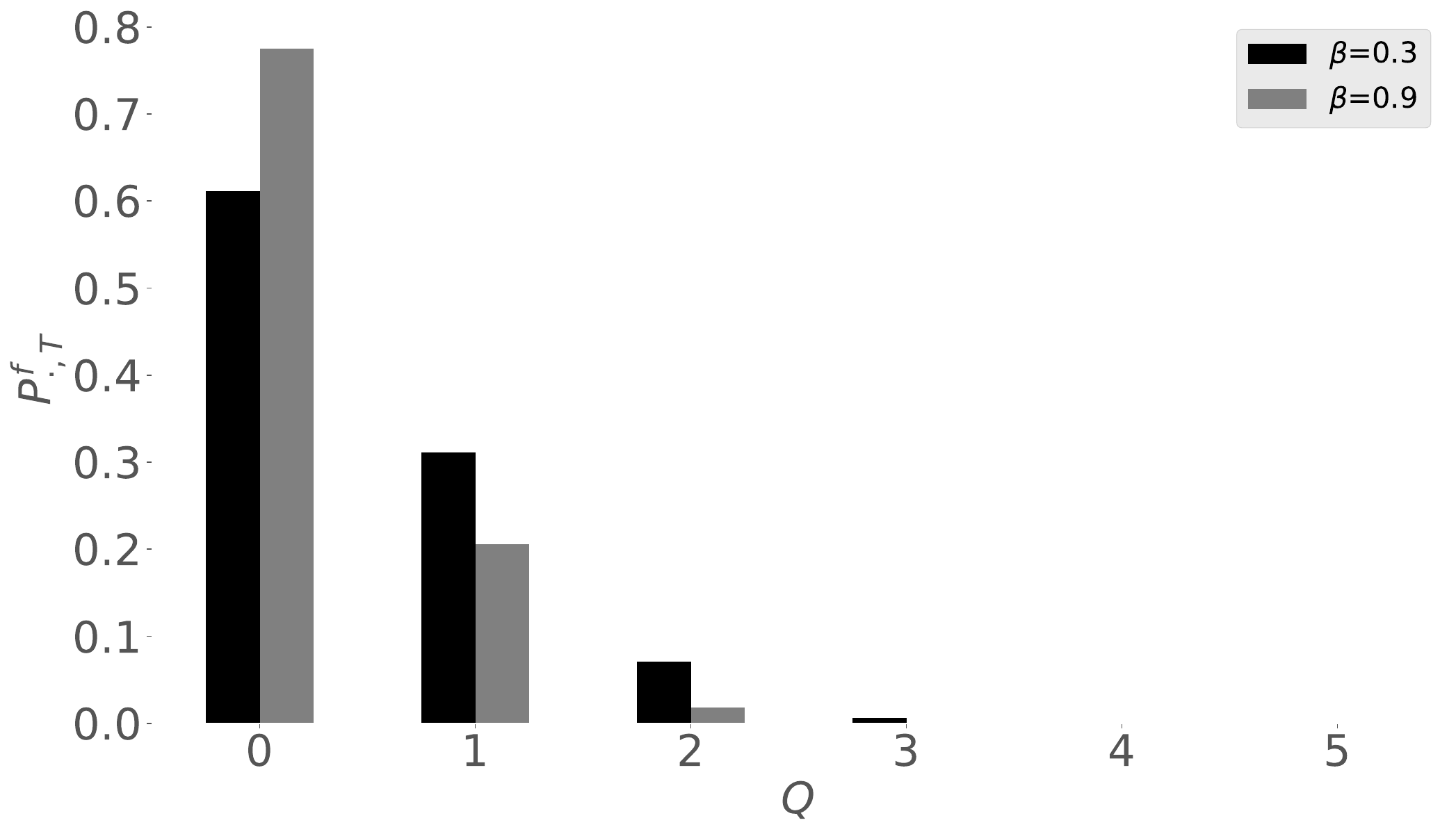}
			\caption{Mean spreads and the distribution of inventory across agents at terminal time in equilibrium with different values of competition parameter $\beta$. Other parameters are $T=10$, $\overline{Q}=5$, $\alpha=0.1$, $\kappa=1$, $\phi=0.03$, $A=1$, $\gamma=0.1$, and $\underline{B}=-10$.}
			\label{fig:DifferentBetaMeanQuote&P_T}
		\end{figure}
		
		To illustrate the effects of competitiveness on the consumers and agents, we define cumulative cost $C(t)$, cumulative revenue $R(t)$ and cumulative volume $V(t)$ up to time $t$ as
		\begin{align*}
			C(t) &= \sum_{q=1}^{\overline{Q}}\int_0^{t}f(u,q)\,P^f_{q,u}\,\lambda(f(u,q),\overline{\delta}_u)\,du\,, \\
			R(t) &= \sum_{q=1}^{\overline{Q}}\left(\int_0^{t}f(u,q)\,P^f_{q,u}\,\lambda(f(u,q),\overline{\delta}_u)\,du-\alpha\,P^f_{q,t}\, q^2\right)\,, \\
			V(t) &= \sum_{q=1}^{\overline{Q}}\int_0^{t}\,P^f_{q,u}\,\lambda(f(u,q),\overline{\delta}_u)\,du\,.
		\end{align*}
		Cumulative cost $C(t)$ corresponds to the total wealth paid by consumers up to time $t$, and cumulative revenue $R(t)$ represents the profits made by agents up to time $t$. Theoretically, cumulative revenue $R(t)$ should be equal to cumulative cost $C(t)$ until terminal time when the penalty proportional to $\alpha$ is realized, but to ensure $R(t)$ is continuous, we preemptively charge this penalty in the form of $\alpha\,P^f_{q,t}\, q^2$, as if all agents stopped trading at time $t$. With the above functions, we can compute the corresponding average transaction cost $K(t)$ by
		\begin{align*}
			K(t) = \frac{C(t)}{V(t)}\,,
		\end{align*}
		and the instantaneous average transaction cost $\overline{K}(t)$ by
		\begin{align*}
			\overline{K}(t) = \frac{\partial C(t)}{\partial t}/\frac{\partial V(t)}{\partial t}
			= \frac{\sum_{q=1}^{\overline{Q}}f(t,q)P^f_{q,t}\lambda(f(t,q),\overline{\delta}_t)}{\sum_{q=1}^{\overline{Q}}P^f_{q,t}\lambda(f(t,q),\overline{\delta}_t)}\,.
		\end{align*}
		The average transaction cost $K(t)$ represents the average price consumers have paid for one unit of products up to time $t$, and the instantaneous average transaction cost $\overline{K}(t)$ is the the average price to buy one unit at time $t$.
		
		\begin{figure}[h!]
			\centering
			\includegraphics[width=0.48\textwidth]{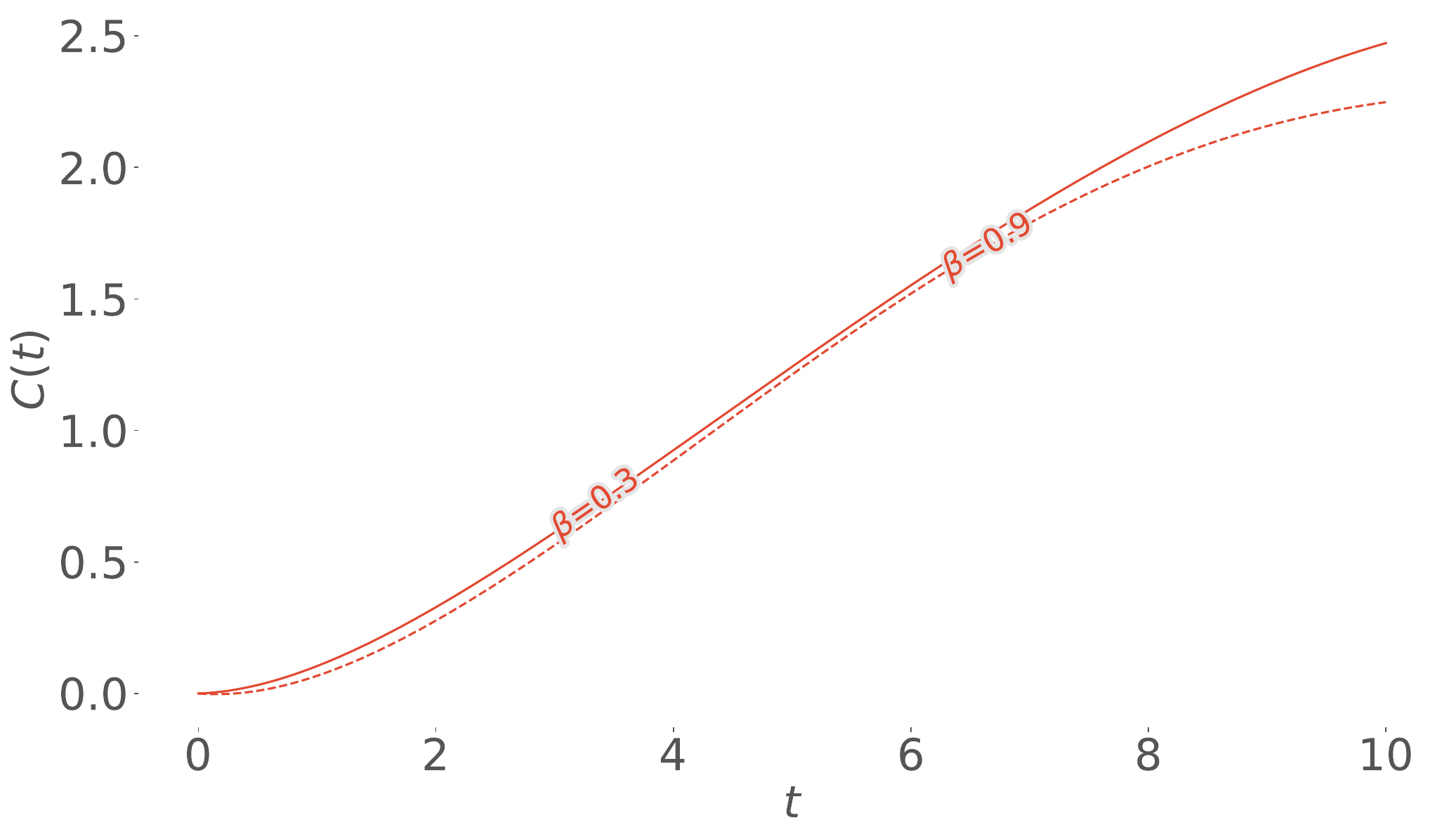}
			\includegraphics[width=0.48\textwidth]{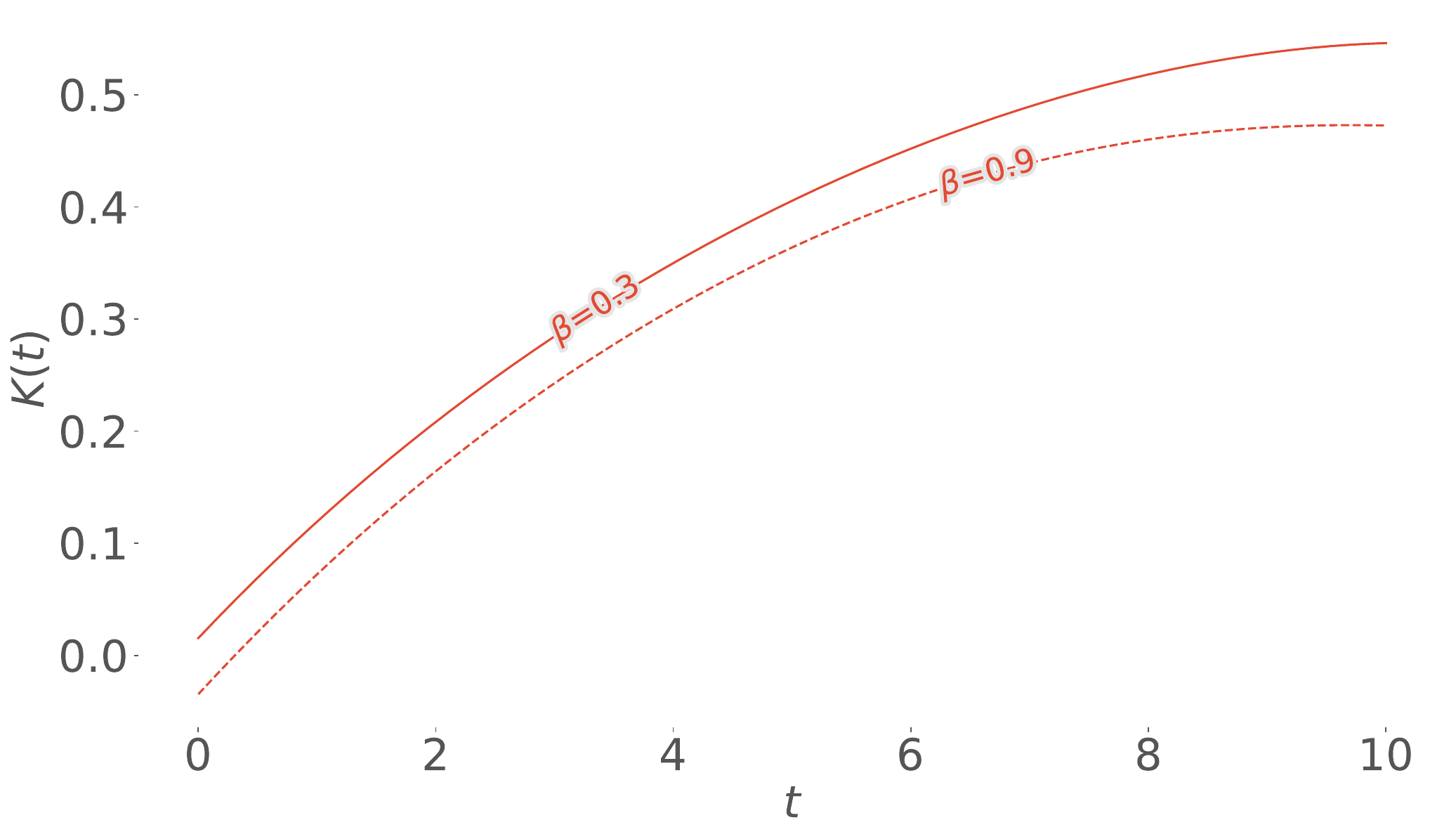}
			\includegraphics[width=0.48\textwidth]{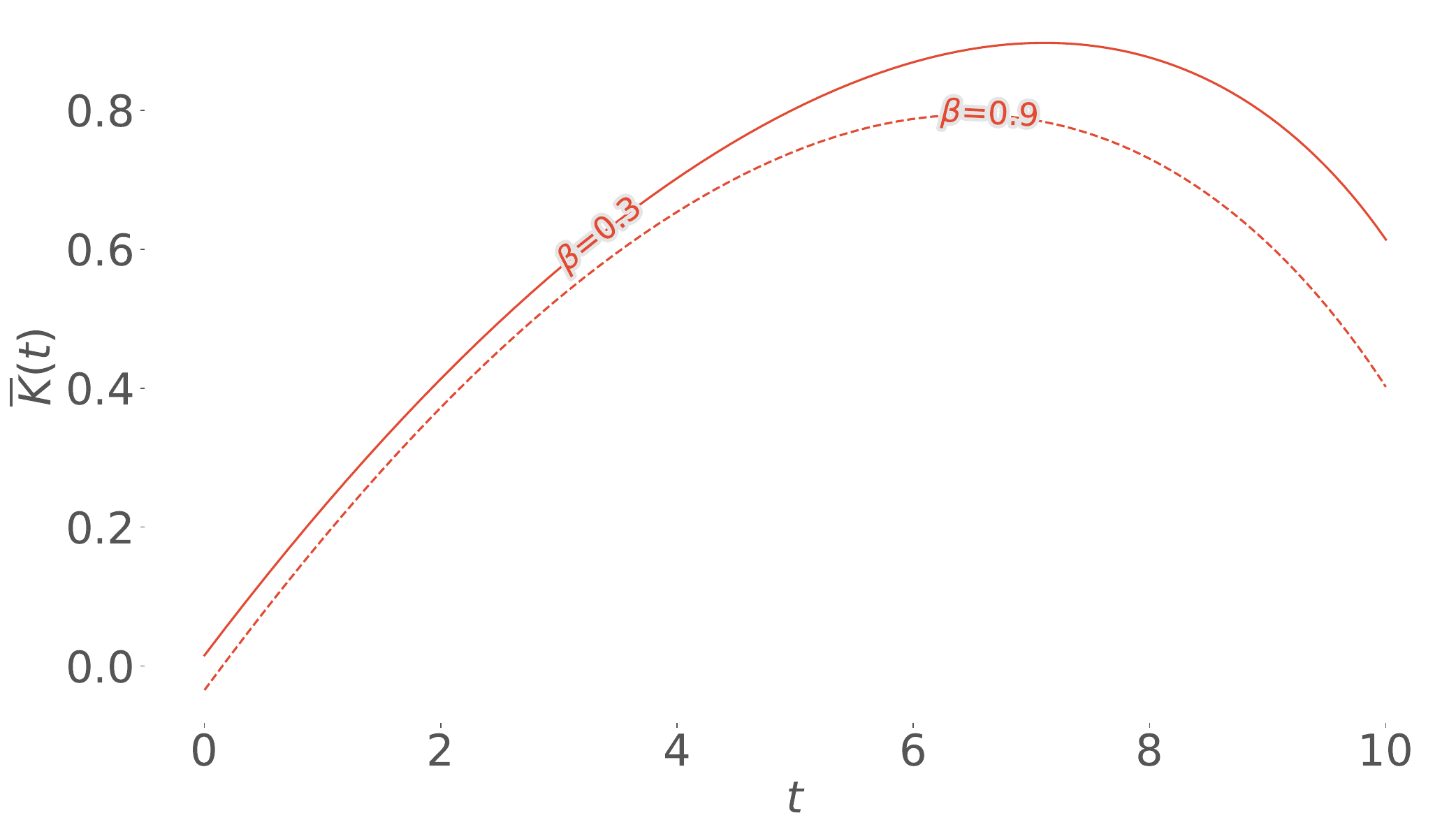}
			\caption{Cumulative cost (top left), average transaction cost (top right) and instantaneous average cost (bottom) with different values of competition parameter $\beta$. The solid curves represent a low level of competitive interaction ($\beta = 0.3$) and the dotted curves represent a high level of competitive interaction ($\beta = 0.9$). Other parameters are $T=10$, $\overline{Q}=5$, $\alpha=0.1$, $\kappa=1$, $\phi=0.03$, $A=1$, $\gamma=0.1$, and $\underline{B}=-10$.}
			\label{fig:DifferentBetaEcoFeatureBuyer}
		\end{figure}
		
		\begin{figure}[h!]
			\centering
			\includegraphics[width=0.48\textwidth]{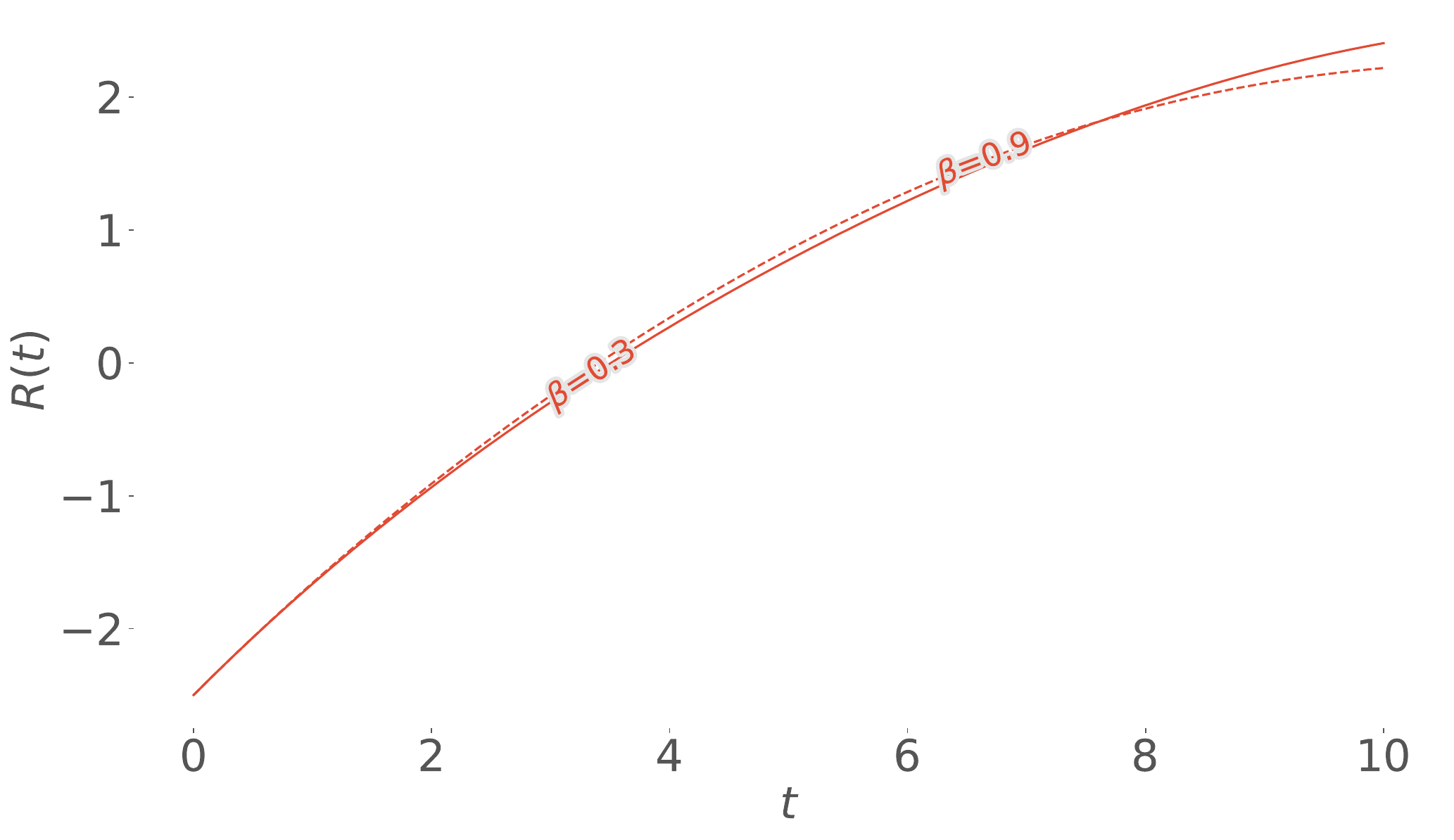}
			\includegraphics[width=0.48\textwidth]{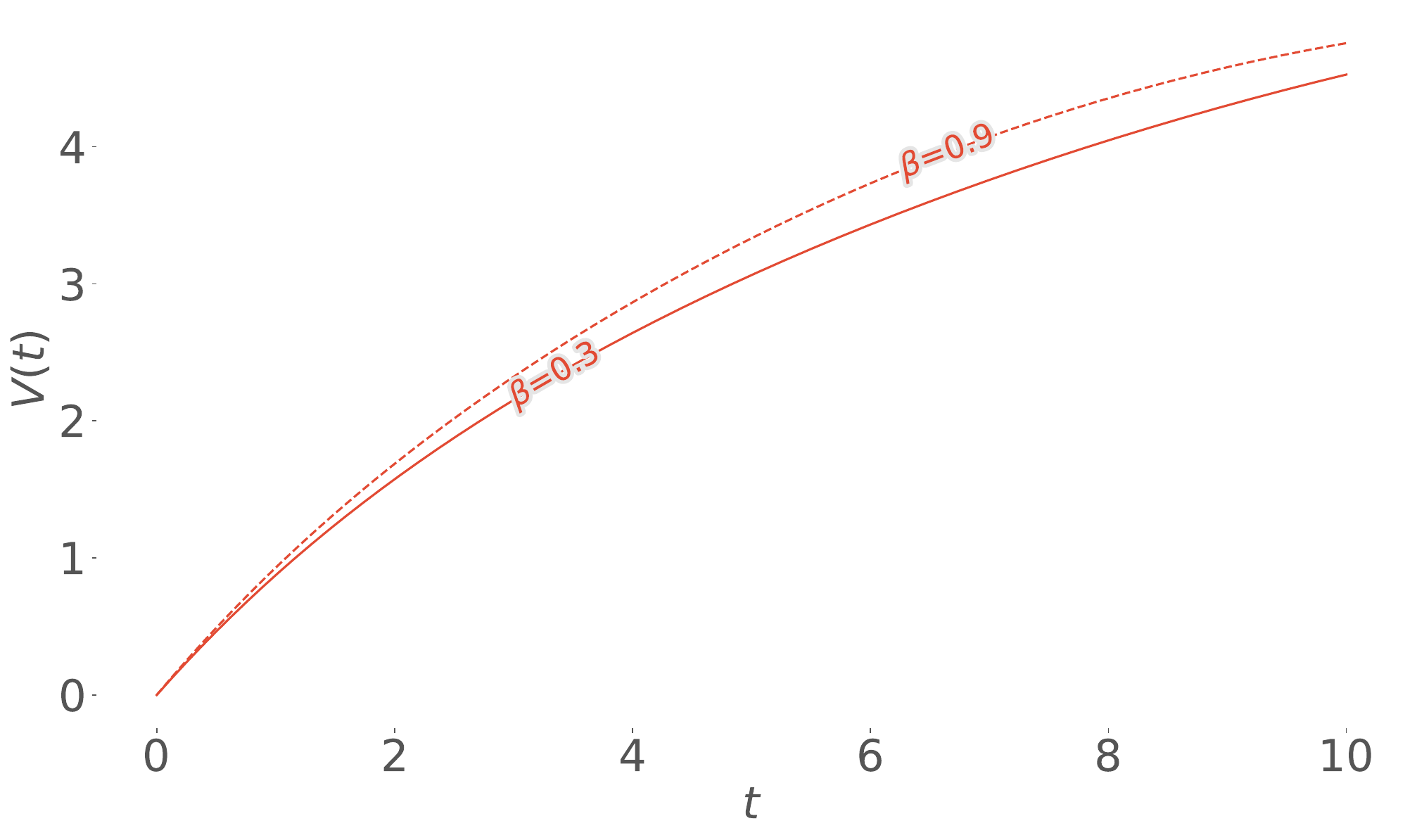}
			\caption{Cumulative revenue and volume with different values of competition parameter $\beta$. The solid curves represent a low level of competitive interaction ($\beta = 0.3$) and the dotted curves represent a high level of competitive interaction ($\beta = 0.9$). Other parameters are $T=10$, $\overline{Q}=5$, $\alpha=0.1$, $\kappa=1$, $\phi=0.03$, $A=1$, $\gamma=0.1$, and $\underline{B}=-10$.}
			\label{fig:DifferentBetaEcoFeatureSeller}
		\end{figure}
		
		In Figure \ref{fig:DifferentBetaEcoFeatureBuyer}, we can see that for this set of parameters, a higher value of competitiveness $\beta$ leads to lower cost, average cost and instantaneous average cost. This implies that consumers indeed benefit from a more competitive market with lower purchase prices. The instantaneous average cost $\overline{K}(t)$ shares a similar concave shape with the mean spread $\overline{\delta}$ in Figure \ref{fig:DifferentBetaMeanQuote&P_T}. This means that the best purchasing time for consumers is either the start or the end of the time horizon, consistent with the early bird price and closing sale in reality. Additionally, for the larger value of $\beta$, the maximum of the curves $\overline{K}(t)$ and $\overline{\delta}$ occur at earlier times. This means that the vendors will decrease price earlier in a more competitive market, and consumers indeed enjoy this lower market price earlier. In Figure \ref{fig:DifferentBetaEcoFeatureSeller}, we can see that for this set of parameters, a higher value of competitiveness $\beta$ leads to a higher volume, but lower revenue near the terminal time. One interesting phenomenon is that, before the terminal time, we enter into a win-win situation for both agents and consumers with a higher value of competitiveness $\beta$. We observe that for this time period, with higher demand of the market, a higher level of competitive interaction can boost the volume and cumulative revenue for agents and decrease the transaction cost for consumers simultaneously.
		
		\subsection{Effect of Upper Bound on Equilibrium}
		
		In this section, we impose an upper bound on admissible strategies. In reality, this price control is sometimes enforced by the government on the necessities, and we discuss how the behaviour of optimal strategy $\delta^*$, the mean spread $\overline{\delta}$, and the inventory distribution $P$ changes under an additional price ceiling.
		
		Suppose we define a new set of admissible controls by
		\begin{align*}
			\mathcal{A} = \left\{ \delta\, \Big\vert\, \delta_t = f(t,Q_t^{\delta,\overline{\delta}})\,, \quad \underline{B}\le f \leq \overline{B}, \quad f \mbox{ is continuous in } t \right\}\,,
		\end{align*}
		where $\underline{B}$ is a negative constant and $\overline{B}$ is a positive constant. Then the optimal feedback control becomes
		\begin{align*}
			\delta^*(t,q;\overline{\delta}) =\min\left\{\overline{B}\,,\,\, \max\left\{  \frac{1}{\kappa+\beta}+h_q(t;\overline{\delta})-h_{q-1}(t;\overline{\delta})\,,\,\, \underline{B}\right\}\right\}\,, \quad q\neq0\,.
		\end{align*}
		The proof of uniqueness of the solution to HJB equation and verification theorem still stands with minor adjustments.
		
		\begin{figure}[!htp]
			\centering
			\includegraphics[width=0.7\textwidth]{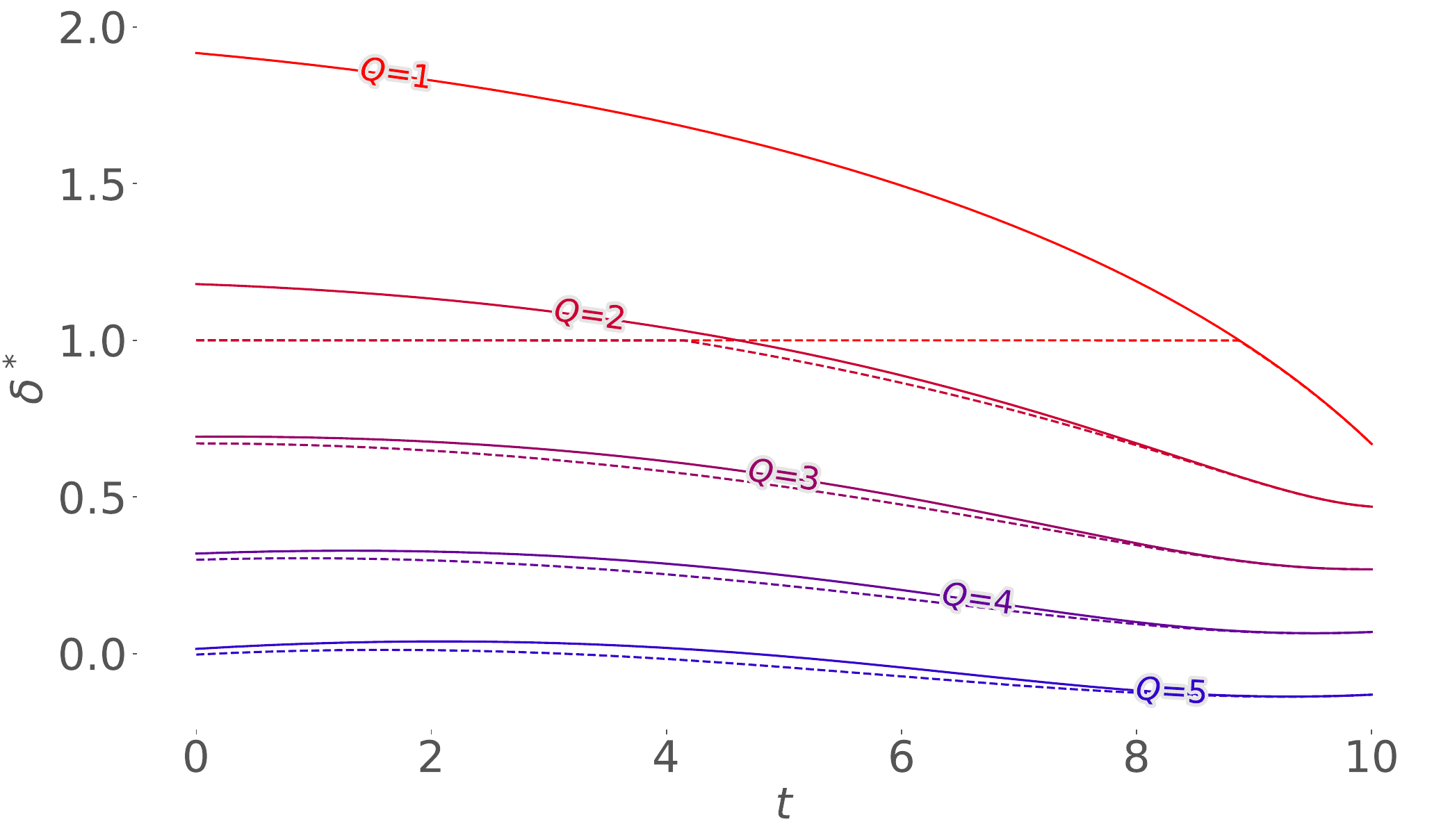}
			\caption{Optimal spreads in equilibrium with (dotted curves) and without (solid curves) an upper bound. Other parameters are $T=10$, $\overline{Q}=5$, $\alpha=0.1$, $\kappa=1$, $\phi=0.03$, $A=1$, $\beta = 0.3$, $\gamma=0.1$, $\underline{B}=-10$, and $\overline{B}=1$.}
			\label{fig:UpperboundMultiQuote}
		\end{figure}
		
		In Figure \ref{fig:UpperboundMultiQuote}, we can observe that optimal spreads with the upper bound decrease for all inventory levels. This is consistent with our expectation that the competition becomes more intense with an upper bound on price, as the agents with initially lower prices will also try to keep their market share while others are lowering their prices. For lower inventory levels $Q\le2$, the differences are more significant. This is due to the fact that agents with lower holdings tend to propose a higher spread, and that is where the restriction on pricing comes in. The optimal spreads converge at the terminal time, because the agents lower their price as much as possible for quick sales, at the level lower than the constraint.
		
		\begin{figure}[!htp]
			\centering
			\includegraphics[width=0.48\textwidth]{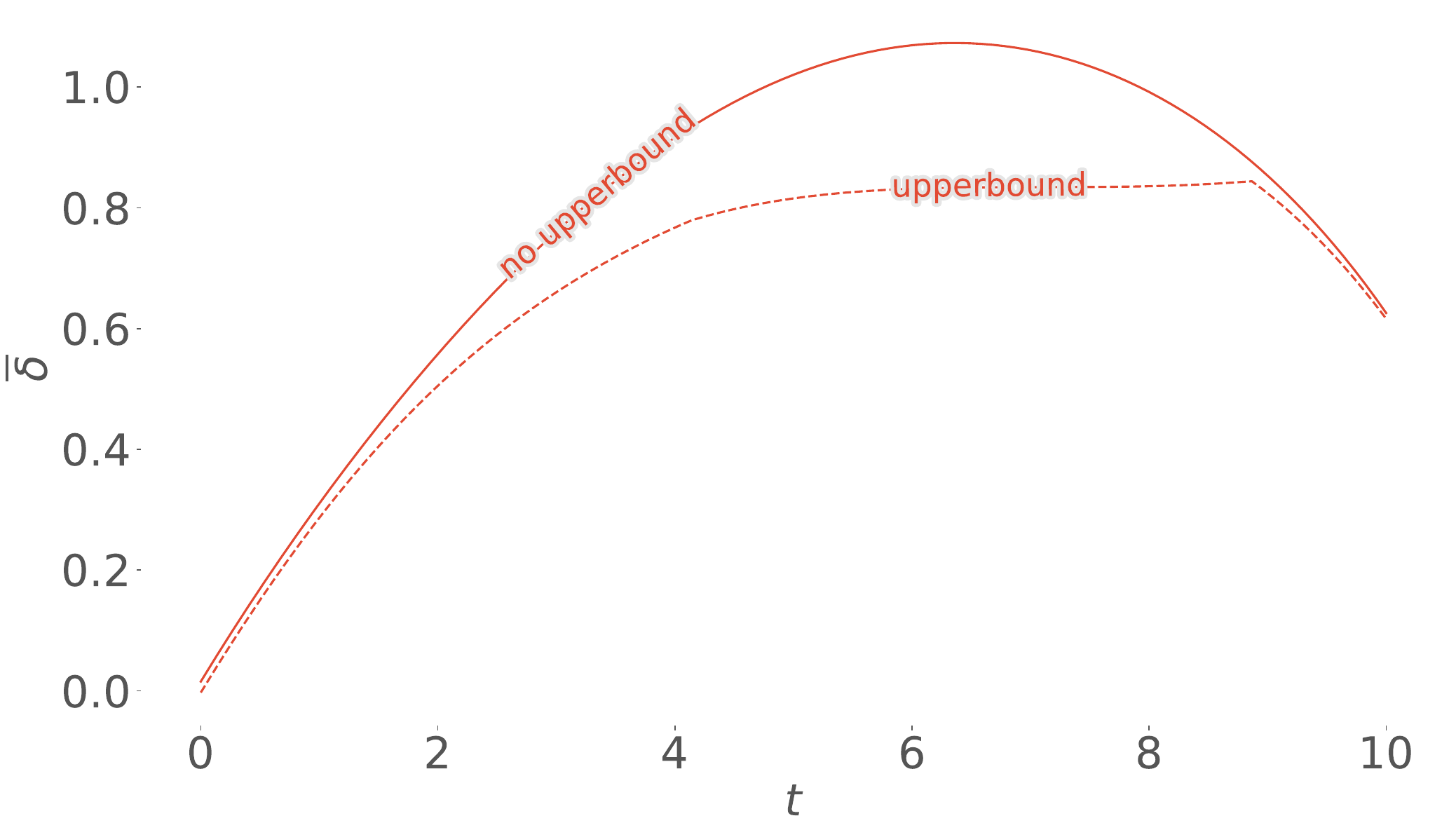}
			\includegraphics[width=0.48\textwidth]{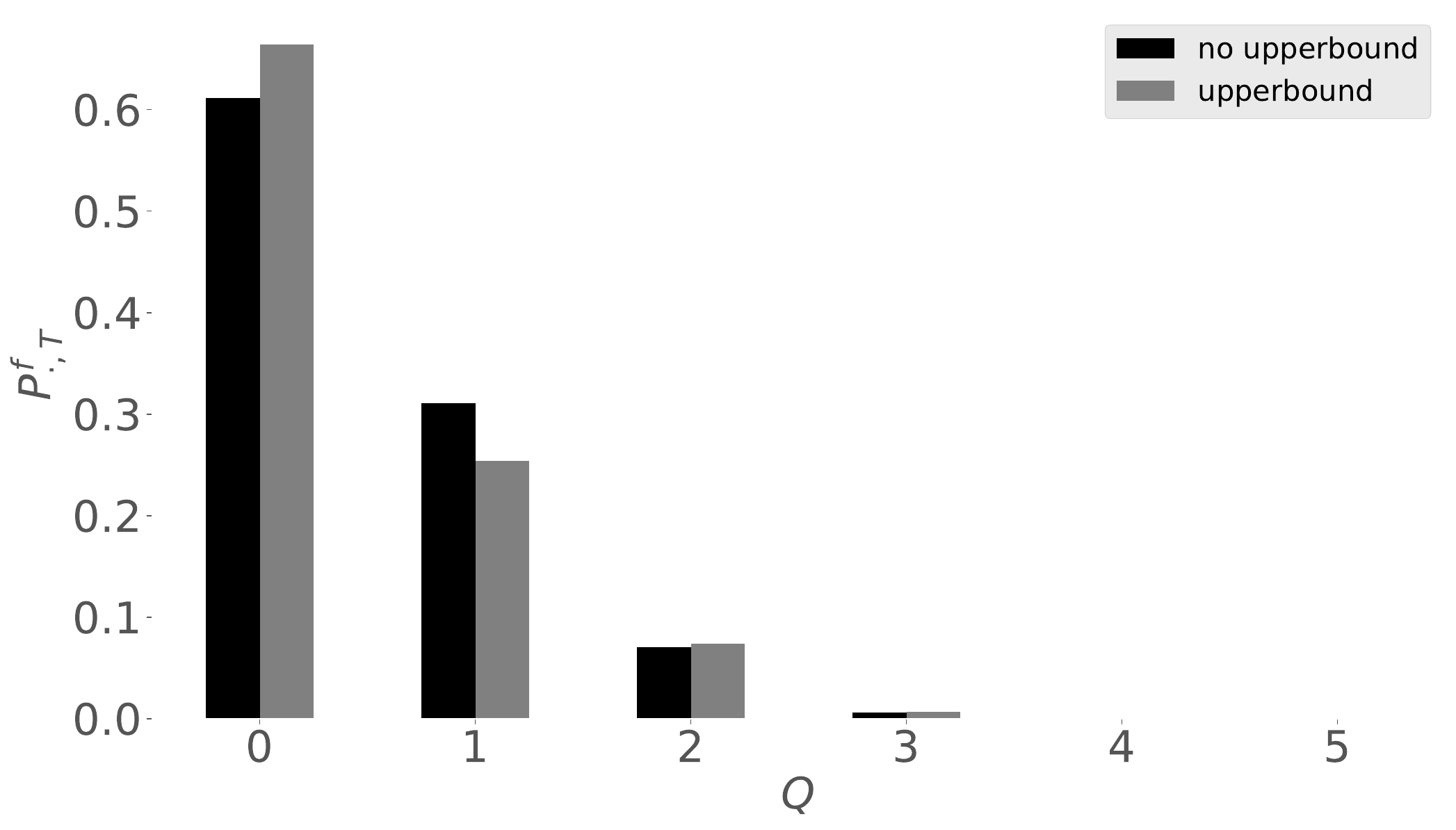}
			\caption{Mean spread and the distribution of inventory across agents at terminal time in equilibrium with and without an upper bound. Parameters are $T=10$, $\overline{Q}=5$, $\alpha=0.1$, $\kappa=1$, $\phi=0.03$, $A=1$, $\beta = 0.3$, $\gamma=0.1$, $\underline{B}=-10$, and $\overline{B}=1$.}
			\label{fig:UpperboundMeanQuote&P_T}
		\end{figure}
		
		In Figure \ref{fig:UpperboundMeanQuote&P_T}, the mean spread with the upper bound is always lower. At the beginning, the two lines are very close but as time goes by, more agents come to lower inventory levels and receive price restrictions, thus resulting in bigger differences between the mean spreads. Closer to the terminal time $T$, the two lines approach each other again. This is because the distribution of inventory levels and the corresponding optimal spreads are similar near terminal time in these two scenarios, as shown in Figure \ref{fig:UpperboundMultiQuote} and Figure \ref{fig:UpperboundMeanQuote&P_T}.
		
		\begin{figure}[h!]
			\centering
			\includegraphics[width=0.48\textwidth]{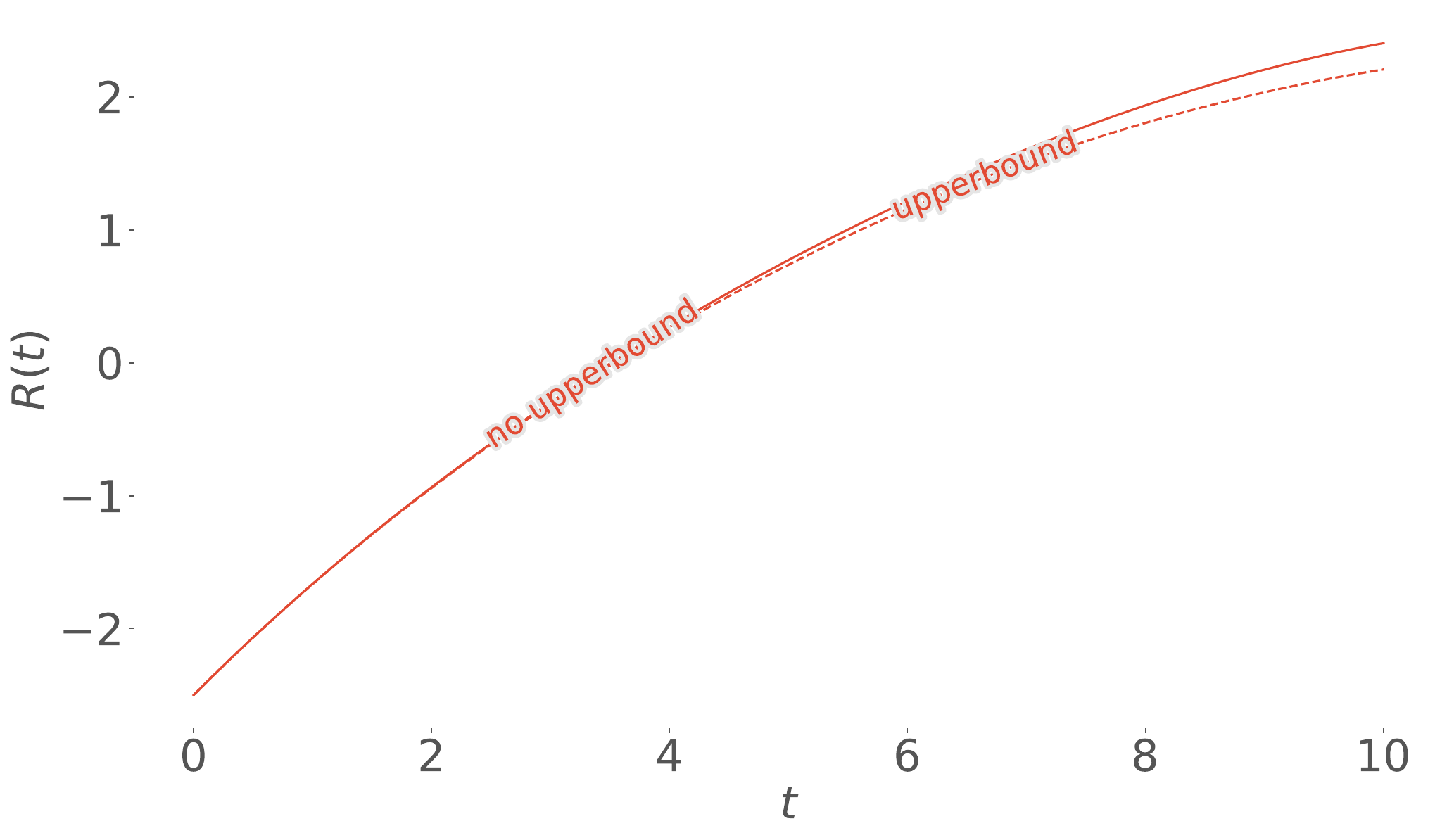}
			\includegraphics[width=0.48\textwidth]{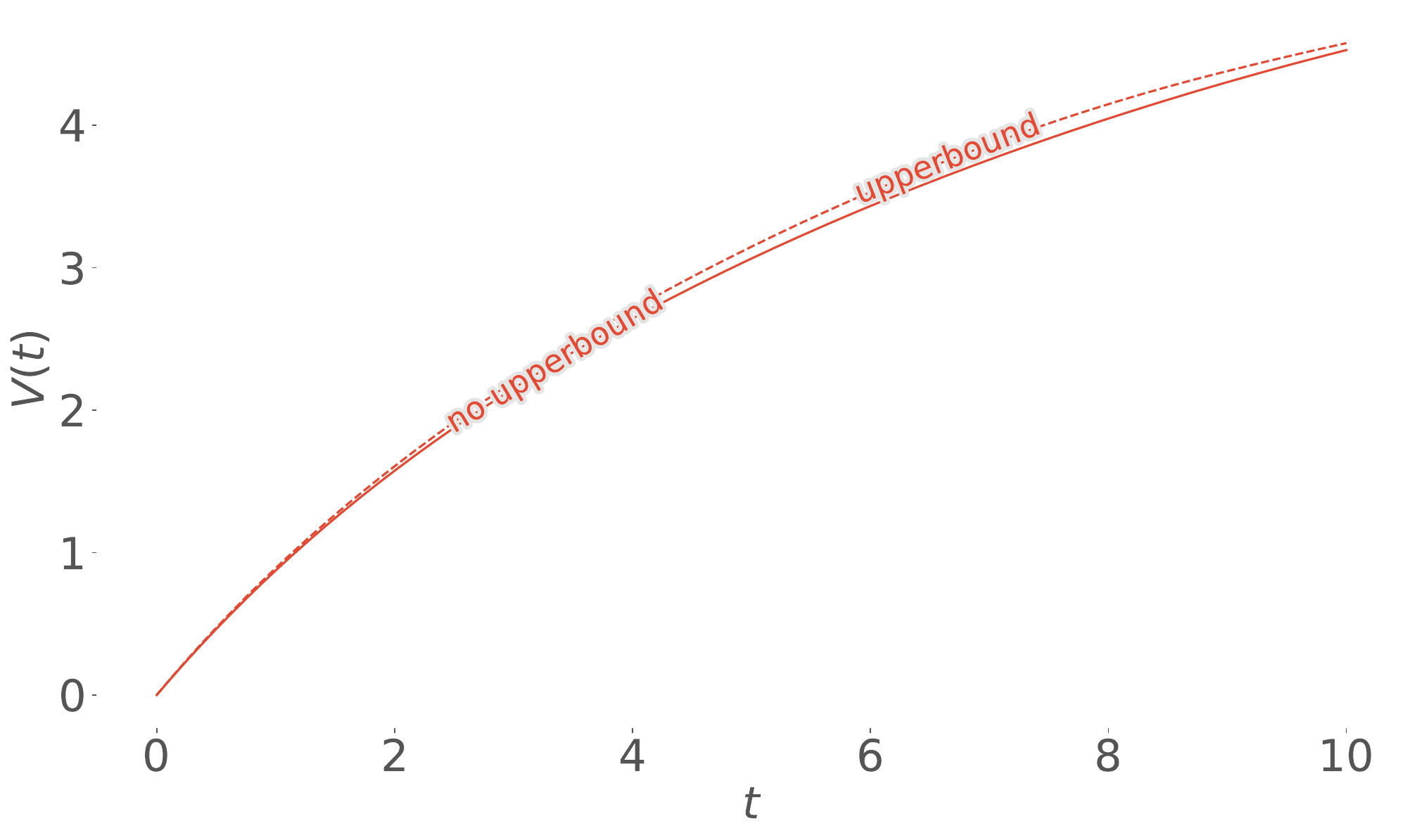}
			\includegraphics[width=0.48\textwidth]{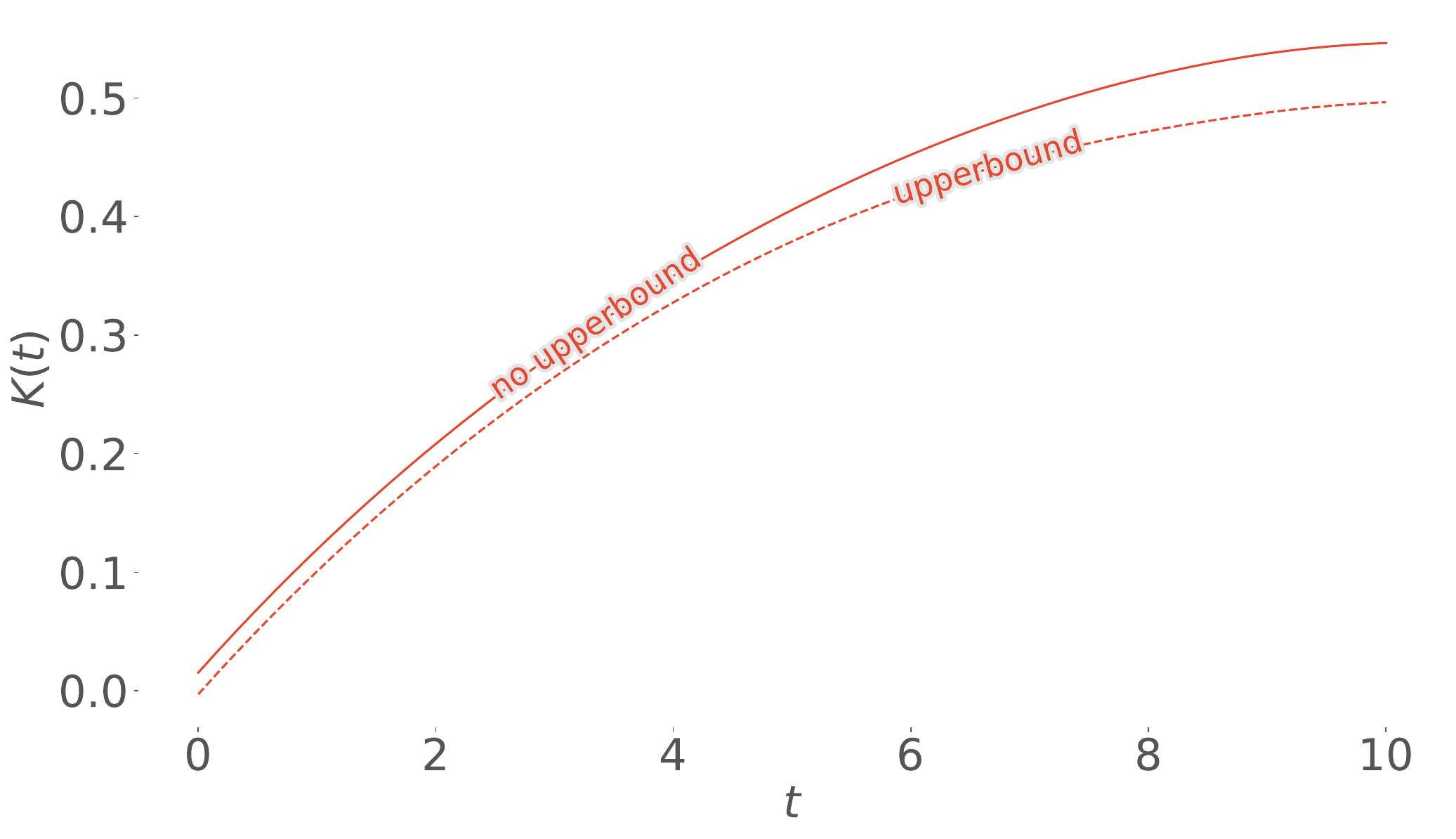}
			\includegraphics[width=0.48\textwidth]{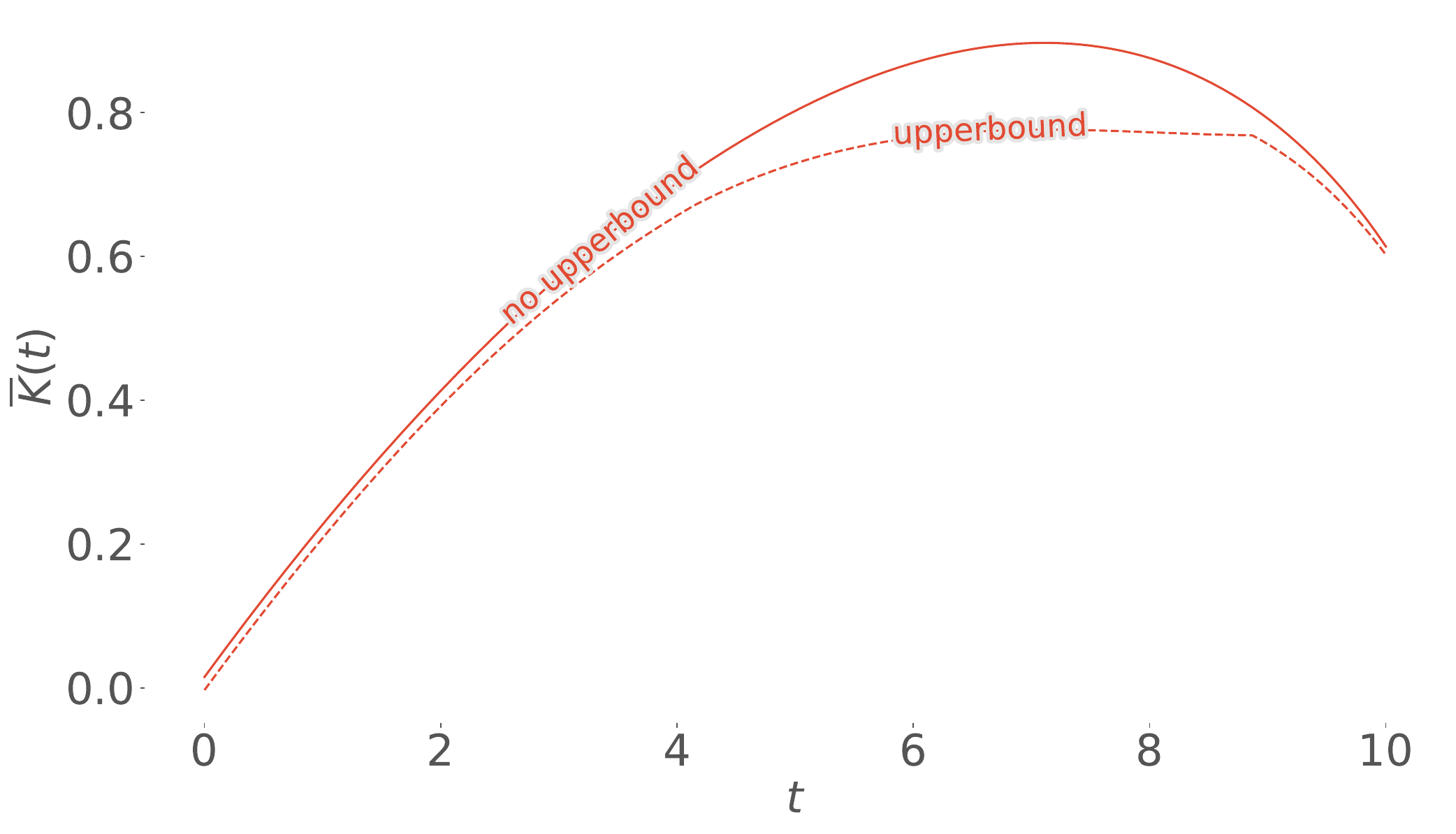}
			\caption{Cumulative revenue (top left), volume (top right), average transaction cost (bottom left) and instantaneous average cost (bottom right) with (dotted curves) and without (solid curves) an upper bound. Parameters are $T=10$, $\overline{Q}=5$, $\alpha=0.1$, $\kappa=1$, $\phi=0.03$, $A=1$, $\beta = 0.3$, $\gamma=0.1$, $\underline{B}=-10$, and $\overline{B}=1$.}
			\label{fig:UpperboundEcoFeature}
		\end{figure}
		
		Figures \ref{fig:UpperboundEcoFeature} shows that the price limit boosts the sales volume but reduces revenue for the agents. We can see that consumers buy products at a lower price with price ceiling, which means the restriction on high prices is effective in reducing consumer expenditure under our setting\footnote{In reality, price ceiling does not always work out as intended. If it leads to a severe imbalance between supply and demand, this can in turn cause shortages and underground markets.}.
		
		\subsection{Effect of Initial Inventory on Equilibrium}
		
		In the left panel of Figure \ref{fig:BiggerQOptimalQuote} we plot the feedback strategy in equilibrium for two different values of initial inventory, $\overline{Q}=5$ and $\overline{Q}=20$. The right panel of this figure shows the resulting mean spread for these initial values as well as for $\overline{Q}=50$ and $\overline{Q} = 110$. In the left panel, it is important to note that the curves corresponding to inventory levels $Q = 1$ to $Q = 5$ are different depending on what the initial inventory is. Specifically, each curve with initial inventory $\overline{Q} = 20$ is below the corresponding curve for $\overline{Q}=5$ (except at time $T$ where they coincide), even though they correspond to a pricing strategy for the same amount of remaining inventory. This is explained by the fact that average spreads as shown in the right panel are lower for large values of starting inventory. If everyone else starts with $\overline{Q}=20$, then the sell intensity of a representative agent is lower compared to if everyone else starts with $\overline{Q}=5$, all else being equal. To compensate for this decreased sell intensity, the representative agent lowers their own price.
		
		\begin{figure}[!htp]
			\centering
			\includegraphics[width=0.48\textwidth]{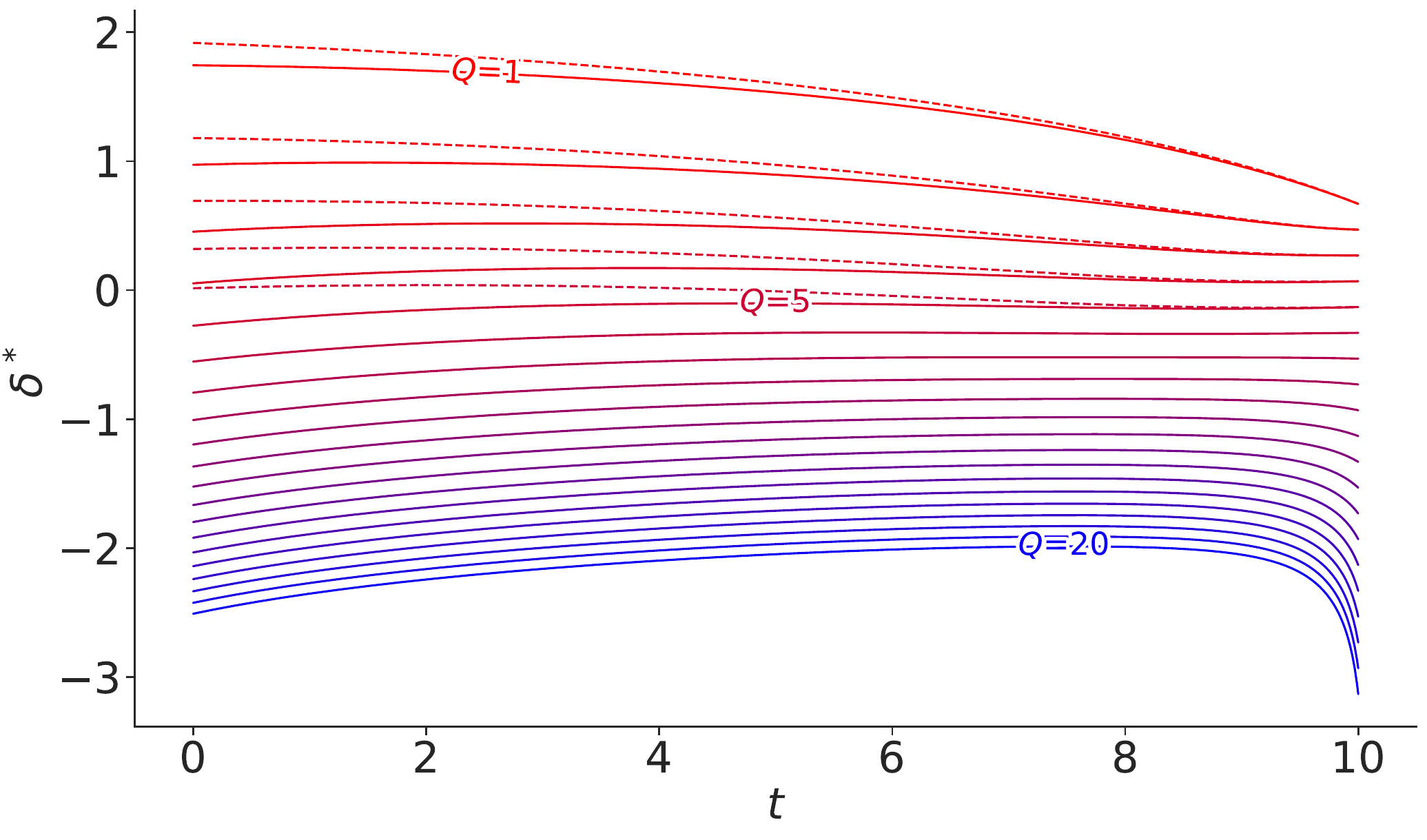}
			\includegraphics[width=0.48\textwidth]{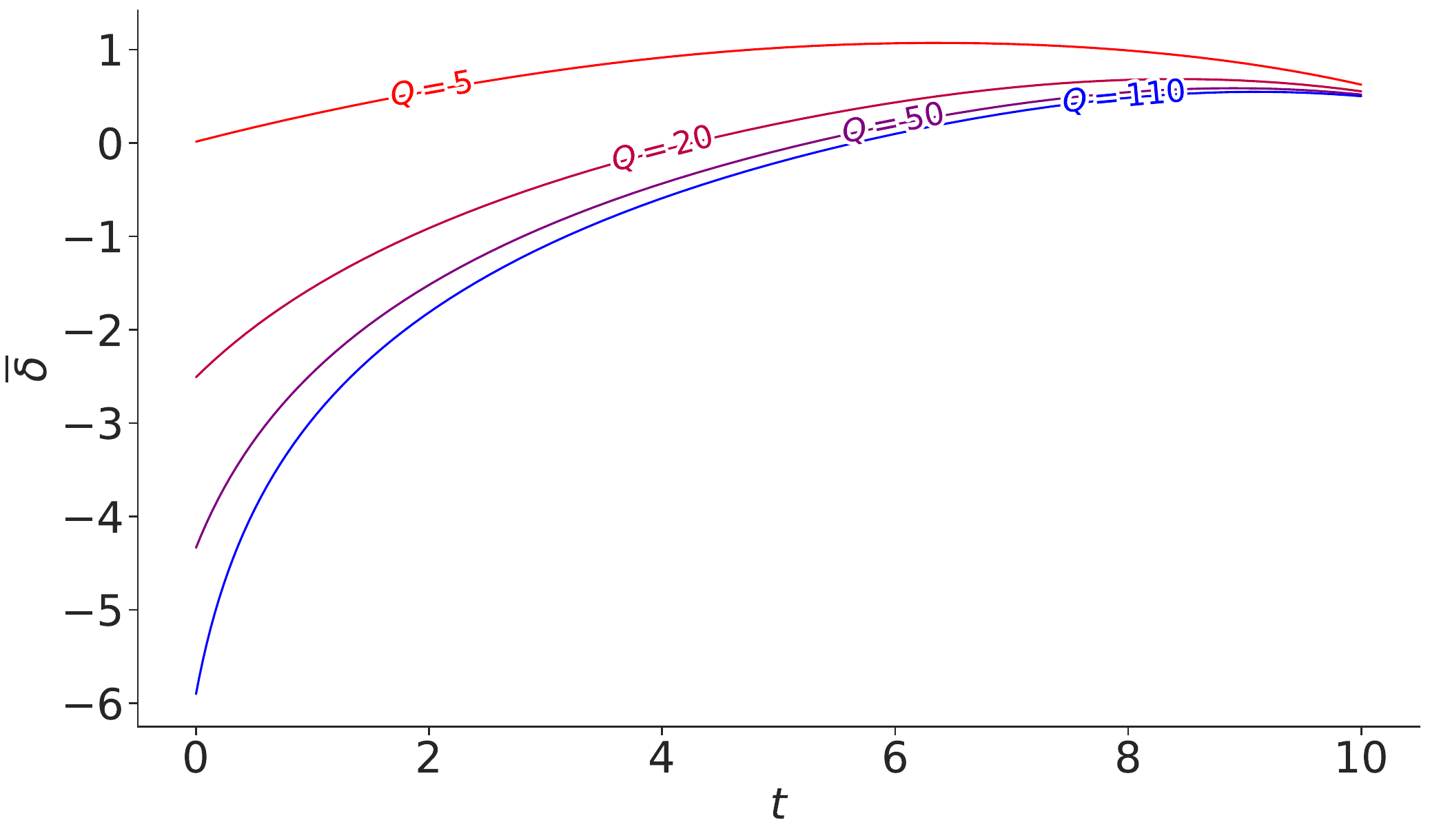}
			\caption{Left panel: Optimal spreads in mean-field equilibrium when $\overline{Q}=20$ (solid curves) compared to when $\overline{Q}=5$ (dotted curves). Right panel: Mean spread in equilibrium for various initial inventory levels. Other parameters are $T=10$, $\alpha=0.1$, $\kappa=1$, $\phi=0.03$, $A=1$, $\beta = 0.3$, $\gamma=0.1$, and $\underline{B}=-10$.}
			\label{fig:BiggerQOptimalQuote}
		\end{figure}

		\subsection{Effect of Overselling on Equilibrium}
		
		In this section, we study the effects of overselling where agents are allowed to sell more volume of inventory than they have physically available during the trading period, and they must pay extra cost to clear out the negative inventory at terminal time. This phenomenon frequently happens in the tourism industry. For example, airlines oversell air tickets anticipating some cancellations or no-shows of the customer bookings. If everyone shows up, airlines may refuse boarding to certain passengers but must pay compensation depending on jusrisdiction.
		
		We set the admissible inventory level set as $\{\underline{Q}, \underline{Q}+1,\cdots, \overline{Q}\}$, and $\underline{Q}$ is a non-positive constant. The new model under overselling setting is the same as the one in Section \ref{sec:competition_model}, except for the lower bound of inventory and the new value function
		\begin{align*}
			H(t,s,x,q;\overline{\delta})=\sup_{(\delta_u)_{t\le u\le T}\in\mathcal{A}} \mathbb{E}_{t,s,x,q} \biggl[ &X_T^{\delta,\overline{\delta}} + Q_T^{\delta,\overline{\delta}}\,S_T-\left(\alpha_1\,\mathcal{X}_{Q_T^{\delta,\overline{\delta}}\ge 0}+\alpha_2\,\mathcal{X}_{Q_T^{\delta,\overline{\delta}}< 0}\right)\left(Q_T^{\delta,\overline{\delta}}\right)^2\\
			&-\int^T_t\left(\phi_1\, \mathcal{X}_{Q_u^{\delta,\overline{\delta}}\ge 0}+\phi_2\,\mathcal{X}_{Q_u^{\delta,\overline{\delta}}< 0}\right)\left(Q_u^{\delta,\overline{\delta}}\right)^2 du \biggl]\,,
		\end{align*}
		where $\alpha_1<\alpha_2$ and $\phi_1<\phi_2$ are all positive constants. Compared to equation \eqref{eqn:ControlProblem}, we set different parameters $\alpha_1$ and $\alpha_2$ of terminal penalty for positive and negative inventory respectively. Here, the difference $\left(\alpha_2-\alpha_1\right)(Q_{T}^{\delta,\overline{\delta}})^2$ represents the extra expected compensation paid by agents in the case that consumers make a claim on the oversold units. Since overselling brings more inventory risk, we set parameters of running inventory penalty $\phi_1<\phi_2$ to discourage this behavior. Then the associated HJB equation is given by
		\begin{align*}
			\begin{split}
				&\partial_t H + \frac{1}{2}\,\sigma^2\,\partial_{ss}H-\left(\phi_1\mathcal{X}_{q\ge0}+\phi_2\mathcal{X}_{q<0}\right)\, q^2\\
				&\phantom{\partial_t H }+\sup_{\delta\ge \underline{B}}\lambda(\delta, \overline{\delta}) \left[ H(t,s,x+s+\delta,q-1;\overline{\delta})-H(t,s,x,q;\overline{\delta})\right]\mathcal{X}_{q>\underline{Q}} = 0\,,\\
				&H(T,s,x,q;\overline{\delta}) = x+q\,s-\left(\alpha_1\mathcal{X}_{q\ge0}+\alpha_2\mathcal{X}_{q<0}\right)\,q^2\,.
			\end{split}
		\end{align*}
		We use the same ansatz in Proposation \ref{prop:SolutiontoHJBEquation} to solve the above equations, and get the following ODEs
		\begin{align*}
			&\partial_t h_q - \left(\phi_1\mathcal{X}_{q\ge0}+\phi_2\mathcal{X}_{q<0}\right)\, q^2\\
			&\phantom{\partial_t h_q}+\sup_{\delta\ge \underline{B}} A\,\exp{\left\{-(\kappa+\beta)\,\delta+\beta\, \overline{\delta}\right\}}\left[\delta+h_{q-1}(t;\overline{\delta})-h_q(t;\overline{\delta}) \right]\mathcal{X}_{q>\underline{Q}} = 0\,,\\
			&h_q(T;\overline{\delta}) = -\left(\alpha_1\mathcal{X}_{q\ge0}+\alpha_2\mathcal{X}_{q<0}\right)\,q^2\,.
		\end{align*}
		As overselling does not affect the $\sup$ term in the ODEs, the feedback form of the optimal control is still the same as in equation \eqref{eqn:OptimalFeedbackControls}. We can show that the above ODEs of $h$ still has a unique solution under overselling setting, and the proof of verification theorem still stands with minor adjustments.
		
		\begin{figure}[h!]
			\centering
			\includegraphics[width=0.48\textwidth]{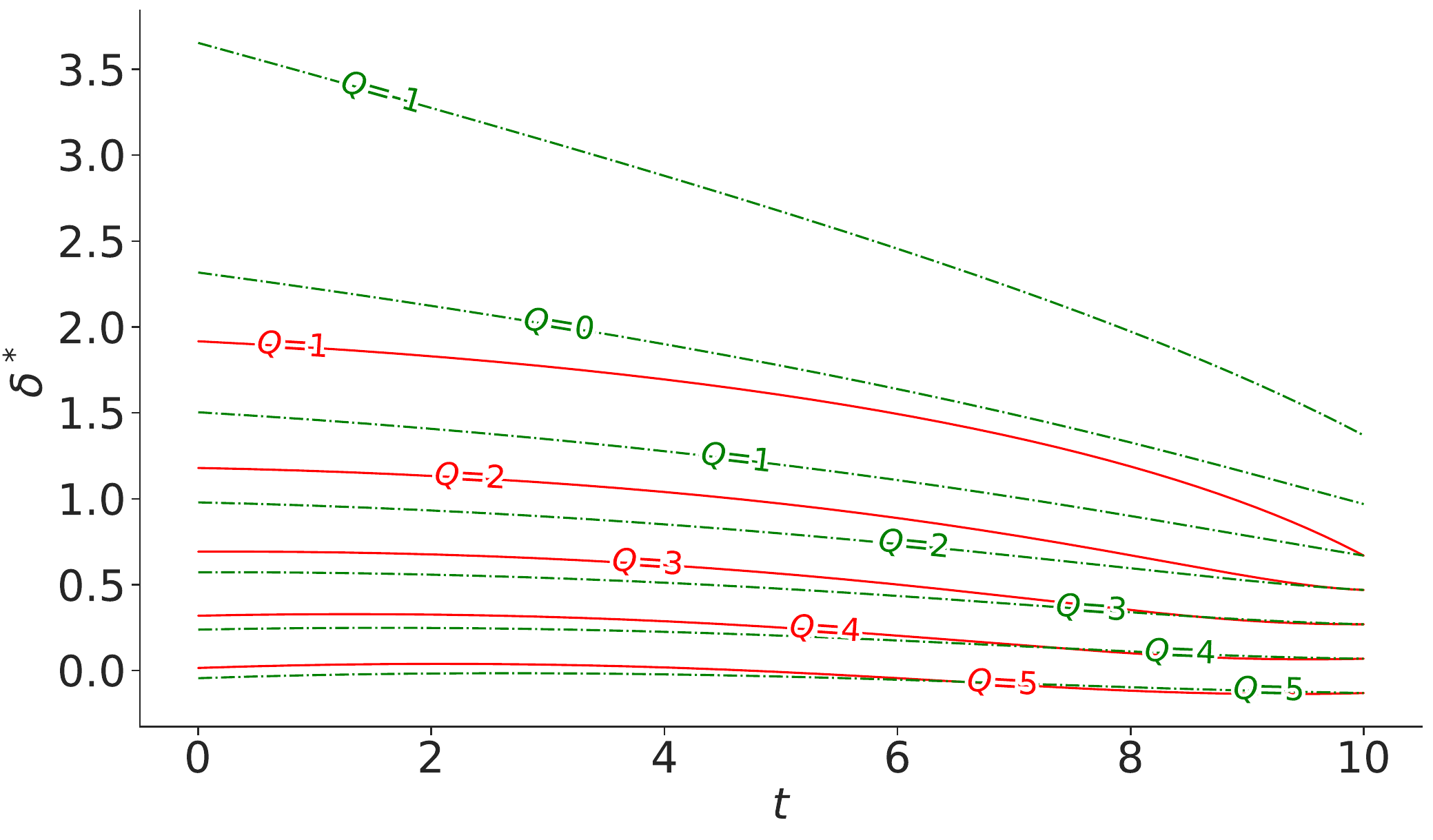}
			\includegraphics[width=0.48\textwidth]{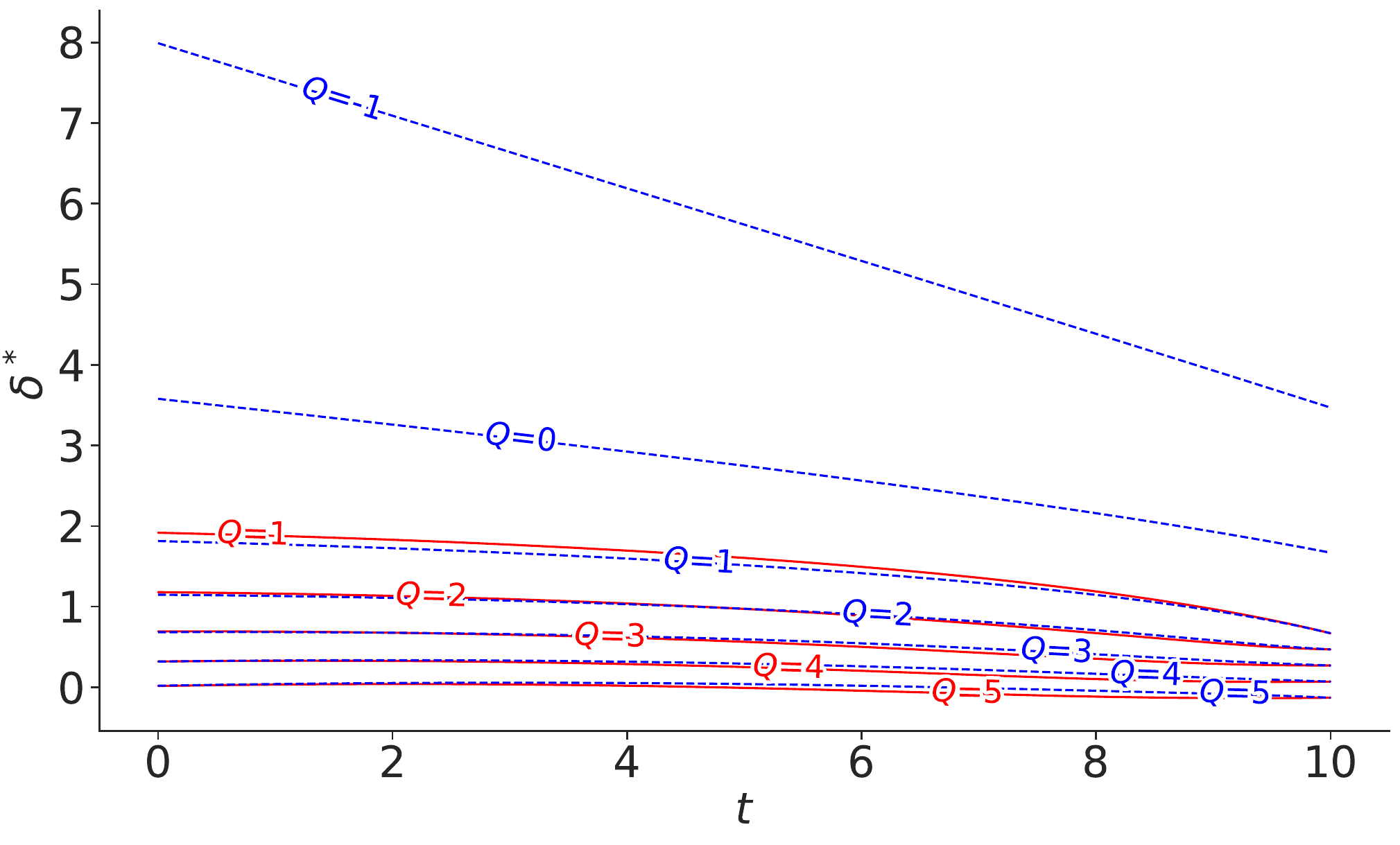}
			\caption{Optimal spreads in equilibrium with overselling (dotted curves) and without overselling (solid curve). The dotted curve on the left figure is with lower $\alpha_2=0.2$, $\phi_2=0.06$, and the one on the right figure is with higher $\alpha_2=0.9$, $\phi_2=0.15$. Other parameters are $T=10$, $\overline{Q}=5$, $\underline{Q}=-2$, $\alpha_1=0.1$, $\kappa=1$, $\phi_1=0.03$, $A=1$, $\beta = 0.3$, $\gamma=0.1$, $\underline{B}= -10$, and $\overline{B}=20$.}
			\label{fig:OversellMultiQuote}
		\end{figure}
		
		Figure \ref{fig:OversellMultiQuote} shows the change in optimal spreads when the ability to oversell is introduced, and it also indicates how the size of terminal penalty parameter $\alpha_2$ and running inventory penalty parameter $\phi_2$ for negative inventory affects the change in spreads. As expected, the optimal spreads under overselling setting are still decreasing in inventory level monotonically. On the left of Figure \ref{fig:OversellMultiQuote}, with smaller $\alpha_2$ and $\phi_2$, we can see that for all positive inventory levels, the optimal spreads with overselling are generally lower than the ones without overselling. On the right panel of Figure \ref{fig:OversellMultiQuote}, with larger values of $\alpha_2$ and $\phi_2$, the spreads for positive inventory still generally decrease when overselling is allowed, but the effect of this feature is less pronounced. This indicates that higher value of $\alpha_2$ and $\phi_2$ leads to higher optimal spreads across all inventory levels with overselling. Additionally, with stronger penalties for overselling, prices at early times are similar to prices when overselling is prohibited, but over time as more agents reach zero and negative inventory levels, the quoted spreads become very large because these agents are only willing to sell at high prices to recoup the large overselling penalty. 
		
		\begin{figure}[h!]
			\centering
			\includegraphics[width=0.48\textwidth]{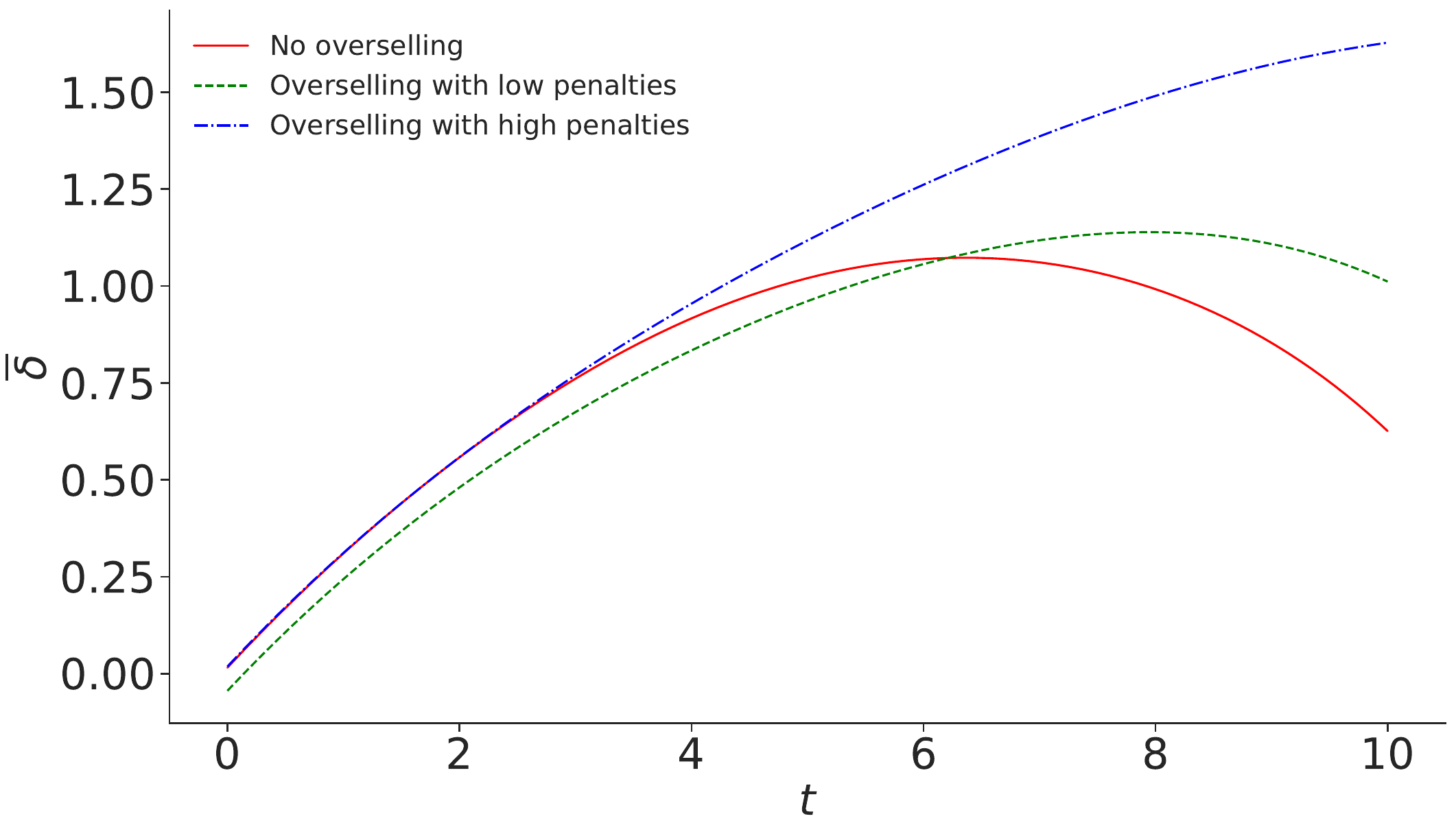}
			\includegraphics[width=0.48\textwidth]{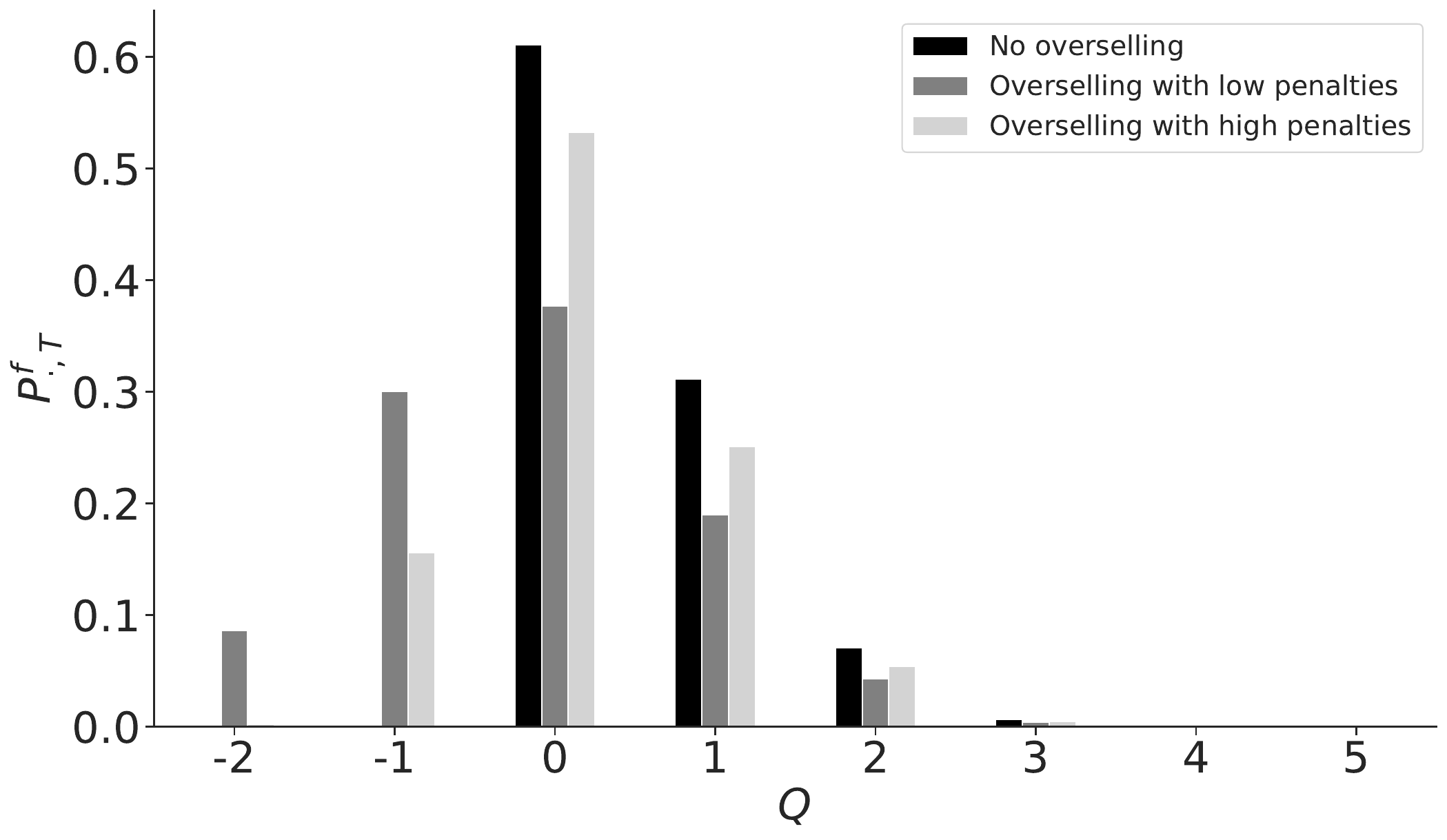}
			\caption{Mean spreads in equilibrium and the distribution of inventory at terminal time without overselling (solid curve) and with overselling (dotted and dash-dotted curves). The dotted curve with overselling is with lower $\alpha_2=0.2$, $\phi_2=0.06$, and the dash-dotted curve with overselling is with higher $\alpha_2=0.9$, $\phi_2=0.15$. Other parameters are $T=10$, $\overline{Q}=5$, $\underline{Q}=-2$, $\alpha_1=0.1$, $\kappa=1$, $\phi_1=0.03$, $A=1$, $\beta = 0.3$, $\gamma=0.1$, $\underline{B}=-10$, and $\overline{B}=20$.}
			\label{fig:OversellMeanQuote&P_T}
		\end{figure}
		
		The left of Figure \ref{fig:OversellMeanQuote&P_T} shows the change in mean spreads when the ability to oversell is introduced. The solid line stands for the mean spread without overselling, and the dotted and dash-dotted curves represent overselling for different values of $\alpha_2$ and $\phi_2$. The effect of overselling depends also heavily on the penalty parameters for negative inventory levels. The crossing of mean spread without overselling and mean spread with overselling with low penalty is consistent with the result about optimal spreads in the left of Figure \ref{fig:OversellMultiQuote}. In the beginning, agents who can oversell quote lower prices, thus resulting in a lower mean spread. As time goes by, the optimal spreads with overselling get closer to the ones without overselling. In the meantime, more agents achieve negative inventory levels with much higher spreads, so the mean spread with overselling exceeds the one without overselling closer to terminal time. For the mean spread without overselling and mean spread with overselling with high penalty, we can see that if the punishment for overselling is high enough, the mean spread with overselling can always be higher than the one without overselling. The right of Figure \ref{fig:OversellMeanQuote&P_T} compares the distributions of inventory at terminal time. The distribution of inventory at time $T$ with overselling follows a bell-shaped curve. With overselling, most agents stay at zero inventory, and there are more agents with negative inventory than the ones with positive inventory. With overselling, we can observe that higher terminal penalty parameter $\alpha_2$ and running inventory penalty parameter $\phi_2$ for shorting indeed reduces the proportion of agents ending at negative inventory levels.
		
		\begin{figure}[h!]
			\centering
			\includegraphics[width=0.48\textwidth]{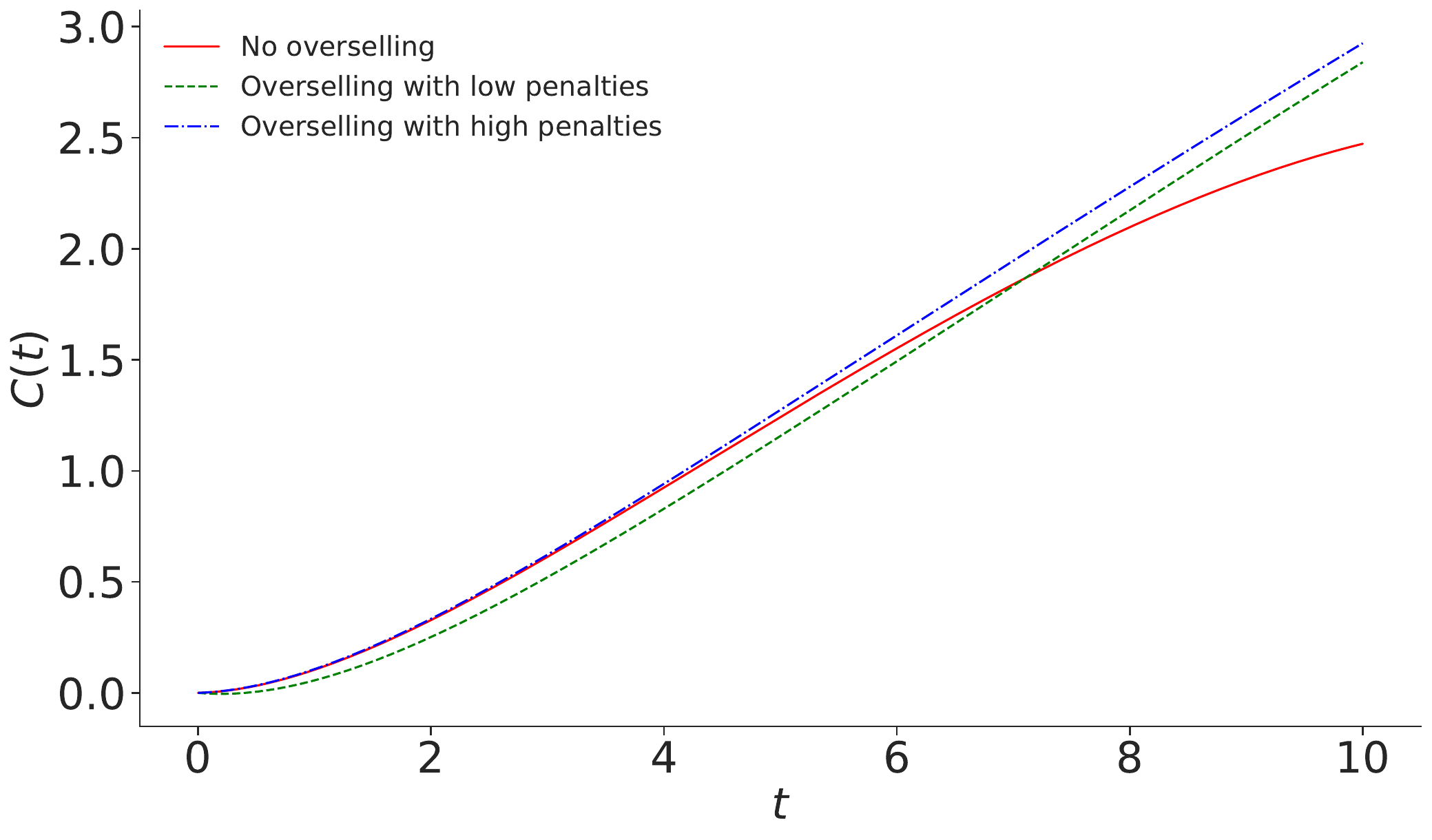}
			\includegraphics[width=0.48\textwidth]{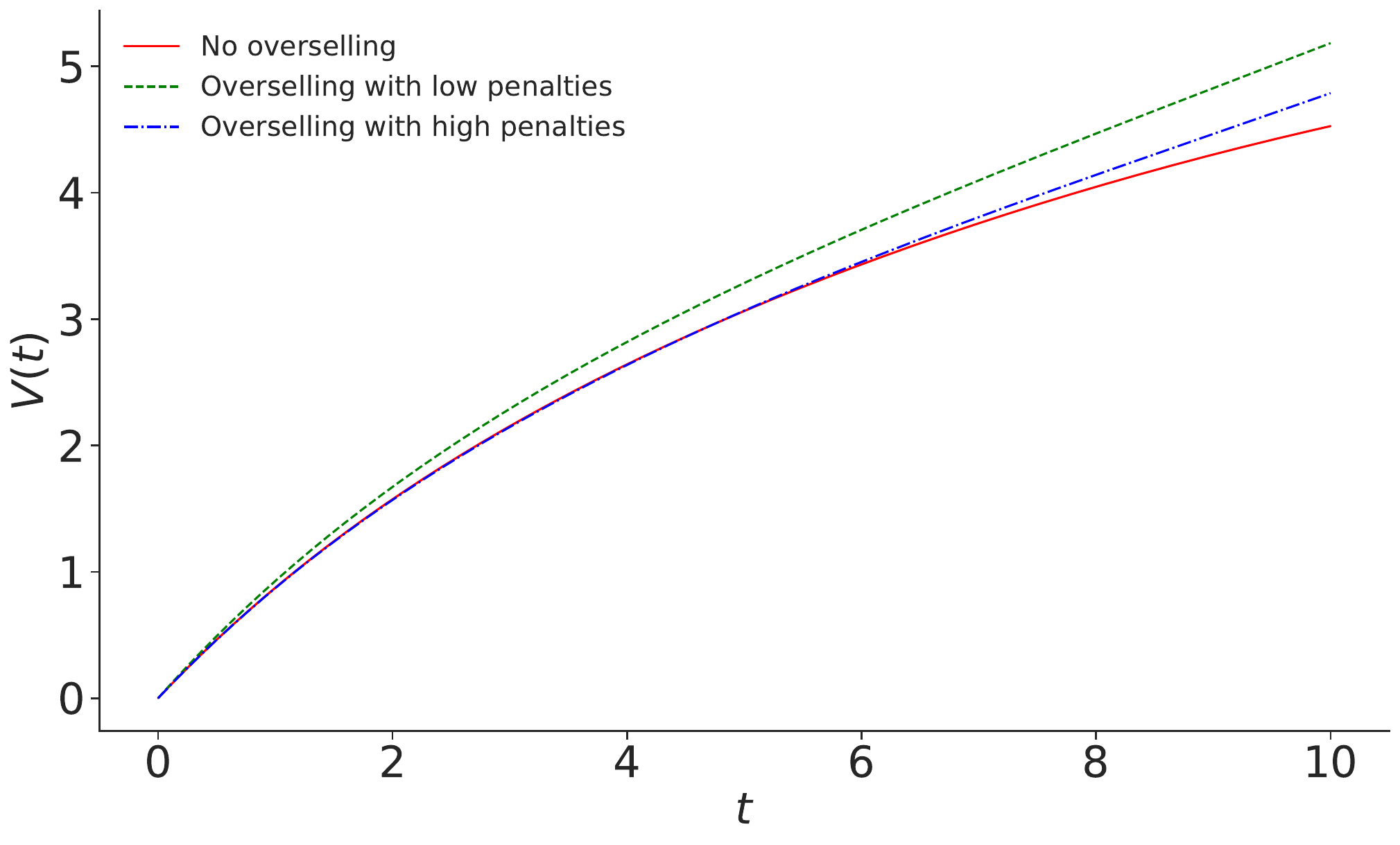}
			\includegraphics[width=0.48\textwidth]{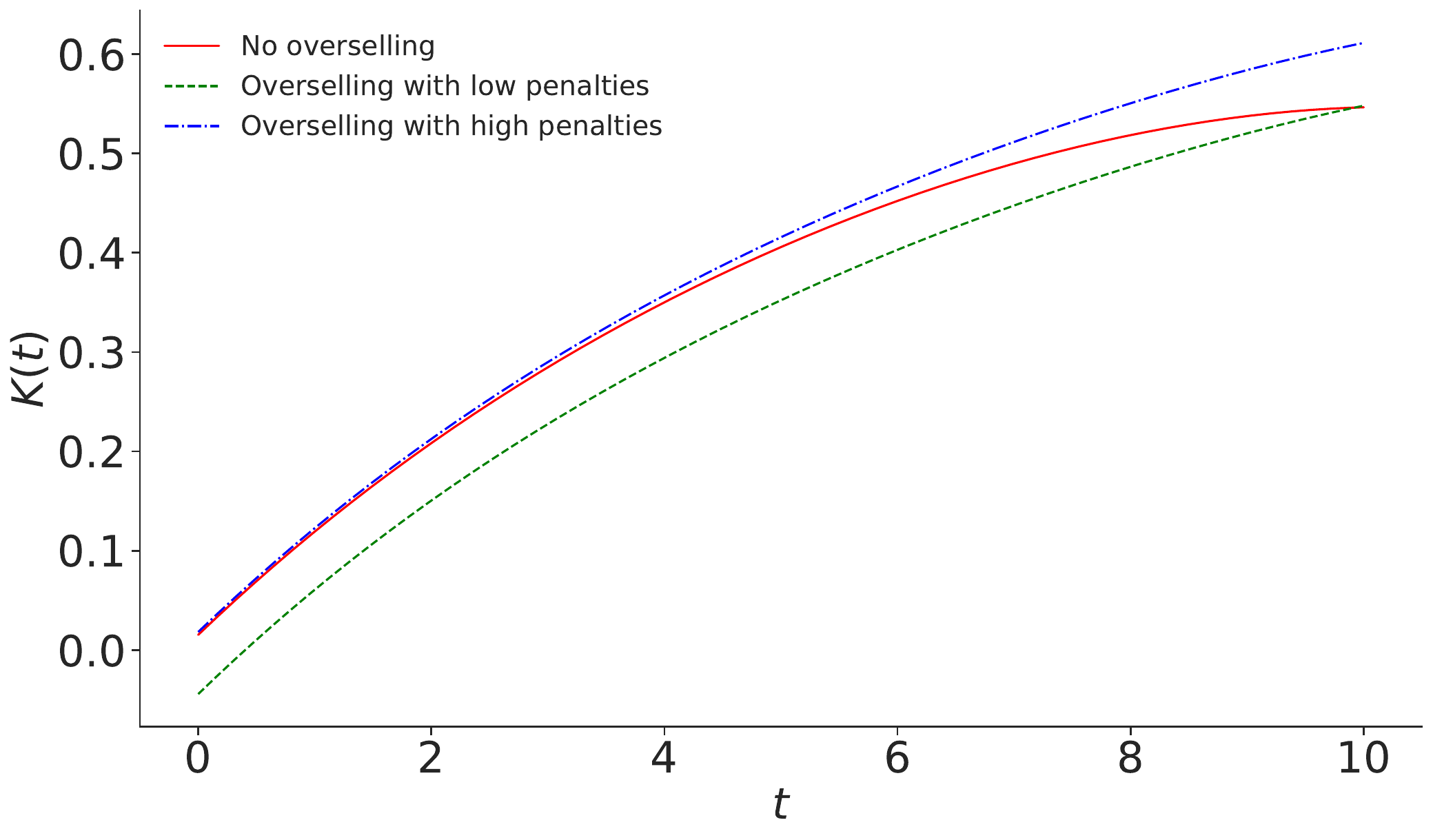}
			\includegraphics[width=0.48\textwidth]{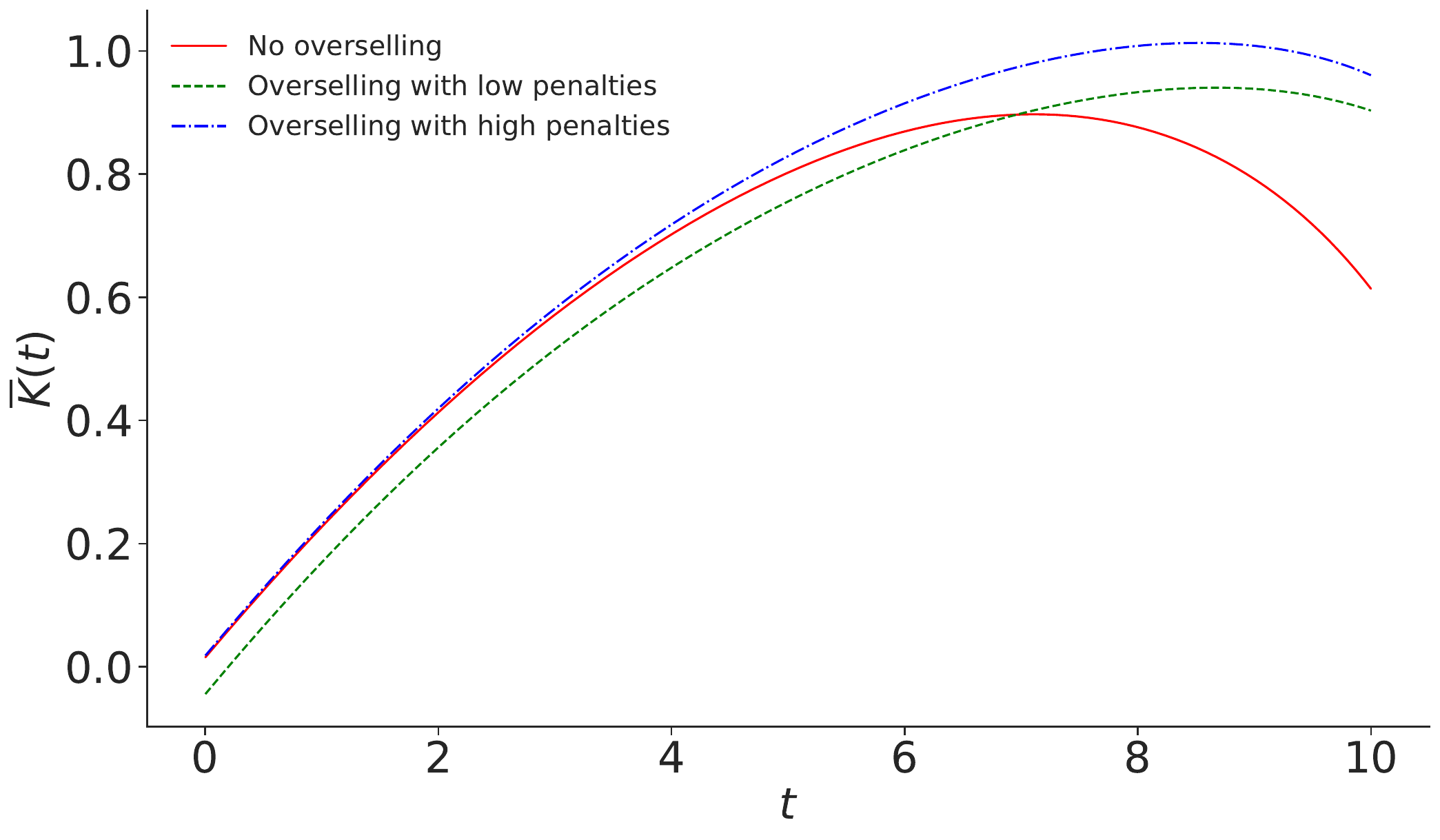}
			\caption{Cumulative revenue (top left), volume (top right), average transaction cost (bottom left) and instantaneous average cost (bottom right) without overselling (solid curve) and with overselling (dotted and dash-dotted curves). The dotted curve is with lower $\alpha_2=0.2$, $\phi_2=0.06$, and the dash-dotted curve is with higher $\alpha_2=0.9$, $\phi_2=0.15$. Other parameters are $T=10$, $\overline{Q}=5$, $\underline{Q}=-2$, $\alpha_1=0.1$, $\kappa=1$, $\phi_1=0.03$, $A=1$, $\beta = 0.3$, $\gamma=0.1$, $\underline{B}=-10$, and $\overline{B}=20$.}
			\label{fig:OversellEcoFeature}
		\end{figure}
		
		The cumulative cost $C(t)$, cumulative revenue $R(t)$, and cumulative volume $V(t)$ are slightly modified when overselling is introduced, and we write them as
		\begin{align*}
			C(t) &= \sum_{q=\underline{Q}+1}^{\overline{Q}}\int_0^{t}f(u,q)\,P^f_{q,u}\,\lambda(f(u,q),\overline{\delta}_u)\,du\,, \\
			R(t) &= \sum_{q=\underline{Q}+1}^{\overline{Q}}\left(\int_0^{t}f(u,q)\,P^f_{q,u}\,\lambda(f(u,q),\overline{\delta}_u)\,du-\alpha_1\,P^f_{q,t}\, q^2\,\mathcal{X}_{q>0}-\alpha_2\, P^f_{q,t}\, q^2\,\mathcal{X}_{q<0}\right)\,, \\
			V(t) &= \sum_{q=\underline{Q}+1}^{\overline{Q}}\int_0^{t}\,P^f_{q,u}\,\lambda(f(u,q),\overline{\delta}_u)\,du\,.
		\end{align*}
		Figure \ref{fig:OversellEcoFeature} shows that regardless of the level of penalty, cumulative cost near the terminal time with overselling is higher. From the top right panel we see that overselling leads to higher traded volume, with smaller penalties leading to higher total volume traded. In the bottom left panel, the average transaction cost without overselling lies between the ones with overselling. When overselling with high penalties, overselling leads to higher average transaction cost. However, with low penalty, consumers indeed pay less on average to buy more because of overselling. This indicates that with the appropriate levels of penalty parameters $\alpha_2$ and $\phi_2$ for overselling, agents can sell more products, while consumers pay almost the same average transaction cost at terminal time compared to the non-overselling case. In the bottom right panel, we can observe that the instantaneous average cost with overselling is always high close to the terminal time. This is due to the fact that the instantaneous fraction of oversold products is higher near the terminal time, thus raising up the unit price. It is generally more beneficial for consumers to purchase earlier to avoid higher price due to overselling.
		
		\begin{figure}[h!]
			\centering
			\includegraphics[width=0.7\textwidth]{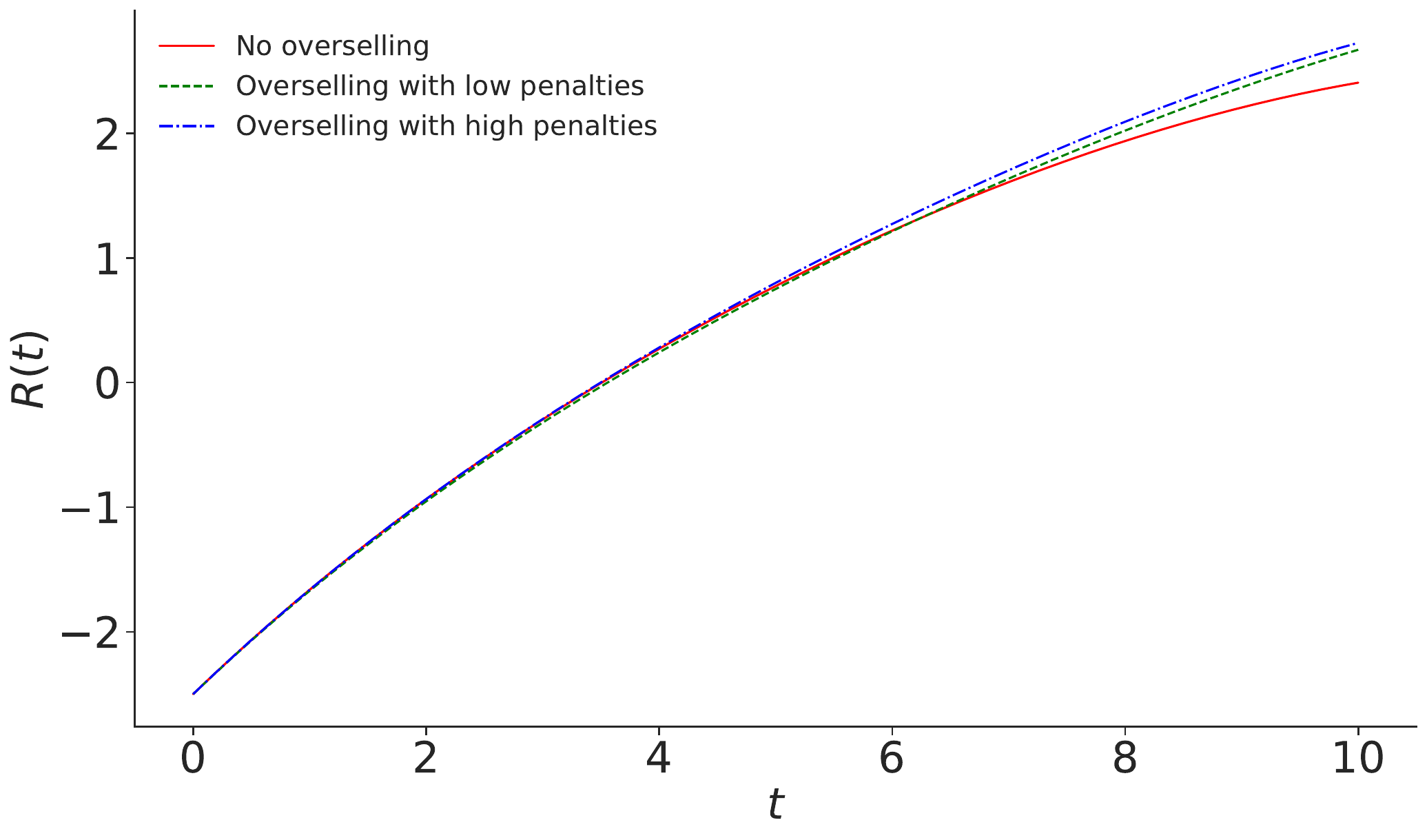}
			\caption{Cumulative revenue without overselling (solid curve) and with overselling (dotted and dash-dotted curves). The dotted curve is with lower $\alpha_2=0.2$, $\phi_2=0.06$, and the dash-dotted curve is with higher $\alpha_2=0.9$, $\phi_2=0.15$. Other parameters are $T=10$, $\overline{Q}=5$, $\underline{Q}=-2$, $\alpha_1=0.1$, $\kappa=1$, $\phi_1=0.03$, $A=1$, $\beta = 0.3$, $\gamma=0.1$, $\underline{B}=-10$, and $\overline{B}=20$.}
			\label{fig:OversellTotalRevenue}
		\end{figure}
		
		Perhaps counterintuitive is the result that a higher overselling penalty can result in greater total revenue by agents, as seen in Figure \ref{fig:OversellTotalRevenue}. However, as the discussion of Figure \ref{fig:OversellMultiQuote} has shown, when overselling penalties are large, agents spend much of the time interval quoting prices as if they are not allowed to oversell. These prices are higher than the situation with weak overselling penalties, because those agents wish to lower prices and accelerate the rate of liquidating their inventory. With a large penalty, when most agents have sold off their positive positions, the market moves into an overselling regime where prices are much higher. This is why we see a deviation of the solid (no overselling) and dash-dotted (highly penalized overselling) curves in Figures \ref{fig:OversellMeanQuote&P_T} (left panel) and \ref{fig:OversellEcoFeature} (all panels) which becomes apparent around $t=T/2$.
		
		Our final investigation is on the effect of overselling on consumers from the perspective of cancellation. If more units of the product are sold than are physically available, then some consumers will end up empty handed. Depending on the industry, the compensation policy, and the consumer herself, she may be indifferent between having the product versus the compensation, but nevertheless the probability of being in this situation is of interest.
		
		For one particular consumer at terminal time, we define event $E$ as the scenario of her product being cancelled due to overselling, and events $E_q$ as the scenario of her buying from agents who oversold $q$ share, $\forall q \in \mathbb{Z}^+$. We assume that at the terminal time, agents who oversold randomly cancel orders due to shortage uniformly across all consumers they transacted with, so the cancellation probability for a particular consumer is given by
		\begin{align*}
			\mathbb{P}(E) = \sum_{q=1}^{-\underline{Q}} \mathbb{P}(E_q)\, \mathbb{P}(E\vert E_q) =\sum_{q=1}^{-\underline{Q}} \frac{P^f_{-q,T}}{1-P^f_{\overline{Q},T}}\,\frac{q}{\overline{Q}+q}\,.
		\end{align*}
		
		\begin{figure}[h!]
			\centering
			\includegraphics[width=0.48\textwidth]{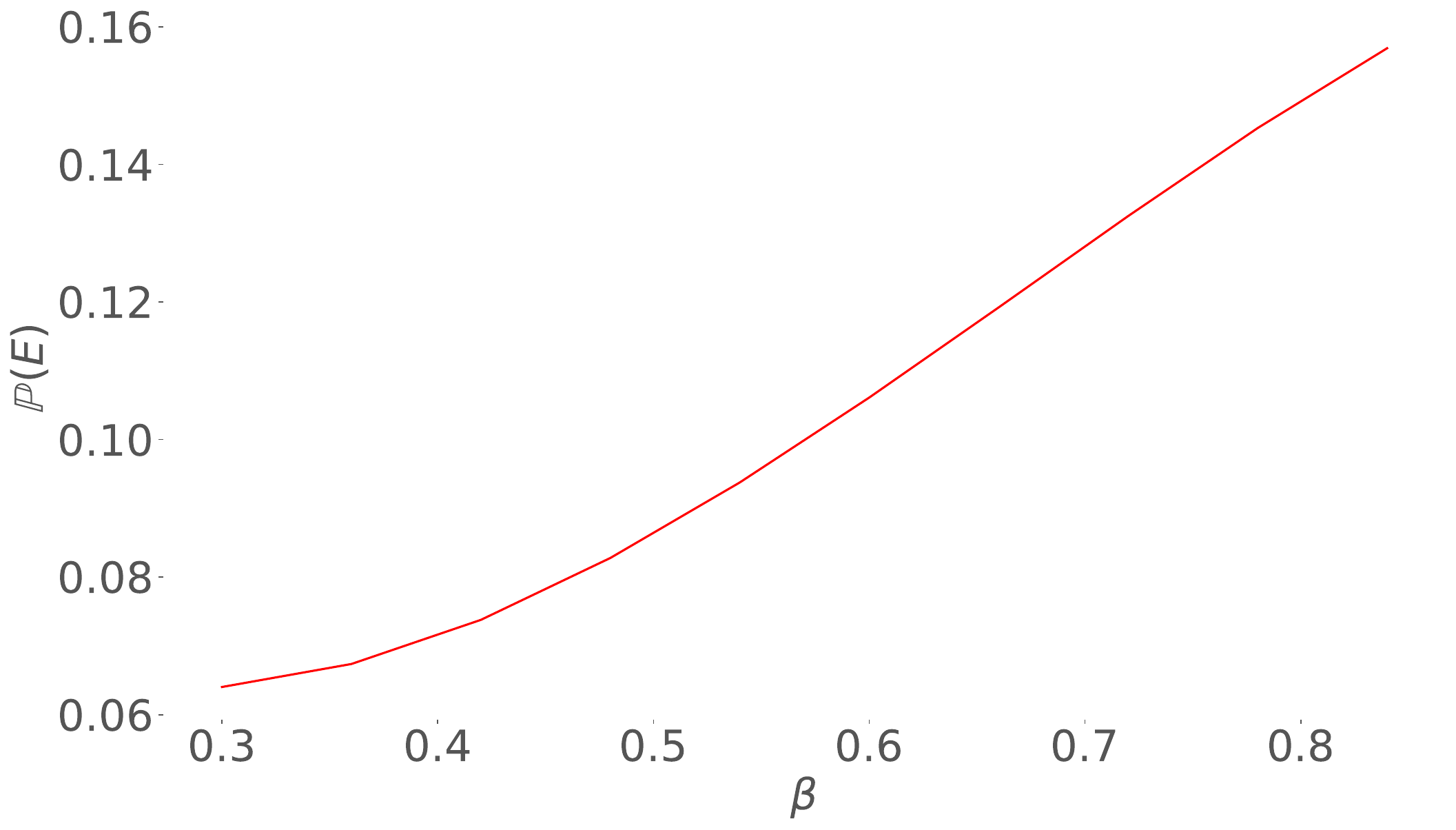}
			\includegraphics[width=0.48\textwidth]{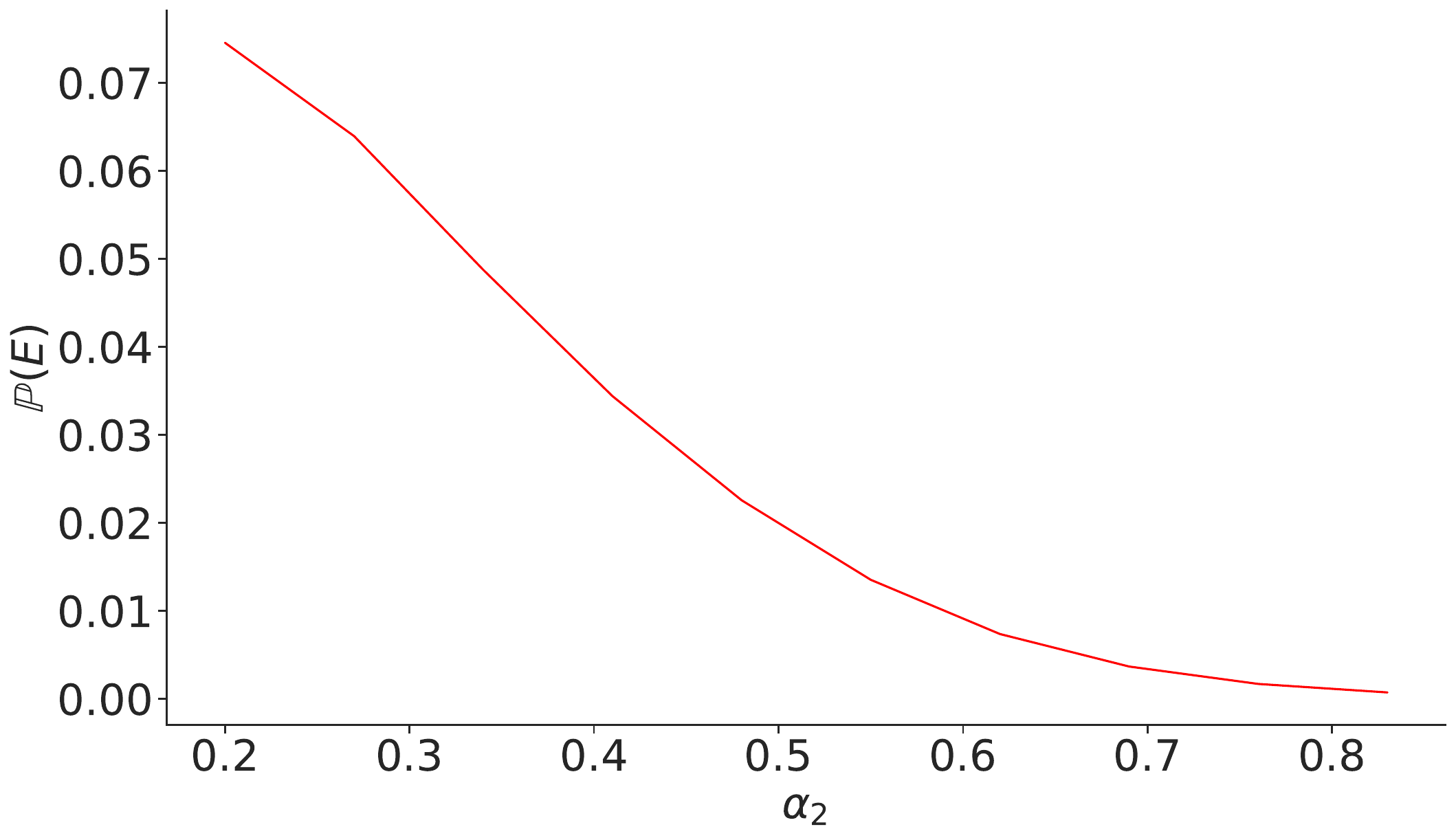}
			\caption{Cancellation probability with different values of competition parameter $\beta$ and different values of negative inventory penalty parameters $\alpha_2$ and $\phi_2$ (we fix the ratio between these parameters by $\phi_2=0.3\,\alpha_2$). In both figures, the other fixed parameters are given by $T=10$, $\overline{Q}=5$, $\underline{Q}=-2$, $\beta = 0.3$, $\kappa=1$, $A=1$, $\alpha_1=0.1$,  $\phi_1=0.03$, $\gamma=0.1$, $\underline{B}=-10$, and $\overline{B}=20$.}
			\label{fig:OversellP(E)}
		\end{figure}
		
		The left of Figure \ref{fig:OversellP(E)} shows how the value of competitiveness parameter $\beta$ effects this cancellation probability. We can see that the cancellation probability is increasing with competitiveness. With higher competitiveness, agents will generally lower their spreads, so the proportion of agents ending at negative inventory levels will increase. 
		In the right panel of Figure \ref{fig:OversellP(E)}, we can observe that higher penalty for overselling leads to a lower probability of cancellation. This is consistent with the result on the right of Figure \ref{fig:OversellMeanQuote&P_T} where we demonstrated that a higher terminal penalty parameter $\alpha_2$ and running inventory parameter $\phi_2$ reduces the proportion of agents who oversold, thus decreases the cancellation probability. For consumers, there is a trade-off between the cancellation probability and the average transaction cost and instantaneous average cost, which is increasing in penalty parameters as per Figure \ref{fig:OversellEcoFeature}. 
		
		\section{Conclusion}\label{sec:conclusion}
		
		We have formulated a model for dynamic inventory pricing which accounts for the effects of competition through a mean-field interaction. First, we introduce a reference model with one agent looking to liquidate a significant number of certain product. We expand our model into the case of infinite players. In our model, the realized sales of each agent is described by a doubly stochastic Poisson process, whose intensity depends on both one's own quoted price and the distribution of quoted prices among all agents. Through the frequency of individual sales, agents compete with each other. This mean-field game system consists of two equations: the dynamic programming equation describing the optimality of representative agent, and the Kolmogorov forward equation governing the dynamics of the distribution of inventory across agents. The two equations are coupled by the consistency condition, which enables us to find an equilibrium numerically by fixed point iterations. For a fixed parameter set, all of our numerical experiments converge to the same equilibrium within tolerance. As expected, competition leads to more sales and lower mean quoted prices of the market as a whole. Interestingly, when market competition's level increases, the optimal quoted prices for all inventory levels do \emph{not} necessarily decrease all the time.
		
		\section*{Appendix}\label{sec:appendix}
		\renewcommand{\thesubsection}{\Alph{subsection}}
		
		\subsection{Proof of Proposition \ref{prop:SolutiontoHJBEquation}}
		\label{proof:SolutiontoHJBEquation}
		\begin{proof}
			Following a similar proof of Proposition \ref{prop:SolutiontoHJBEquationSingle}, we can show that the optimal strategy given by equation \eqref{eqn:OptimalFeedbackControls} is a maximizer.
			
			Equation \eqref{eqn:HJB} is of the form $\partial_t \bold{h} = \bold{F}(\bold{h})$. To show existence and uniqueness of the solution to this equation, by the Picard-Lindelof theorem, the function $\bold{F}$ need to be Lipschitz continuous.
			
			The Lipschitz continuity property of $f$ implies the same for $\bold{F}$, where $f$ satifies
			\begin{align*}
				f(x,y)=\sup_{\delta\ge \underline{B}}A\,\exp{\left\{-(\kappa+\beta)\,\delta+\beta\,\overline{\delta}\right\}}\,(\delta+x-y)\,.
			\end{align*}
			The Lipschitz continuity property of $f$ will be a result of showing that all directional derivatives of $f$ exist and are bounded for all $(x,y)\in \mathbb{R}^2$.
			
			The supremum is attained at $\delta^*$
			in equation \eqref{eqn:OptimalFeedbackControls}. Thus, two separate domains for $f$ must be considered: $ x-y< - \underline{B}+\frac{1}{\kappa+\beta}$ and $x-y\ge-\underline{B}+\frac{1}{\kappa+\beta}$. First, consider $ x-y< -\underline{B}+\frac{1}{\kappa+\beta}$, so that $\delta^*=\frac{1}{\kappa+\beta}+y-x$. Substituting this into the expression for $f$ yields
			\begin{align*}
				f(x,y)=\frac{A}{\kappa+\beta}\,\exp{\left\{(\kappa+\beta)\,(x-y)+\beta\,\overline{\delta}-1\right\}}\,.
			\end{align*}
			Taking partial derivatives of $f$ in this domain gives us
			\begin{align*}
				\partial_x f(x,y)=-\partial_y f(x,y)=A\,\exp{\left\{(\kappa+\beta)\,(x-y)+\beta\,\overline{\delta}-1\right\}}\,,
			\end{align*}
			and this expression is bounded in $\left(0, A\,\exp{\left\{-\underline{B}\,(\kappa+\beta)+\beta\,\overline{\delta}\right\}}\right)$. Thus, $\partial_x f$ and $\partial_y f$ are bounded in this domain, and so directional derivatives exist and are also bounded everywhere in the interior of the domain. On the boundary, directional derivatives exist and are bounded if the direction is towards the interior of the domain.
			
			Now consider $x-y\ge-\underline{B}+\frac{1}{\kappa+\beta}$, which implies $\delta^*=\underline{B}$. The expression of $f(x,y)$ in this domain is
			\begin{align*}
				f(x,y)=A\,\exp{\left\{-\underline{B}\,(\kappa+\beta)+\beta\overline{\delta}\right\}}\,(\underline{B}+x-y)\,.
			\end{align*}
			Partial derivatives of $f$ are given by
			\begin{align*}
				\partial_x f(x,y)=-\partial_y f(x,y)=A\,\exp{\left\{-\underline{B}(\kappa+\beta)+\beta\,\overline{\delta}\right\}}\,.
			\end{align*}
			So similarly to the first domain, directional derivatives exist and are bounded in the interior. On the boundary, they exist and are bounded in the direction towards the interior of the domain. Thus, we have existence and boundedness on the boundary towards every of the two domains. The directional derivative on the boundary is zero when the direction is parallel to the boundary. Existence and boundedness of directional derivatives for all $(x,y)\in\mathbb{R}^2$ allows us to show the Lipschitz continuity condition easily:
			\begin{align*}
				\begin{split}
					\left|f(x_2,y_2)-f(x_1,y_1)\right|&=\left|\int_V\triangledown f(x,y)\cdot d\overset{\to}{r}\right|\le\int_V\left|\triangledown f(x,y)\right|\,ds\\
					&\le\int_VR\,ds=R\,\left|(x_2,y_2)-(x_1,y_1)\right|\,
				\end{split}
			\end{align*}
			where $V$ is the curve which connects $(x_1,y_1)$ to $(x_2, y_2)$ in a straight line and $R$ is a uniform bound on the
			gradient of $f$. This proves that there exists a unique solution $h$ to equation \eqref{eqn:HJB}. \qed
		\end{proof}
		
		\subsection{Proof of Theorem \ref{theo:VerificationTheorem}}
		\label{proof:VerificationTheorem}
		\begin{proof}
			We define a candidate value function $\hat{H}(t,s,x,q;\overline{\delta})=x+q\,s+h_q(t;\overline{\delta})$. 
			From Ito’s lemma we have
			\begin{align*}
				\begin{split}
					\hat{H}(T, S_T,X_{T^-}^{\delta,\overline{\delta}},Q_{T^-}^{\delta,\overline{\delta}};\overline{\delta}) = &\hat{H}(t,s,x,q;\overline{\delta}) + \int_t^T\partial_u h_{Q^{\delta,\overline{\delta}}_{u}}(u)\,du + \int_t^T \sigma \,Q_u^{\delta,\overline{\delta}}\,dW_u \\
					&+ \int^T_t\left(\delta_{u^-}+h_{Q^{\delta,\overline{\delta}}_{u^-}-1}(u)-h_{Q^{\delta,\overline{\delta}}_{u^-}}(u) \right)dN_u^{\lambda(\delta,\overline{\delta})}\,.
				\end{split}
			\end{align*}
			Let $\delta=(\delta_t)_{0\le t\le T}$ be an arbitrary admissible control and let $\epsilon>0$ be arbitrary. Then since $h$ satisfies equation \eqref{eqn:HJB}, the following inequality holds almost surely for every $t$
			\begin{align*}
				\partial_t h_{Q_{t^-}} -\phi\,Q_{t^-}^2+ A\,\exp{\left\{-(\kappa+\beta)\,\delta_{t^-}+\beta\, \overline{\delta}_{t}\right\}}\left[\delta_{t^-}+h_{Q_{t^-}-1}(t;\overline{\delta}_t)-h_{Q_{t^-}}(t;\overline{\delta}_t) \right] < \epsilon\,.
			\end{align*}
			Thus, taking an expectation of $\hat{H}(T, S_T,X_{T^-}^{\delta,\overline{\delta}},Q_{T^-}^{\delta,\overline{\delta}};\overline{\delta})$, we have
			\begin{align*}
				\begin{split}
					&\mathbb{E}_{t,s,x,q} \left[\hat{H}(T, S_T, X_{T^-}^{\delta,\overline{\delta}},Q_{T^-}^{\delta,\overline{\delta}};\overline{\delta})\right]\\
					&=\hat{H}(t,s,x,q;\overline{\delta})
					+\mathbb{E}_{t,s,x,q}\bigg[\int_t^T\partial_u h_{Q^{\delta,\overline{\delta}}_{u}}(u)\,du + \int_t^T \sigma \,Q_u^{\delta,\overline{\delta}}\,dW_u\\
					&\hspace{40mm}+ \int^T_t\left(\delta_{u^-}+h_{Q^{\delta,\overline{\delta}}_{u^-}-1}(u)-h_{Q^{\delta,\overline{\delta}}_{u^-}}(u) \right)dN_u^{\lambda(\delta,\overline{\delta})}\bigg]\\
					&\le \hat{H}(t,s,x,q;\overline{\delta})+ \epsilon\,(T-t)+\mathbb{E}_{t,s,x,q}\left[\phi\int_t^T\left(Q_u^{\delta,\overline{\delta}}\right)^2 du\right]\,,
				\end{split}
			\end{align*}
			where $\mathbb{E}_{t,s,x,q}\left[\int^T_t\left(\delta_{u^-}+h_{Q^{\delta,\overline{\delta}}_{u^-}-1}(u)-h_{Q^{\delta,\overline{\delta}}_{u^-}}(u)\right)dN_u^{\lambda(\delta,\overline{\delta})}\right]$ exists due to the boundedness of $\delta$. Therefore $\hat{H}$ satisfies
			\begin{align*}
				\begin{split}
					\hat{H}(t,s,x,q;\overline{\delta})+ \epsilon\,(T-t) &\ge \mathbb{E}_{t,s,x,q}\left[\hat{H}(T, S_T,X_{T^-}^{\delta,\overline{\delta}},Q_{T^-}^{\delta,\overline{\delta}};\overline{\delta})-\phi\int_t^T\left(Q_u^{\delta,\overline{\delta}}\right)^2 du\right]\\
					&=\mathbb{E}_{t,s,x,q}\left[X_T^{\delta,\overline{\delta}}+Q_T^{\delta,\overline{\delta}}\,\left(S_T-\alpha\, Q_T^{\delta,\overline{\delta}}\right)-\phi \int^T_t\left(Q_u^{\delta,\overline{\delta}}\right)^2 du \right]
				\end{split}
			\end{align*}
			This inequality holds for the arbitrarily chosen control $\delta=(\delta_t)_{0\le t\le T}$, therefore
			\begin{align*}
				\hat{H}(t,s,x,q;\overline{\delta})+ \epsilon\,(T-t) \ge \sup_{(\delta_u)_{t\le u\le T}\in\mathcal{A}}\mathbb{E}_{t,s,x,q}\left[X_T^{\delta,\overline{\delta}}+Q_T^{\delta,\overline{\delta}}\,\left(S_T-\alpha\, Q_T^{\delta,\overline{\delta}}\right)-\phi \int^T_t\left(Q_u^{\delta,\overline{\delta}}\right)^2 du\right]\,,
			\end{align*}
			and letting $\epsilon \to 0$ we finally obtain
			\begin{align*}
				\hat{H}(t,s,x,q;\overline{\delta}) \ge H(t,s,x,q;\overline{\delta})\,.
			\end{align*}
			
			Now let $\delta^*$ be the control process defined as equation \eqref{eqn:OptimalFeedbackControls}, then we have
			\begin{align*}
				\begin{split}
					&\mathbb{E}_{t,s,x,q} [\hat{H}\left(T, S_T, X_{T^-}^{\delta^*,\overline{\delta}},Q_{T^-}^{\delta^*,\overline{\delta}};\overline{\delta})\right]\\
					&=\hat{H}(t,s,x,q;\overline{\delta})
					+\mathbb{E}_{t,s,x,q}\bigg[\int_t^T\partial_u h_{Q^{\delta^*,\overline{\delta}}_{u}}(u)\,du + \int_t^T \sigma \,Q_u^{\delta^*,\overline{\delta}}\,dW_u\\
					&\hspace{40mm}+ \int^T_t\left(\delta^*_{u^-}+h_{Q^{\delta^*,\overline{\delta}}_{u^-}-1}(u)-h_{Q^{\delta^*,\overline{\delta}}_{u^-}}(u) \right)dN_u^{\lambda(\delta^*,\overline{\delta})}\bigg]\\
					&\ge \hat{H}(t,s,x,q;\overline{\delta})+\mathbb{E}_{t,s,x,q}\left[\phi\int_t^T\left(Q_u^{\delta^*,\overline{\delta}}\right)^2 du\right]\,,
				\end{split}
			\end{align*}
			and so $\hat{H}$ satisfies
			\begin{align*}
				\begin{split}
					\hat{H}(t,s,x,q;\overline{\delta})
					&\le \mathbb{E}_{t,s,x,q}\left[\hat{H}(T, S_T,X_{T^-}^{\delta^*,\overline{\delta}},Q_{T^-}^{\delta^*,\overline{\delta}};\overline{\delta})-\phi\int_t^T\left(Q_u^{\delta^*,\overline{\delta}}\right)^2 du\right]\\
					&=\mathbb{E}_{t,s,x,q}\left[X_T^{\delta^*,\overline{\delta}}+Q_T^{\delta^*,\overline{\delta}}\left(S_T-\alpha\, Q_T^{\delta^*,\overline{\delta}}\right)-\phi \int^T_t\left(Q_u^{\delta^*,\overline{\delta}}\right)^2du) \right]\,,
				\end{split}
			\end{align*}
			Therefore,
			\begin{align*}
				\begin{split}
					\hat{H}(t,s,x,q;\overline{\delta}) &\le \sup_{(\delta_u)_{t\le u\le T}\in\mathcal{A}}\mathbb{E}_{t,s,x,q}\left[X_T^{\delta,\overline{\delta}}+Q_T^{\delta,\overline{\delta}}\left(S_T-\alpha\, Q_T^{\delta,\overline{\delta}}\right)-\phi \int^T_t\left(Q_u^{\delta,\overline{\delta}}\right)^2 du) \right]\\
					&=H(t,s,x,q;\overline{\delta})\,.
				\end{split}
			\end{align*}
			Combining the above results we have that
			\begin{align*}
				\hat{H}(t,s,x,q;\overline{\delta})=H(t,s,x,q;\overline{\delta}) \,.
			\end{align*} \qed
		\end{proof}
		
		\subsection{Numerical Stability and Uniqueness of Equilibrium }
		\label{exp:NumericalUniqueness}
		
		We run the algorithm described in Section 3 100 times, each with a randomized initial value of mean spread, and check the converged value of mean spread $\{\overline{\delta}^1, \overline{\delta}^2, \cdots, \overline{\delta}^{100}\}$ obtained from the experiments. Over all 100 simulations, the average values of the mean spread at any time $t \in \{t_j\}_{j=0}^N$, $\frac{1}{100}\sum_{i=1}^{100} \overline{\delta}^i_{t}$, is of order $10^{-2}$ to $1$, while the corresponding standard errors are of order $10^{-16}$ to $10^{-15}$, which is smaller than the tolerance used in the algorithm. Thus we conclude that the final value of $\overline{\delta}$ is always the same, which shows that the algorithm is robust with respect to initial point, and numerically supports our argument that the equilibrium exists and is unique.
		
		\bibliographystyle{chicago}
		\bibliography{References}

	\end{document}